\documentclass[11pt]{article}

\usepackage[letterpaper,margin=1.0in]{geometry}
\usepackage{setspace}
\usepackage{microtype}
\usepackage{booktabs}
\usepackage[ruled]{algorithm2e}

\SetAlFnt{\small}
\SetAlCapFnt{\small}
\SetAlCapNameFnt{\small}
\SetAlCapHSkip{0pt}
\IncMargin{-\parindent}

\usepackage{amsfonts,amsmath,amsthm}
\usepackage{xcolor}
\usepackage[suppress]{color-edits}

\usepackage{hyperref}
\usepackage{url}
\usepackage{subcaption}
\usepackage{natbib}
\usepackage{graphicx}
\usepackage{placeins}

\newtheorem{theorem}{Theorem}[section]
\newtheorem{lemma}[theorem]{Lemma}

\newtheorem{corollary}[theorem]{Corollary}
\newtheorem{definition}[theorem]{Definition}
\newtheorem{proposition}[theorem]{Proposition}

\newtheorem{remark}[theorem]{Remark}

\newcommand{\fitprob}{\pi}
\newcommand{\fitprobest}{\hat{\pi}}
\newcommand{\pihatvec}{\fitprobest^{\mathtt{vec}}} %

\usepackage{xspace}
\newcommand{\term}[1]{\ensuremath{\mathtt{#1}}\xspace}
\newcommand{\xhdr}[2][1]{\vspace{#1mm} \noindent{\bf #2}}

\newcommand{\mD}{\mathcal{D}}

\newcommand{\mM}{\mathcal{M}}
\newcommand{\mP}{\mathcal{P}}
\newcommand{\mX}{\mathcal{X}}
\newcommand{\mY}{\mathcal{Y}}
\newcommand{\oV}{\overline{V}}
\newcommand{\tmM}{\widetilde{\mathcal{M}}}
\newcommand{\E}{\operatornamewithlimits{\ensuremath{\mathbb{E}}}} %
\newcommand{\R}{\mathbb{R}} %
\newcommand{\N}{\mathbb{N}} %
\newcommand{\LDOTS}{\, ,\ \ldots\ ,}
    \newcommand{\rbr}[1]{\left(\,#1\,\right)}
\newcommand{\sbr}[1]{\left[\,#1\,\right]}

\newcommand{\transU}{U^{\mathtt{tr}}}  %
\newcommand{\totC}{C^{\mathtt{tot}}}  %
\newcommand{\priorf}{\mP} %
\newcommand{\priorF}{\mP^{\mathtt{F}}} %
\newcommand{\priorQ}{\mP^{\mathtt{Q}}} %
\newcommand{\distV}{\mD^{\mathtt{V}}}  %
\newcommand{\learnedS}{S^{\mathtt{lnd}}} %
\newcommand{\hist}[1][\mM]{H^{#1}} %
\newcommand{\FOSD}{\term{FOSD}} %
\newcommand{\limP}{\mP^{\mathtt{lim}}} %
\newcommand{\Lmin}{L_{\min}}  %

\newcommand{\pricingnew}[1]
{{#1}}%
\newenvironment{pricingrevision}
{\begingroup}{\endgroup}

\newcommand{\conW}{\pricingnew{\mathrm{CS}}} %
\newcommand{\totW}{\texttt{Wel}}    %
\newcommand{\Rev}{\texttt{Rev}}     %
\newcommand{\CanEqm}{\Phi_{\mathtt{can}}} %

\addauthor{bl}{blue}
\addauthor{nsi}{red}
\addauthor{jmh}{purple}
\addauthor{mmm}{green}
\addauthor{spj}{orange}
\addauthor{as}{brown}
\addauthor{dmr}{pink}
\addauthor{dgg}{teal} 
\addauthor{ni}{red}

\setcitestyle{authoryear,round}

\title{Consumer Search and Social Learning\\in Agentic Markets%
\footnote{Preliminary version in \emph{NeurIPS 2026}. The conference version includes all results presented here except those on endogenous prices (Section~\ref{sec:extension.pricing} and Appendix~\ref{sec:strategic.firms}). This work subsumes an  initial working paper on {\texttt arxiv.org} entitled``Agentic Markets: Equilibrium Effects of Improving Consumer Search".

\vspace{1mm}
\emph{Author affiliations.} BL, MM, AS, DGG, JMF, DMR: Microsoft Research; NI: Microsoft Research and Yale University; SJ: Microsoft AI for Good. 

\vspace{1mm}
\emph{Contact author:} Brendan Lucier, email {\texttt brlucier@microsoft.com}.
}}
\author{Brendan Lucier \and Nicole Immorlica \and Markus Mobius \and Aleksandrs Slivkins \and Daniel G.\ Goldstein \and Jake M.\ Hofman \and Sonia Jaffe \and David M.\ Rothschild}
\date{%
    First version: March 2026\\
    This version: November 2026.
}

\begin{document}

\maketitle
\addtocontents{toc}{\protect\setcounter{tocdepth}{0}}

\begin{abstract}
Motivated by \emph{agentic markets}---where AI tools facilitate search---we propose a model of consumer search and market-wide learning.
A sequence of  consumers engage in costly search to acquire progressively refined signals of product fit prior to purchase.
The market observes acquired signals and post-transaction feedback thereby impacting future searches. We solve for the per-consumer optimal search policy and characterize the  steady-state of the  learning dynamics. Search technologies (agentic or otherwise) that reduce costs at fixed inspection depth improve long-run learning and consumer surplus. In a special case with screen-level feedback, more informative search also improves these outcomes. In contrast, we show through counterexamples and numerical simulations that more informative search can degrade long-run learning and consumer surplus when the market observes only aggregate inspection and transaction feedback, highlighting the importance of feedback design.
\end{abstract}

\vspace{-3mm}
\thispagestyle{empty}

\section{Introduction}
\label{sec:introduction}

Recent advances in AI are enabling a new class of \emph{agentic markets}: markets in which consumers and firms delegate search, interaction, and even transactions to AI agents~\citep{rothschild2025agentic,hadfield2025economy,shahidi2025coasean}.
Consumer-facing systems such as OpenAI's Operator (now ChatGPT agent mode) and Amazon Rufus already allow consumers to offload complex search and comparison tasks, while similar tools on the firm side increasingly automate consumer interaction.
In addition to simply reducing the cost of clicking and browsing, agentic interactions can enable new forms of inquiry.
Consumers can engage in unscripted, adaptive questioning to assess whether a product or service meets their needs.
At the same time, agents can lower the cost of extracting, summarizing, and organizing information generated after transactions, such as consumer experiences and feedback.
Together, these capabilities stand to make search both \emph{cheaper}---by enabling consumers to investigate more firms with the same amount of effort---and \emph{more informative}---by allowing consumers to probe more detailed aspects of fit with any given firm prior to transacting.

These changes should be beneficial for individual transactions.
Cheaper and more informative search helps consumers identify promising options and avoid poor matches.
However, markets are not composed of isolated interactions.
Outcomes emerge from the accumulation of many consumer-firm interactions over time, as information generated by early searches and transactions shapes the beliefs and choices of those who arrive later.
When the technology governing search and interaction changes, it may therefore alter not only individual decisions, but also how information aggregates at the market level, with consequences for learning
and welfare that are not immediately obvious.

\begin{figure}[t]
\centering
\includegraphics[width=.98\textwidth]{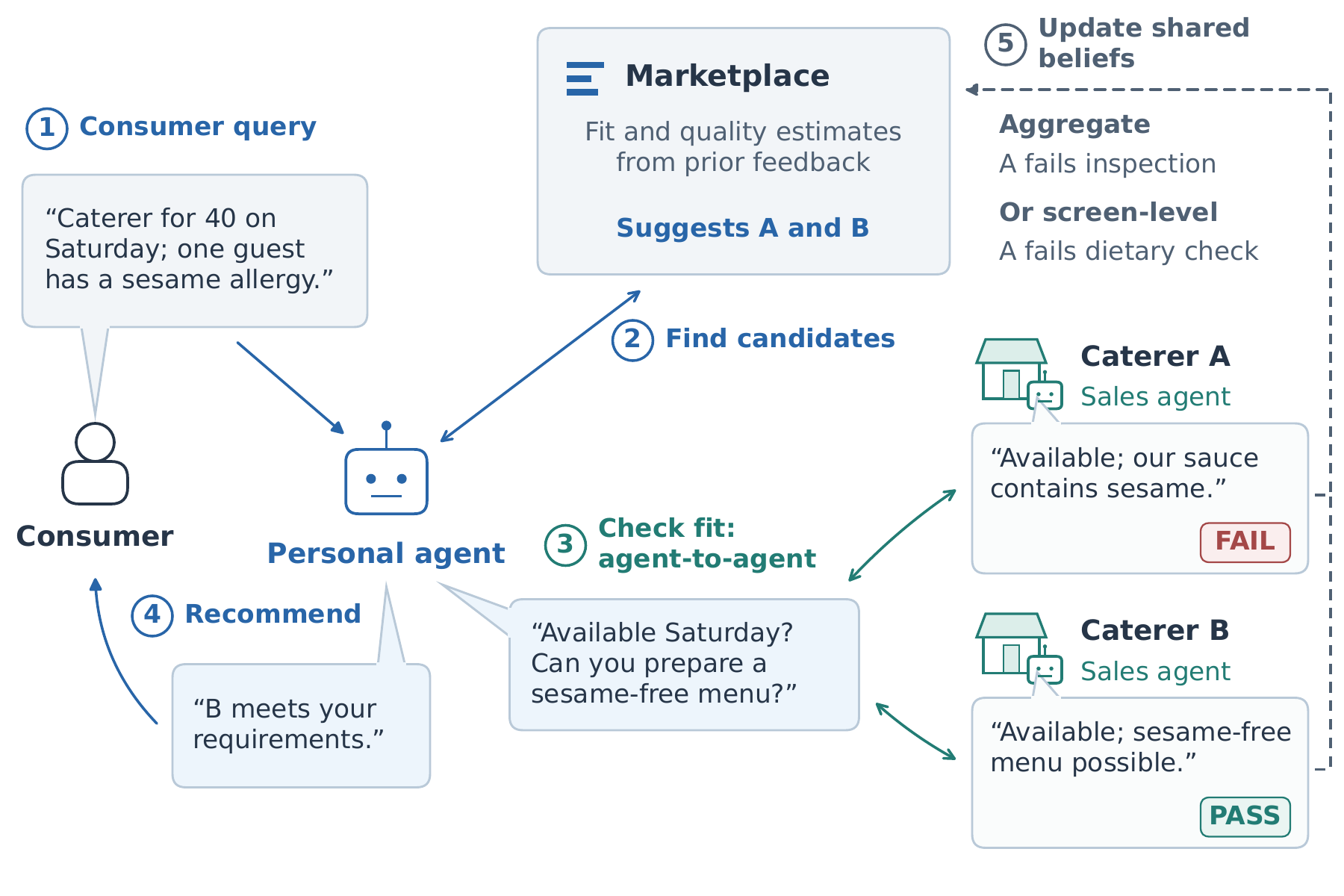}
\caption{{Illustrative agent-assisted search and market learning. A marketplace uses prior feedback to suggest caterers; agent-to-agent conversations check their fit. Inspection outcomes update beliefs for subsequent consumers. Aggregate feedback reports whether inspection passed, whereas screen-level feedback identifies which checks passed or failed. Post-transaction quality feedback is omitted.}}
\label{fig:agentic-market}
\end{figure}

As an example, consider a consumer hiring a caterer for an event { (see Figure~\ref{fig:agentic-market})}. The value of a transaction depends on caterer features.
Some features, such as cuisine, dietary accommodations, or calendar availability, determine \emph{fit} with the consumer's needs and in principle can be assessed pre-transaction.
But, in practice, today's interfaces limit how many features can be efficiently checked in advance, often pushing more detailed assessment into costly back-and-forth conversations or even into the post-transaction phase.
{An AI agent could ask the caterer, or its agent, whether a particular menu can accommodate a guest's dietary restriction not addressed in the listing.
Automating a check that would otherwise require a phone call reduces search costs; resolving a mismatch that would otherwise emerge only when the food is served makes search more informative.}
Other
features
are harder to verify in advance and tend to be revealed or evaluated only after a transaction takes place; we refer to these as \emph{quality} features.
This might include
taste, execution, and reliability, which are typically only learned through experience.

Consumers leverage information generated by others to help estimate quality and the probability of a fit prior to a more costly investigation.
Early interactions produce signals through
feedback
that guide later consumers' choices.
Over time, some caterers become effectively
well-understood by the market, while others may be effectively \emph{lost} if pessimistic beliefs discourage further exploration.
{Returning to our example, if the platform observes the answer to each individual query, it can distinguish, e.g., a dietary mismatch from a lack of availability. If it sees only whether the caterer passed inspection, these failures are pooled.
Thus, the effects on future consumers depend not only on what the agent learns, but also on what the platform observes and shares.}

\xhdr{Our contributions} are two-fold: 
a tractable model that we hope can be the basis for future research, %
and {an analysis of how search costs, pre-transaction information, and feedback observability jointly shape market learning and consumer surplus.
These mechanisms are motivated by conversational agents but are not specific to AI; our results also apply to other search technologies with the same information structure.}

\xhdr{Model.} We develop a tractable model of consumer search in which consumers inspect products to determine fit,
generate feedback (through both inspections and transactions), and the market aggregates this information over time. We assume a satisficing model of utility, where each consumer has her own fixed private value for any product that satisfies her needs~\citep{simon1956rational}.
Consumer inspection of a product is modeled as a sequence of costly \emph{screens}: pass/fail tests.
The search technology determines the maximum depth of inspection -- which screens can be revealed pre-transaction -- and the cost of resolving a particular screen. {A consumer may stop searching, abandon a firm, or switch between partially inspected firms and return to them later. A consumer that stops searching can transact with any firm whose inspectable screens have been completed.}

\xhdr{Optimal search and learning.}
Consumers enter the market sequentially and (prior to search)
update their beliefs based on past consumers' feedback.
Despite the resulting complex dynamics, we show that this model converges to a distribution over steady-states
and we characterize the long-run learning outcomes.
{To analyze these dynamics, we characterize optimal search using an index-based policy, building on the Pandora's box literature~\citep{weitzman1978optimal,kleinberg2016descending}.}
This policy, which is efficiently computable, induces an endogenous search order over products and screens that is responsive to consumers' beliefs and values.

\xhdr{Comparative statics.}
We prove that lowering search costs, {at fixed inspection depth}, robustly improves market-level learning and consumer surplus in the long-run.
{We also study more informative search, which shifts information about fit into the pre-purchase stage without changing the underlying probability that a firm's product meets the consumer's need.
When all pre-purchase screens can be revealed for a single inspection cost, this improves long-run learning and consumer surplus if the market observes each individual screen outcome.
By contrast, if the market observes only whether a firm's inspection was successful, more informative search can reduce both.
Pooling failures across screens changes the evidence available to future consumers, potentially causing useful firms to be prematurely abandoned.}

\xhdr{Simulations.}
{The analytical results above concern long-run market outcomes; our simulations examine finite-horizon dynamics.
We find that cheaper and more informative search improve consumer utility, but 
learning failures (i.e., useful firms abandoned prematurely)
need not be monotone 
as search becomes more informative.
With transcripts, this observed inversion reverses in longer runs; without transcripts, 
the non-monotonicity persists
over the extended horizon and across market variants.}

\begin{pricingrevision}
\xhdr{Endogenous prices.}
Section~\ref{sec:extension.pricing} extends the single-cost inspection setting to endogenous prices. We characterize learning under a demand-responsiveness condition on pricing rules and construct symmetric steady-state pricing equilibria. For firms of a common type, cheaper search weakly increases consumer surplus and total welfare at the canonical equilibrium. More informative search, however, can reduce equilibrium consumer surplus by weakening competition.
\end{pricingrevision}

\xhdr{Implications and roadmap.}
These results identify risks and opportunities from search improvements motivated by agentic interaction.
Long-run learning depends critically on how information generated during search and transactions is surfaced and aggregated.
These findings highlight the importance of market and platform design choices in determining whether agentic tools ultimately lead to more efficient and effective markets.
\pricingnew{We introduce our baseline model in Section~\ref{sec:model} and derive the optimal consumer search policy in Section~\ref{sec:search}.
Section~\ref{sec:convergence} characterizes the learning dynamics and their limiting outcomes.
Section~\ref{sec:compare} contains our comparative statics: how search costs and informativeness affect learning and consumer surplus under different feedback regimes.
In Section~\ref{sec:expts-body} we describe and report on numerical simulations that illustrate these effects and their finite-horizon dynamics.
Section~\ref{sec:extensions} reports on extensions to endogenous inspection depth and prices.
Section~\ref{sec:discussion} concludes with a discussion of further directions.}

\xhdr{Related Work.}
{Our motivation draws on work on AI-mediated search.
\citet{shahidi2025coasean} discuss reductions in search, communication, and preference-elicitation costs, while \citet{kumar2026conversational} model the matching benefits and communication costs of follow-up questions.
\citet{yan2026shopping} document pre-purchase use of a platform AI assistant for exploratory, hard-to-keyword tasks.
In retail field experiments, \citet{fang2026generative} find conversion gains consistent with reduced search and information frictions.
We study how such changes interact with the feedback available for social learning.}

Our model of consumer search is a variant of the \emph{Pandora's box problem} \citep{weitzman1978optimal}.%
\footnote{Weitzman's model is itself a (non-obvious) special case of Markovian formulation of multi-armed bandits~\citep{gittins1979bandit,whittle80}.}
Many researchers have investigated variants of Weitzman's highly tractable model; see~\citet{beyhaghi2024recent} for a recent survey. 
Particularly relevant are sequential partial inspections where the contents of box can be partially revealed 
\citep{aouad2025pandora,gibbard2022model,kleinberg2016descending}
and correlated Pandora's Box models where the consumer can learn about the market as they search~\citep{bikhchandani1996optimal,adam2001learning}.
In contrast, learning in our model occurs at the platform level; each consumer learns from previous consumers rather than 
from their own search. %

Our setting features consumers that learn from past consumers. It thus
connects to 
social learning with myopic short-lived agents, most notably to the models where both past actions and outcomes
are observed
\citep[e.g.,][]{kannan2018smoothed,AcemogluMMO19,BSL-myopic23}.
In essence, the agents collectively face explore-exploit tradeoff, while individually favoring exploitation. This literature tends to focus on simple models such as multi-armed bandits and study asymptotic ``success" vs ``failure" of the learning dynamics.
It does not address more complex models such as ours.%
\footnote{One exception is~\citet{immorlica2019diversity} who also
study a market where consumers perform optimal search given posterior beliefs shaped by prior consumers' inspections. Their utility and information structures are different from ours, motivated by a focus on diversity rather than search efficacy. Moreover, our model does not fully reveal firm's quality upon inspection.}

\pricingnew{Related work designs platform interventions to incentivize exploration through information disclosure \citep{Kremer-JPE14} or payments \citep{Frazier-ec14}. In contrast, we take the search and feedback system as given: our platform simply collects and disseminates consumer feedback.}

Other social learning models feature agents who learn from \emph{privately observed signals} as well as the actions taken by past agents \citep{Golub-survey16}.
\citet{banerjee1992simple,bikhchandani1992theory} introduce this model and show that
herds result
when signals and/or action spaces are coarse.
Subsequent work described a variety of conditions necessary for such results, including
the richness of the action sets \citep{ali2018role,lee1993convergence}  and the  signal spaces \citep{smith2000pathological}, and costly signal acquisition  \citep{ali2018herding,burguet2000social,hendricks2012observational,
mueller2016social}.
Relative to these literatures, our action and signal spaces are different -- being motivated by sequential consumer search -- and we focus on quantifying and comparing what is learned in the long-run limit, rather than determining necessary/sufficient conditions for asymptotic success.

\pricingnew{Section~\ref{sec:extension.pricing} relates our model to costly search and price competition \citep{diamond1971model}, and to pricing with restricted consumer access or attachment to firms \citep{babaioff2013bertrand,shilony1977mixed}. We adapt the transformed-market approach of~\citet{choi2018consumer} to characterize steady-state pricing conditional on the learned set.}

\section{Model and Preliminaries}
\label{sec:model}

We model a two-sided market with sequential consumer arrivals,
individual consumer search (as a variant
of the Pandora's box problem), and
bandit-like feedback across consumers.%
\footnote{The technical connections to Weitzman's Pandora's box and multi-armed bandits are spelled out in Appendix~\ref{app:connection}.}

\xhdr{Consumers and firms.}
We have an infinite stream of \emph{consumers} $i\in\N$
and a fixed pool of $m$  \emph{firms} $j\in[m]$.
The consumer has a need and receives value $V_i \geq 0$
if this need is met. Each consumer $i$ can transact with at most one firm $j$ at a fixed and known price $P_j \geq 0$, yielding \emph{transaction utility}
$\transU_i = V_i-P_j$ if the need is met by firm $j$ or $\transU_i = -P_j$ if it is not.
{Value $V_i$ is private to consumer $i$ and drawn i.i.d.\
from a known distribution $\distV$
over $[0,\oV]$, for some finite $\oV$.}

Whether consumer $i$'s need is met by firm $j$ is determined by a \emph{screen sequence}
$\vec{f}_{ij} = \rbr{f_{ijk},\,k\in[r]}$
of $r$ binary random variables, called 
\emph{screens}. We think of $\vec{f}_{ij}$
as the outcomes of progressively deeper consumer inspection. Each screen is pass/fail (resp., value 1/0). E.g., $f_{ij1}$ might indicate whether the firm sells the right product; $f_{ij2}$ whether basic product features listed on the webpage are suitable; and so on.
The need is met ex post
if and only if all screens pass.

The screen sequence $\vec{f}_{ij}$ is drawn independently for each consumer $i$ and firm $j$, from some distribution that depends on the firm but is the same for all consumers. The screens can be correlated with each other.
We write
$\fitprob_{jk}:= \Pr\sbr{ f_{ijk}=1 \mid f_{i,j,1} = \cdots = f_{i,j,k-1}=1}$ for the conditional probability that the $k$-th screen passes for firm $j$ (and a given consumer), given that all previous screens have passed.
\footnote{{$\pi_{jk}$ can be interpreted as the fraction of consumers for which this conditional event holds.}}
We posit $\fitprob_{jk}\in(0,1)$, i.e., the screen outcome is not deterministic.

Thus, firm $j$ is characterized by the parameters $\fitprob_{jk}$, $k\in[r]$. The consumers do not know these parameters, but they know %
prior distributions $\priorf_{jk}$
from which they are independently drawn.
The common prior over the firm parameters is given by the tuple
    $\mP := (\; \priorf_{jk}:\; j\in[m],k\in[r]\;)$.

Each screen $k \in [r]$  has an associated inspection cost $c_k \in [0,\infty]$. Costs are public knowledge and uniform across firms. %
 Costs can be infinite: $c_k = \infty$ means that screen $k$ cannot be inspected, only experienced upon a transaction. %
{The search technology fixes an \emph{inspection depth} $L\in[r]$: the first $L$ screens have finite inspection costs, and any remaining screens can only be experienced upon a transaction.}\footnote{\pricingnew{Many of our results extend to a setting where each consumer strategically chooses their inspection depth; see Section~\ref{sec:extension.depth}.}}

\begin{definition}
A \emph{market} $\mM$ is defined by
firm parameters $\rbr{\fitprob_{jk}:\; j\in[m],k\in[r]}$,
prices $P_j$,
value distribution $\distV$,
common prior $\mP$,
and search costs $\vec{c} = \rbr{c_k:\, k\in[r]}$.
\end{definition}

\begin{remark}
\label{rem:pointwise}
All of our results will hold \textbf{pointwise for each market}, as defined above. That is, for each feasible realization of the firm parameters $(\fitprob_{jk})$.
\end{remark}

\xhdr{Market Timing and Consumer Search.}
The market evolves according to the following protocol:

\begin{itemize}

\item[(P1)] Time proceeds in rounds $i\in\N$. Consumer $i$ arrives in round $i$ and leaves when this round is over. All firms are present in all rounds.

\item[(P2)] Within round $i$, the consumer engages in sequential search.
At each step, she chooses some firm to \emph{inspect}, depending on the outcome of the previous steps. Each firm $j$ can be inspected at most $L$ times: the $k$-th time, the consumer observes screen $f_{ijk}$ and pays cost $c_k \geq 0$.
{She may interleave inspections of different firms, abandon a firm, or return to a partially inspected firm.
She can stop searching at any point and optionally transact with any ``fully inspected'' firm, meaning one whose first $L$ screens have been revealed.}

\item[(P3)] After each round $i$, some feedback is collected:
(i) inspection feedback: the sequence of inspected screens $f_{ijk}$ and their values; (ii) transaction feedback: firm $j$ transacted with (if any), whether consumer's need was satisfied, or why not: first screen $k$ that failed.  The current history (the feedback over all rounds so far) is then revealed to the next consumer.

\end{itemize}

The (total) utility of consumer $i$, $U_i$, is her transaction utility $\transU_i$ minus her total search cost, {$\totC_i$}.
Each consumer is Bayesian-rational and seeks to maximize her expected utility $\E[U_i]$.
 We assume consumers break ties in favor of not inspecting a firm {and not transacting}.

{\xhdr{Preliminaries.}
A \emph{search policy} is (any) policy used by a consumer to implement the search in (P2).

All expectations are by default over prior $\mP$ and value distribution $\distV$. If there's a possibility of confusion, we'll explicitly condition on history, consumer values, or firm parameters.}

Fixing market $\mM$, a \emph{history} is a sequence $\hist_n=(h_1, h_2 \LDOTS h_n)$ for some $n\leq \infty$, where $h_i$ is the feedback collected in round $i$, {as per (P3)}. Thus, $\hist_n$ is a random variable determined by the $V_i$'s and $f_{ijk}$'s.
Given a history realization $\hist_{i-1}=H$, the posterior beliefs of customer $i$
are denoted
    $\priorf_{jk}(H) := P(\pi_{jk}\mid \hist_{i-1}=H)$.
The respective \emph{posterior means} are
    $\fitprobest_{jk}(H) = \E\sbr{\fitprob_{jk}\mid \hist_{i-1}=H}$.
As a shorthand, we define
    $\priorf_{ijk} := \priorf_{jk}(\hist_{i-1})$
        for the beliefs, and
    $\fitprobest_{ijk} := \fitprobest_{jk}(\hist_{i-1})$
        for the posterior means, as random variables over all randomness in $\hist$.

\xhdr{Model Interpretation.} In practice, a market platform collects and aggregates various consumer feedback signals (such as the ones discussed above)
and shares them with future consumers.

Inspecting a firm can be interpreted as a purchase funnel where a consumer passes through increasingly selective screening before choosing to make a purchase. Inspection may happen in various modes, from a
phone call to browsing a webpage to having an AI agent browse a webpage to having an AI agent interact with a firm's AI agent. Different screens may correspond to these different modes and/or to different 
costs for engaging in these modes (e.g., AI tokens or human time spent).
{We treat the stages as ordered: identifying a product category, for example, precedes checking details on a product page.
Cheaper information gathering within an inspectable stage is modeled by a lower $c_k$ when the screen and feedback remain unchanged.}

Inspection outcomes in our model
are binary, random,
and fully revealed to the platform.
{A failed screen rules out a satisfactory match, while any firm meeting the consumer's need yields the same gross value $V_i$.
Such need-based utility specification is common in the advertising literature (see, e.g.,~\citet{athey2011position}). 
Consumers differ in value and realized screen outcomes, but the probabilities $\fitprob_{jk}$ are common across consumers; 
these can be viewed as the likelihood that firm $j$ satisfies the needs of a randomly-chosen consumer.}\footnote{{In particular, as discussed in Section~\ref{sec:introduction}, the first $L$ screens may correspond to (inspectable) features that determine the \emph{fit} of a firm, 
whereas the latter $r-L$ screens may reflect the inherent \emph{quality}.}}

{Together, this utility specification, binary screens, and the requirement in (P2) to complete the first $L$ screens before transacting support the persistent policy in Section~\ref{sec:search}.
Obligatory inspection is also common in sequential search~\citep{beyhaghi2024recent}.}

\xhdr{Long-Term Outcomes.}
We are interested in long-term outcomes of a given market, in terms of learning and consumer utility. Specifically, we focus on \emph{limit beliefs}
    $\lim_{i\to\infty}(\;\priorf_{ijk}\;)$
and \emph{limit consumer utility}
    $\lim_{i\to\infty} \E\sbr{ U_i \mid \hist_{i-1}}$.
        (We'll establish in Section~\ref{sec:convergence} that these limits exist.)
We ask whether a firm's parameters are learned in the following asymptotic sense:

\begin{definition}
Firm $j$ is \emph{learned}, given history realization $\hist_\infty=H$,
if for all $k\in[r]$ beliefs $\priorf_{jk}(H_{i-1})$ weakly converge to $\pi_{jk}$ as $i\to\infty$,
where $H_{i-1}$ is the prefix of $H$ visible to consumer $i$.
\end{definition}

{The set of all learned firms will be a key object in our results, with direct utility implications.}

\section{Optimal Consumer Search Policy}
\label{sec:search}

In this section we describe each consumer's optimal search policy given their beliefs.
{The sequential screens represent a search funnel and distinguish the cost of inspection from the information revealed before purchase.
Our analysis builds on the index-rule approach to search~\citep{weitzman1978optimal,kleinberg2016descending} to characterize the feedback that drives market learning.}

\begin{algorithm}[t]
	\SetAlgoNoLine\SetAlgoNoEnd
     For each firm $j$, define the {{\bf\em index}} as
$\sigma_{ij}
    := V_i \cdot \prod_{\ell > L}\fitprobest_{j\ell} - P_j
        - \sum_{\ell \leq L}
                (\;c_\ell \,\big/\,
               \prod_{k\in [\ell,L]}\fitprobest_{jk}\;)$
\;
	\For{each firm $j$ with $\sigma_{ij}>0$, in order of decreasing index $\sigma_{ij}$,}
        {
        \For{each $\ell = 1,2 \LDOTS L$}
        {
        Pay $c_\ell$ to inspect screen $f_{ij\ell}$\;
        If $f_{ij\ell}=0$, break and continue to the next firm.
        }
        Transact with firm $j$ and leave the market \quad \emph{// We must have $f_{ij\ell}=1$ for all $\ell\in [L]$}
		}
      Leave the market without transacting \quad \emph{// No firms with positive index were a fit}
	\caption{Optimal search policy for customer $i$
}
	\label{alg:index}
\end{algorithm}

Our search policy is Algorithm~\ref{alg:index}, which relies on
the current posterior estimates
    $\fitprobest_{jk} = \fitprobest_{jk}(\hist_{i-1})$.
The policy orders the firms $j$ according to a certain score $\sigma_{ij}$, traditionally called \emph{index}.
For each firm in this order, the consumer investigates until either observing a failure, in which case the firm is permanently abandoned, or until $L$ successful searches, in which case the consumer transacts.

\begin{theorem}
\label{thm:optimal.search}
Fix customer $i$ with value $V_i$ and history $\hist_{i-1}$.  Algorithm~\ref{alg:index} %
maximizes expected utility
$\E[U_i \mid V_i, \hist_{i-1}]$ over all
search policies (and respects our tie-breaking).
\end{theorem}

{Although the model permits interleaving inspections across firms, an optimal policy can be chosen to inspect each firm in contiguous steps.
The proof first establishes the Pandora reduction within this class of \emph{persistent} policies, then shows that restricting attention to this class is without loss; details appear in Appendix~\ref{app:optimal.search}.
This characterization lets us analyze the feedback generated by consumer search and its consequences for market learning.}

{Under this optimal policy, the feedback generated by inspecting a firm has the following form.}

\begin{corollary}
\label{cor:revealed.signals}
In each round $i$, if a given firm $j$ is inspected at all, then all screens \bledit{$k\in[L]$} are inspected until the first one that fails (if any), and no other screens are inspected.
\end{corollary}

\section{Feedback and Learning}
\label{sec:convergence}

This section studies the evolution of market history and consumer beliefs. 
We show that the market converges to a steady-state where some firms are learned and the rest are never again searched. We characterize the distribution over steady-states,  including the probability of being learned.
  Finally, we show that expected consumer utility converges 
  to the utility of a consumer with steady-state beliefs.

To see how posterior beliefs evolve over time, consider their impact on the index in Algorithm~\ref{alg:index}. {Let
     $\pihatvec_{ij} := \rbr{\fitprobest_{ijk}:\,k\in[r]}$
be the vector of posterior mean estimates for consumer $i$ and firm $j$.} Let $\sigma_j(\pihatvec_{ij}, V)$ be the index assigned to firm $j$ by a consumer with value $V$ and posterior mean estimates $\pihatvec_{ij}$.
After some round $i$,
the estimates $\pihatvec_{ij}$
may fall low enough that $\sigma_{j}\rbr{\pihatvec_{ij}, \oV} \leq 0$. 
Then even the highest-valued consumers have non-positive index for firm $j$, so the firm would never be inspected after round $i$ and its associated posterior distribution would never update. It turns out that this is the only barrier to learning: we show that as long as $\sigma_{j}\rbr{\pihatvec_{ij}, \oV} > 0$ for all consumers $i$,
the posterior distributions for firm $j$ will converge in expectation to  ground truth.

\begin{proposition}\label{prop:market.learning}
Fix market $\mM$ and firm $j$. Almost surely over randomness in history $\hist$, the following happens. Beliefs $\priorf_{ijk}$ weakly converge to some $\limP_{jk}\in \Delta(\R)$ as $i \to \infty$, for all $k\in[r]$. Moreover, either firm $j$ is learned, in which case $\priorf_{ijk} \Rightarrow \fitprob_{jk}$ for all $k\in[r]$, or else
$\sigma_{j}\rbr{\pihatvec_{ij}, \oV} \leq 0$
for some consumer $i$ (and every consumer afterwards), in which case we say that firm $j$ is \emph{lost}.
\end{proposition}

Thus, posterior beliefs converge to \emph{limit beliefs} 
    $\limP := \rbr{\limP_{jk}:\,j\in[m],k\in{r}}$.
In every execution of the market dynamics, each firm is either learned or lost. Learned firms are inspected and transacted with infinitely often, whereas lost firms are never inspected after some finite time.

Proposition~\ref{prop:market.learning} follows by noting that, due to our assumption that screen probabilities lie strictly between $0$ and $1$, each firm that is not lost will be inspected and transacted with infinitely often. In this case the market observes an infinite sequence of independent Bernoulli draws for each screen, and hence the associated posterior beliefs weakly converge to their respective true means.

Proposition~\ref{prop:market.learning} and Corollary~\ref{cor:revealed.signals} imply that whether firm $j$ is 
lost
depends only on the progression of its posterior mean estimates, which
are driven solely by inspection outcomes and feedback for firm $j$. This can be used to show that each firm has an independent probability of being lost.

\begin{proposition}
\label{prop:lost}
Fix any market $\mM$. Over all randomness in $\hist$, each firm $j$ is lost independently with a probability $\omega_j \in [0,1]$ that depends only on $(\fitprob_{jk})_k$, $\vec{c}$, $P_j$ and $\overline{V}$.  Probability $\omega_j$ is weakly increasing in each {(finite)} coordinate of $c$ and in $P_j$ and weakly decreasing in each $\fitprob_{jk}$ and in $\overline{V}$.
\end{proposition}

Write $\learnedS(\mM)$ for the set of firms learned in market $\mM$; call it the \emph{learned set}. Note that it is a random variable, with randomness coming from the history. Since the precise posterior beliefs for lost firms are not payoff-relevant to any consumer or firm (as they are never searched nor transacted with), $\learnedS(\mM)$
effectively determines the steady-state to which the market converges.

\xhdr{Utility implications.}
Given a set of firms $S \subseteq [m]$, consider a
hypothetical consumer who (a) knows the features of all firms in $S$ %
and (b) is not allowed to transact with any firm outside $S$.  We write $U^*(S)$ to denote the expected utility of following the optimal search policy on firm set $S$.
This is a static comparator: it is defined for a fixed subset of firms, independently of the learning dynamics.

We show that expected consumer utility converges to this comparator, evaluated at $S = \learnedS(\mM)$.
More precisely, we consider $\E\sbr{U_i \mid \hist_{i-1}}$, 
conditioning on history up to round $i$.

\begin{lemma}\label{lem:util.converge}
$\E\sbr{U_i \mid \hist_{i-1}} \to U^*(\learnedS(\mM))$ {as $i\to\infty$},
almost surely over random history $\hist$.
\end{lemma}

Fixing history $\hist$, {the comparator} $U^*(\learnedS(\mM))$ is a fixed real value, as is each $\E\sbr{U_i {\mid \hist_{i-1}}}$.
Thus, Lemma~\ref{lem:util.converge} says expected consumer utilities almost surely converge in our market for a random history. Moreover, they converge to the utility obtained by a hypothetical consumer who holds the limit posterior beliefs from Proposition~\ref{prop:market.learning} (which are determined by the learned set $\learnedS(\mM)$).

A complication in the proof of Lemma~\ref{lem:util.converge} is that the expected (total) utility of a consumer is discontinuous in her beliefs, such as when updated beliefs lead to a change in firm inspection order.
To prove Lemma~\ref{lem:util.converge}, 
we use the convergence of consumer beliefs to the ground truth for learned firms to discipline the evolution of consumer utility and establish convergence in the limit.

Motivated by Lemma~\ref{lem:util.converge}, we write
$U(\mM) := U^*(\learnedS(\mM))$ and call it
{\emph{limit consumer utility}}
in market $\mM$. Like $\learnedS(\mM)$, $U(\mM)$ is a random variable with randomness coming from the history.

\section{The Impact of Improving Consumer Search}
\label{sec:compare}

We show that reducing the vector $\vec{c}$ of search costs, coordinate-wise, can only increase the learned set and the long-run consumer utility. While intuitive, this is not obvious: fixing a history, cheaper search clearly improves the current consumer's outcomes but it can also alter the dynamics of learning and impact future consumers. In principle this could impede long-run utility, but our proof rules this out.

On a technical level, our comparison considers two markets, with a specific difference between them, and holds almost surely over a \emph{coupling} between their histories.%
\footnote{\label{fn:coupling}
A \emph{coupling} between random variables $X,Y$ (taking values in $\mX,\mY$) is a joint distribution $\mD$ over $\mX\times\mY$ with correct marginals. More formally: for a random pair $(X',Y')\sim \mD$, the marginal distribution of $X'$ (resp., $Y'$) coincides with that of $X$ (resp., $Y$). Note that the original random variables $X,Y$ do not need to lie in the same probability space.}
When and if the comparison is between real-valued metrics (stating that something is larger in one market than in the other), the same comparison holds in the sense of first-order stochastic dominance (FOSD).%
\footnote{First-order stochastic dominance (FOSD) is defined as follows. Consider real-valued random variables $X,X'$ (not necessarily lying in the same probability space). We say $X\geq_\FOSD X'$ if it holds that $\Pr[X>x] \geq \Pr[X'>x]$ for every $x$. The relation is strict, denoted $X>_\FOSD X'$, if additionally $\Pr[X>x] > \Pr[X'>x]$ for some $x$.}

\begin{theorem}\label{thm:cheaper.search.general}
Suppose markets $\mM,\tmM$ are identical except $\tmM$ has reduced finite search costs: $\tilde{c}_k \leq c_k$ $\forall k\leq L$. 
Then for some coupling between their histories, the following holds almost surely:
\begin{itemize}
    \item $\learnedS(\tmM) \supseteq\learnedS(\mM)$.  That is, more firms are learned under history $\hist[\mM]$ than under $\hist[\tmM]$, for any joint realization of the two histories.
    \item $U(\tmM) \geq U(\mM)$.  That is, the {limit consumer utility}
    is weakly higher under history $\hist[\tmM]$ than under $\hist[\mM]$, for any joint realization of the two histories.
\end{itemize}
\end{theorem}

{Only finite costs change in Theorem~\ref{thm:cheaper.search.general}, so inspection depth and feedback are fixed.}

The proof 
proceeds by coupling histories of $\mM$ and $\tmM$ on (a) the sequence of consumer values and (b) the outcome when screen $k$
of firm $j$ is revealed for the $t$'th time, for each $k$, $j$, and $t$.\footnote{Note that even though consumer values are not explicitly represented in a history, they influence the history through the order in which consumers inspect firms. Fixing both the consumer values and the random outcomes of screen inspections for each firm
fully specifies each history, defining the coupling.} Even under this coupling the resulting histories can generate very different posterior beliefs at any given round $i$, since different costs can lead to different firms being inspected in each round. {Here we use the persistent structure of Algorithm~\ref{alg:index}}: while firms might be inspected in different rounds, Corollary~\ref{cor:revealed.signals} implies similar information is revealed in markets $\mM$ and $\tmM$ on the $t$'th round that a firm is inspected.
Lemma~\ref{lem:indep.signals} then implies that the posterior estimates for each firm $j$ evolve similarly in markets $\mM$ and $\tmM$ \emph{when evaluated only at rounds where those estimates change}. Search costs do not influence this progression, but they do change the condition for a  firm being lost. This allows us to couple the set of firms lost across the two histories, which informs long-run consumer utility.

\subsection{Search Informativeness and the Importance of Transcripts}
\label{sec:search.informativeness}

{The aggregate-feedback model below is a two-screen special case, so Theorem~\ref{thm:cheaper.search.general} also guarantees the benefits of cheaper search without transcripts.
We now vary search informativeness rather than cost.
For the special case defined below, screen-level feedback ensures weakly better learning and limit utility (Theorem~\ref{thm:cheaper.search.inform}).
With aggregate feedback, both can instead deteriorate (Proposition~\ref{prop:inform.bad}).}
{To this end,} we consider a simple special case of our model where there are two screens and each firm can only be inspected once.

\begin{definition}
\label{defn:fit-quality}
A \emph{fit-quality market} is a special case with $r=2$ screens and costs $\vec{c} = (c,\infty)$.
\end{definition}

The inspection depth here 
is $L=1$. The consumer can pay $c$ to reveal $f_{ij1}$, whereas $f_{ij2}$ can only be observed post-transaction.  We think of $F_{ij} := f_{ij1}$ as an indicator of \emph{fit} to the particular consumer, {with mean $\fitprob_{j} := \fitprob_{j1}$.}
Whereas $q_{ij} := f_{ij2}$ is a \emph{quality realization} that can only be observed post-transaction.
{Its mean is $Q_{j} := \fitprob_{j2}$,  
which may represent} 
an intrinsic \emph{quality} of firm $j$.

{We think of this market as one} \emph{without} transcripts, in contrast to the following variant \emph{with} transcripts.

\begin{definition}
\label{defn:fit-quality-transcripts}
A fit-quality market \emph{with transcripts} is special case of our model with $r \geq 2$ screens and cost vector $\vec{c}$ of the form $(c, 0, \dotsc, 0, \infty, \dotsc, \infty)$.
\end{definition}

Thus, $\vec{c}$ is parameterized by the cost $c \geq 0$ and the number $k\leq r$ of finite entries. %
Like in fit-quality markets, inspection is one-shot
for each firm: pay $c$ to reveal the first $k$ screens.
We call $c,k$, resp., the \emph{cost} and the \emph{informativeness} of the search.
{Increasing $k$ at fixed cost $c$ moves screens from post-transaction experience to pre-transaction inspection, leaving the underlying screen sequence unchanged.
In this single-cost inspection setting, increasing informativeness with transcripts weakly improves market learning and long-run consumer utility.}

\begin{theorem}\label{thm:cheaper.search.inform}
The conclusions of Theorem~\ref{thm:cheaper.search.general} apply also to any pair
$\mM,\tmM$ of fit-quality markets with transcripts that are identical except that $\tmM$ has increased search informativeness.
\end{theorem}

From the perspective of a single consumer {with fixed beliefs}, a fit-quality market with transcripts is equivalent to a fit-quality market (without transcripts) through a change of notation.  One can take $F_{ij} := \prod_{1 \leq \ell \leq k} f_{ij\ell}$ for the information revealed by search, and $q_{ij} := \prod_{k < \ell \leq r} f_{ij\ell}$ for the \emph{quality realization} observed post-transaction.
{With transcripts, the market observes individual screen outcomes, including which screen first fails; without transcripts, it observes only aggregate inspection and transaction outcomes.
Here ``why'' refers to the failed screen's identity, not a separate causal explanation of consumer preferences.
Consumers in both variants update Bayesian beliefs from the observed feedback.}

This correspondence also allows us to define search informativeness in markets \emph{without} transcripts. Fix a fit-quality market $\mM$ with transcripts and  corresponding fit-quality market $\mM'$ without transcripts. Increasing $k$ in $\mM$ causes, in $\mM'$, each fit probability $\fitprob_j$ to decrease by some factor $\delta_j \geq 1$ while quality $Q_j$ increases by the same factor $\delta_j$.
This motivates the following definition:

\begin{definition}
    Suppose $\mM$ is a fit-quality market without transcripts, and $\tmM$ is an identical market except that, for each firm $j$, $\tilde{\fitprob}_j = \fitprob_j / \delta_j$ and $\tilde{Q}_j = Q_j \times \delta_j$ for some $\delta_j \geq 1$. Then we say that $\tmM$ is \emph{more consumer-informative} than $\mM$.
\end{definition}

{A more consumer-informative market preserves the ex ante probability that the firm's product meets the consumer's need:
$\tilde{\fitprob}_j\tilde{Q}_j=(\fitprob_j/\delta_j)(Q_j\delta_j)=\fitprob_j Q_j$.
The comparison thus concerns this shift from post-transaction to pre-transaction information, which is a specific form of signal improvement.
}

Given Theorem~\ref{thm:cheaper.search.inform}, 
one might guess that increased consumer-informativeness improves learning and long-run consumer utility.
But, as we now show, this is not always the case.

\begin{proposition}\label{prop:inform.bad}
There exist fit-quality markets $\mM,\tmM$ without transcripts such that $\tmM$ is more consumer-informative than $\mM$,
and yet it is strictly worse in the long-run:
 $\learnedS(\tmM) \subseteq \learnedS(\mM)$
 and
 $U(\tmM) \leq U(\mM)$
almost surely over some coupling between the histories of $\mM$ and $\tmM$. 
{This inclusion (resp., inequality) is strict with positive probability over the coupling.}
\end{proposition}

{In the example behind Proposition~\ref{prop:inform.bad}, a firm's fit probability is either high or low, while its quality is known.
More consumer-informative search moves a previously post-transaction screen into inspection, so an aggregate failure no longer identifies a failure of the original fit screen.
Pooling these failures weakens the evidence about the firm's fit probability; even a high-fit firm now fails inspection more often and can be prematurely abandoned.}

\section{\pricingnew{Numerical Simulations}}
\label{sec:expts-body}

We complement our theoretical analysis with numerical simulations of market dynamics. We focus on the impact of (i) improved consumer search and (ii) (un)observable transcripts, in the sense of Section~\ref{sec:compare}, \pricingnew{over a finite sequence of consumer arrivals}.

\xhdr{\pricingnew{Experimental setup.}}
We simulate a stylized market with 15 firms and 5 screens. \pricingnew{Each run has $T=5{,}000$ consumers unless otherwise specified, with firm parameters and other sources of randomness redrawn independently in each run.}\footnote{In our base simulation, values $V_i$ are drawn from $U[0.5,2]$; prices $P_j$ are drawn from $\mathcal{N}(0.08, 0.04^2)$ clipped to $[0.01, 0.20]$; each $\pi_{jk}$ is drawn independently from Beta(4,2) clipped to $[0.25,0.95]$; and we take $R=50$ runs.} All shaded regions and error bars in plots represent 95\% confidence intervals on the means.
\pricingnew{We track firm losses, posterior belief error, and consumer utility.} We refer to firms not lost by $T$ as \emph{surviving}.\footnote{\pricingnew{The plot label ``learned'' denotes the finite-horizon count, not verified convergence. Panels labeled ``steady-state utility'' report terminal welfare estimates, not demonstrated limit outcomes.}} \pricingnew{Appendix~\ref{sec:simulations} gives the full experimental specification and additional diagnostics.
}

\paragraph{Impact of Search Cost.}
\pricingnew{We first fix inspection depth $k=3$, use transcript feedback, and vary $c\in\{0.005,0.01,0.02,0.035,0.05\}$. Figure~\ref{fig:cost_outcomes} reports terminal outcomes.} Lower costs reduce belief error and firm losses and raise consumer utility over the simulated horizon. At low costs ($c \leq 0.01$), nearly all firms survive to $T=5{,}000$, and belief error is small.

We distinguish between \emph{correct} loss (where the index evaluated at the true screen probabilities is truly negative) versus \emph{incorrect} loss (where the true index is positive and hence the loss is due to inaccurate posterior estimates).
In all settings where losses occur, the majority of the observed losses are incorrect, driven by unlucky streaks of screen failures in initial inspections. \pricingnew{Once a firm's index at $\bar V$ becomes negative, no later consumer inspects it, so its pessimistic beliefs persist. The cost sweep thus illustrates how cheaper search can preserve viable firms as well as directly reduce consumers' inspection costs.}

\begin{figure}[!htbp]
\centering
\includegraphics[width=\textwidth]{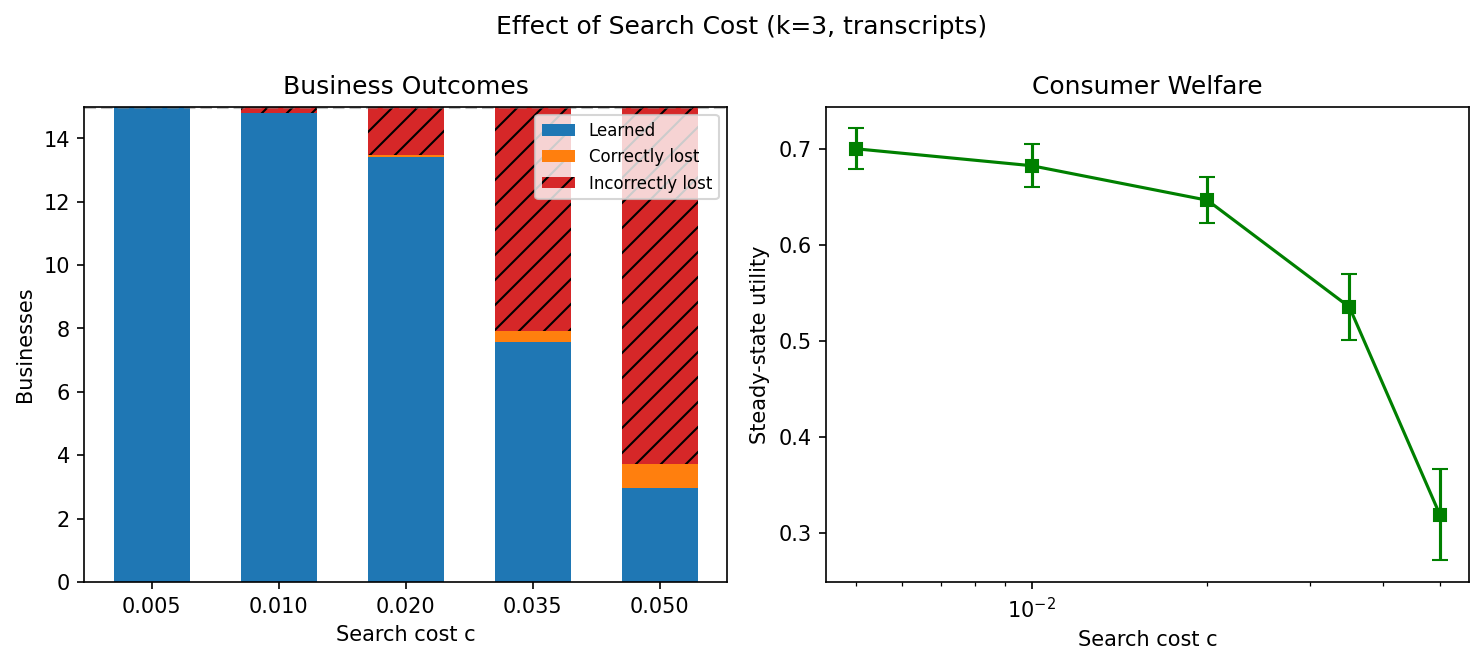}
\caption{Cost sweep outcomes at $T=5{,}000$: surviving firms and losses (left), and terminal welfare estimates (right). Losses are ``correct'' if the index evaluated at the true screen probabilities is negative, and ``incorrect'' if it is positive.}
\label{fig:cost_outcomes}
\end{figure}

\paragraph{Impact of Search Informativeness.}
We next fix cost $c = 0.02$ and vary informativeness $k \in \{1, 2, 3, 4, 5\}$.
\pricingnew{With transcripts (T), the platform records individual screen outcomes; without transcripts (NT), it records only aggregate inspection and transaction outcomes. The distinction concerns the feedback available to later consumers, not the information acquired by the searching consumer.}\footnote{To align with our analytical comparisons (Theorem~\ref{thm:cheaper.search.inform} and Proposition~\ref{prop:inform.bad}), we keep the same initialization of the platform posterior estimates as we vary $k$, so that only consumer informativeness changes across treatments.}
\pricingnew{Figure~\ref{fig:sim_plots_body}(a)--(b) compares terminal outcomes.} At $k=5$ and $T=5{,}000$, transcripts leave 13.3 surviving firms on average versus 6.7 under NT, with terminal welfare estimates of 0.92 versus 0.76.
\pricingnew{Utility increases with $k$ under both feedback regimes, even though losses under NT increase at high $k$. Aggregate feedback pools failures at different screens; screen-level observations instead allow the platform to distinguish which requirements passed and which failed.}

\begin{figure}[!htbp]
\centering
\begin{subfigure}[b]{0.48\textwidth}
\includegraphics[width=\textwidth,trim=335.08bp 0bp 337.89bp 20bp,clip]{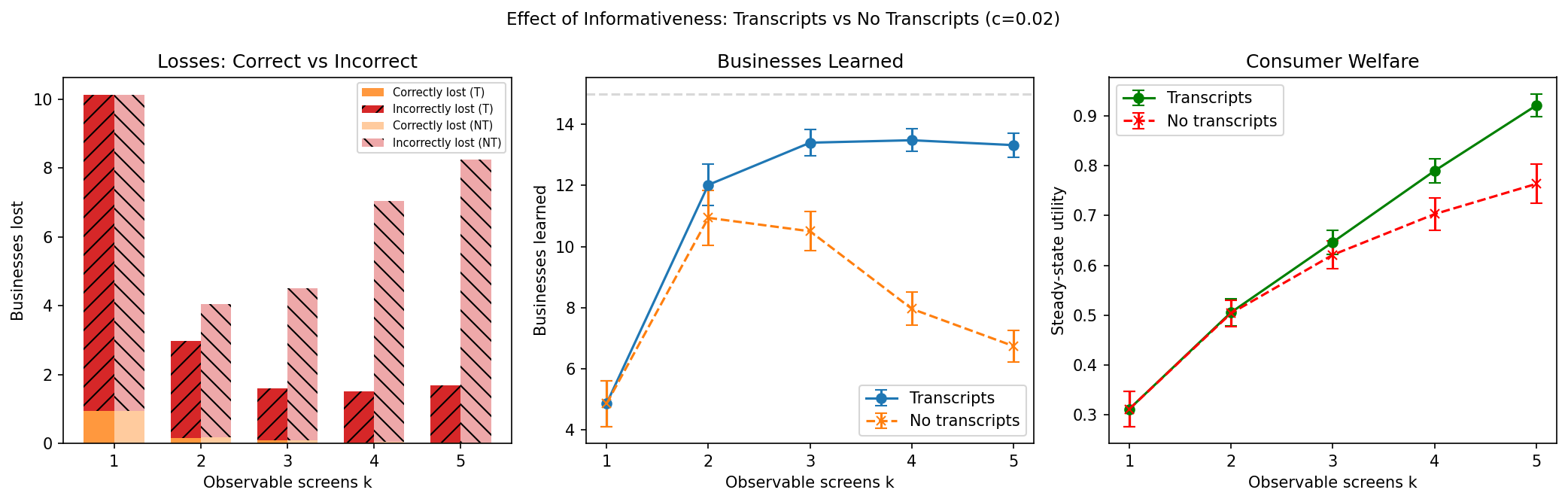}
\caption{\pricingnew{Surviving firms at $T=5{,}000$.}}
\end{subfigure}
\hfill
\begin{subfigure}[b]{0.48\textwidth}
\includegraphics[width=\textwidth,trim=662.35bp 0bp 0bp 20bp,clip]{info_sweep_outcomes.png}
\caption{\pricingnew{Terminal welfare estimates.}}
\end{subfigure}

\medskip
\begin{subfigure}[b]{0.48\textwidth}
\includegraphics[width=\textwidth,trim=361.88bp 0bp 361.88bp 34.87bp,clip]{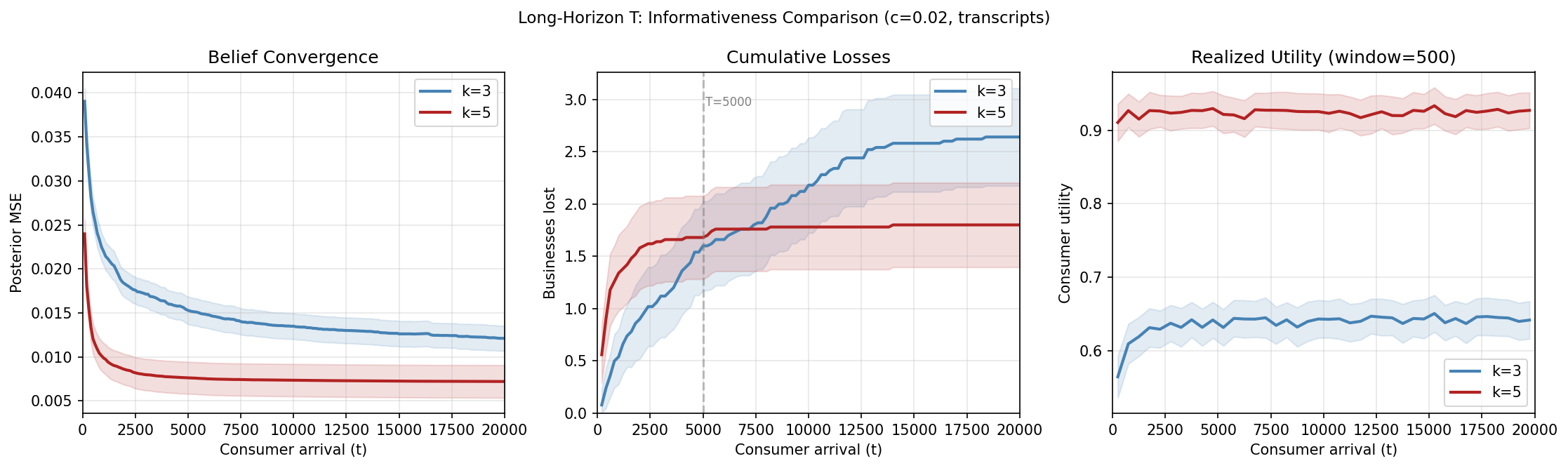}
\caption{\pricingnew{Extended loss trajectories with transcripts.}}
\end{subfigure}
\hfill
\begin{subfigure}[b]{0.48\textwidth}
\includegraphics[width=\textwidth,trim=361.88bp 0bp 361.88bp 34.87bp,clip]{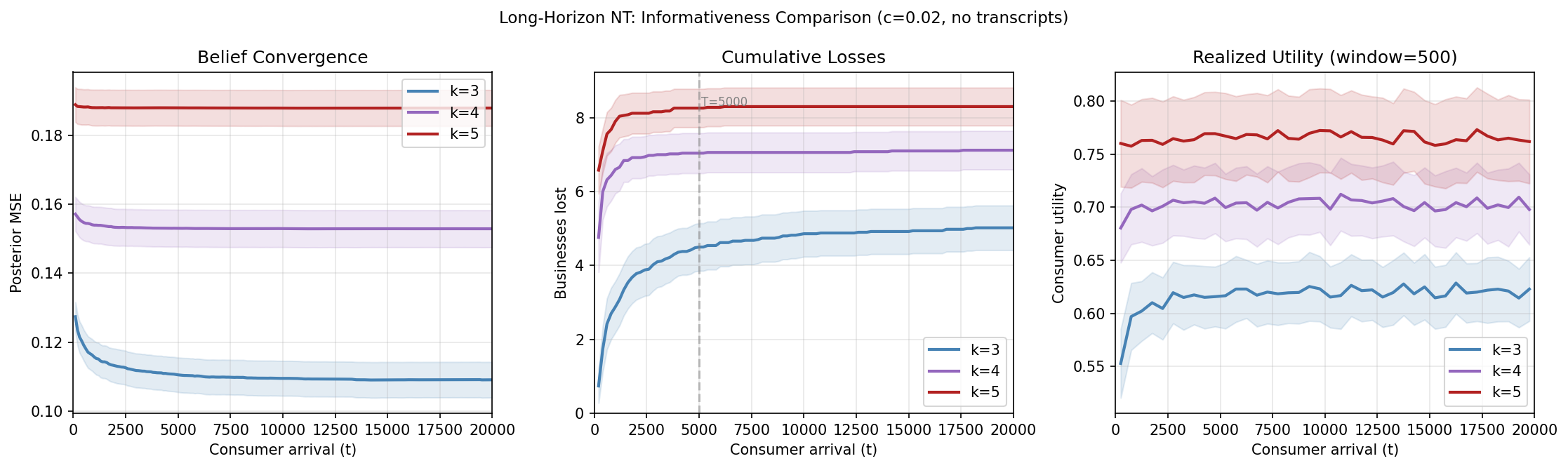}
\caption{\pricingnew{Extended loss trajectories without transcripts.}}
\end{subfigure}
\caption{\pricingnew{Informativeness and feedback at $c=0.02$. Top: surviving firms and welfare estimates in the baseline sweep. Bottom: selected comparisons extended to $20{,}000$ arrivals; the dashed lines mark $T=5{,}000$. With transcripts, the earlier loss ordering reverses; without transcripts, the ordering across $k=3,4,5$ persists through the extended horizon. Shading and error bars show 95\% confidence intervals on the means.}}
\label{fig:sim_plots_body}
\end{figure}

\xhdr{\pricingnew{Finite-horizon versus limiting outcomes.}}
\pricingnew{With transcripts, the number of surviving firms at $T=5{,}000$ is also non-monotone in $k$, despite the long-run comparison in Theorem~\ref{thm:cheaper.search.inform}. The trajectory diagnostics in Appendix~\ref{sec:simulations} suggest that losses accumulate more slowly under less informative search: some firms that survive the shorter experiment are still vulnerable to later loss. To investigate this finite-horizon explanation, we compare $k=3$ and $k=5$ over $T=20{,}000$ rounds (Figure~\ref{fig:sim_plots_body}(c)).}
The average loss count for $k=3$ overtakes that for $k=5$, reversing the earlier ordering.
\pricingnew{By contrast, under NT the ordering across $k=3,4,5$ persists through $20{,}000$ arrivals (Figure~\ref{fig:sim_plots_body}(d)). %
}

\paragraph{\pricingnew{Robustness Checks.}}
\pricingnew{Appendix~\ref{sec:simulations} varies the number of firms, the screen-probability distribution, consumer-value heterogeneity, and prices. Cheaper search improves terminal welfare estimates and surviving-firm counts across the five variants. Welfare estimates increase with informativeness in all transcript variants, while finite-horizon survivor counts are non-monotone in two. Without transcripts, losses are non-monotone in informativeness in all five variants, and welfare estimates need not be monotone. These checks illustrate the distinction between consumer information and market feedback beyond the baseline parameterization.}
\FloatBarrier

\section{\pricingnew{Extensions}}
\label{sec:extensions}

\subsection{\pricingnew{Endogenous Inspection Depth}}
\label{sec:extension.depth}

We can allow each consumer $i$ to strategically choose her inspection depth $L_i$ before the round. This is a two-fold commitment: each firm can be inspected \emph{at most} $L_i$ times, and any firm transacted with must be inspected \emph{exactly} $L_i$ times. We restrict $L_i$ to the range $[\Lmin, L]$, for some $\Lmin\in[L]$, so that our main model is a special case with $\Lmin=L$.
\pricingnew{The chosen depth cannot vary across firms or be revised in response to inspection outcomes. The screen-level inspection and transaction feedback is unchanged. All proofs for this extension appear in Appendix~\ref{sec:depth.proofs}.}

The consumer chooses $L_i$ based on her value, beliefs, and inspection costs. For example, a low-value consumer may choose a smaller $L_i$ to avoid an expensive last screen: e.g., inspect a webpage but avoid a phone call. Also, a consumer may choose a smaller $L_i$ if she believes that the last screen(s) are very likely to pass; but future consumers may reconsider this if that belief shifts due to the history.

\begin{pricingrevision}
\xhdr{Optimal search.}
Write $\mathcal L=\{\Lmin,\ldots,L\}$ for the feasible depths and $\sigma_{j,\ell}(V;\boldsymbol p_j)$ for the index in Algorithm~\ref{alg:index} at depth $\ell$, value $V$, and posterior means $\boldsymbol p_j=(p_{jk})_{k\in[r]}$:
\[
\sigma_{j,\ell}(V;\boldsymbol p_j)
=V\prod_{k>\ell}p_{jk}-P_j
-\sum_{k\leq\ell}\frac{c_k}{\prod_{h=k}^{\ell}p_{jh}}.
\]

\begin{theorem}[Optimal search and depth]\label{thm:depth.optimal}
Fix consumer $i$, value $V_i$, and history $H$. For each $\ell\in\mathcal L$, Algorithm~\ref{alg:index} with $L=\ell$ maximizes expected utility among policies committed to depth $\ell$. An optimal policy chooses a depth that maximizes these conditional expected utilities. The depth and search policy can be computed in $O(|\mathcal L|m(r+\log m))$ arithmetic operations.
\end{theorem}

The conditional search claim follows Theorem~\ref{thm:optimal.search}; only the finite maximization over depths is new. For each depth, sort firms by index and evaluate the resulting policy by backward induction. Ties between depths may be resolved by any fixed measurable rule based on value and history.

\xhdr{Learning and exclusion.}
For $S\subseteq[m]$, let $U_\ell^S(V)$ be the optimal expected utility of a consumer with value $V$ who knows the true parameters of firms in $S$, searches only those firms, and commits to depth $\ell$. Define
\[
U_{\mathrm{endo}}^*(S)
:=\E_{V\sim\distV}\!\left[\max_{\ell\in\mathcal L}U_\ell^S(V)\right],
\qquad
\mathcal L^*(V;S):=\arg\max_{\ell\in\mathcal L}U_\ell^S(V).
\]
These are static comparators, defined independently of the learning dynamics, like $U^*(S)$ in Section~\ref{sec:convergence}. We obtain the following counterpart of Proposition~\ref{prop:market.learning}.

\begin{proposition}[Learning]\label{prop:depth.learning}
Almost surely, every screen posterior converges weakly. Write $S$ for the firms inspected infinitely often. Each firm in $S$ is learned: its screen posteriors converge to point masses at the true parameters. Every remaining firm is inspected only finitely often and its beliefs eventually stop changing. Calling these firms lost and writing $\bar{\boldsymbol p}_j$ for their frozen posterior means, we have
\[
\sigma_{j,\ell}(V;\bar{\boldsymbol p}_j)\leq0
\quad\text{for all }j\notin S,\quad
V\in\operatorname{supp}(\distV),\quad
\ell\in\mathcal L^*(V;S).
\]
\end{proposition}

In Proposition~\ref{prop:depth.learning}, exclusion is relative to every depth optimal for the learned market and the consumer's value, not one exogenous depth evaluated at $\oV$ as in Proposition~\ref{prop:market.learning}. The inequality concerns frozen beliefs, not true firm parameters. A lost firm may have positive index at a different, unchosen depth. The argument is similar to the proof of Proposition~\ref{prop:market.learning}, though we must additionally rule out a strict benefit from adding an excluded firm to an optimal learned-market policy, possibly with a change in inspection depth.

\begin{lemma}[Utility convergence]\label{lem:depth.utility}
Under the optimal policies in Theorem~\ref{thm:depth.optimal},
\[
\E[U_i\mid\hist_{i-1}]\longrightarrow U_{\mathrm{endo}}^*(S)
\quad\text{almost surely}.
\]
\end{lemma}

As in Lemma~\ref{lem:util.converge}, the limit is the utility of searching only the learned firms, now optimizing inspection depth as well. Expected utility converges even when the chosen depth does not converge because of ties. The proof combines continuity over the finite collection of depths and search orders with a vanishing probability of inspecting lost firms. 

Unlike Proposition~\ref{prop:lost}, we do not assert independence of firm losses: optimal depth choices depend on beliefs about the whole market. We leave open whether/which comparative statics in Proposition~\ref{prop:lost} and Section~\ref{sec:compare} extend.
\end{pricingrevision}

\subsection{\pricingnew{Endogenous Prices}}
\label{sec:extension.pricing}

\begin{pricingrevision}
The fixed-price analysis isolates the effect of feedback on market learning. When firms adjust prices, improved matching also changes competition and how gains from trade are divided. We study this additional channel in the single-cost setting. Proofs and further formalization appear in Appendix~\ref{sec:strategic.firms}.

\xhdr{Model and learning.}
We consider a fit-quality market (Definition~\ref{defn:fit-quality}), or its transcript variant (Definition~\ref{defn:fit-quality-transcripts}), with fixed inspection depth. {Recall that the \emph{type} of each firm $j$ is a tuple $(\pi_j,Q_j)$, where $\pi_j$ is expected fit and $Q_j$ is quality.} 

Firm types are fixed, but each firm $j$ chooses a publicly observed price $P_j(H)\geq0$ before each consumer searches. Pricing rules depend only on public history $H$ and randomization independent of private values, screen outcomes, and types conditional on $H$. Firms have no additional private information, so prices themselves reveal no further type information. With posterior means $\hat\pi_j(H),\hat Q_j(H)$, consumers inspect firms in decreasing order of positive indices
\[
\sigma_j=V\hat Q_j(H)-P_j(H)-c/\hat\pi_j(H),
\]
and transact with the first fit. Throughout this subsection, values have a continuous, strictly positive density $f$ on $[\underline V,\oV]$. For the equilibrium results, assume regularity: $v-(1-F(v))/f(v)$ is non-decreasing, where $F$ is the value CDF. These assumptions are local to the pricing analysis.

For learning, we impose \emph{uniform demand-responsiveness}: there is a fixed $\eta\in(0,1)$ such that at every history, whenever $a_j(H):=\oV\hat Q_j(H)-c/\hat\pi_j(H)>0$,
\[
0\leq P_j(H)\leq(1-\eta)a_j(H).
\]
Roughly speaking, whenever positive demand is possible at a non-negative price, firms choose prices that generate strictly positive demand.\footnote{\pricingnew{This is slightly stronger than requiring $P_j(H)<a_j(H)$ whenever $a_j(H)>0$. The fixed margin excludes edge cases where, with unchanged beliefs, prices approach the no-demand threshold asymptotically and inspection probabilities become summable.}}

\begin{proposition}[Learning]\label{prop:pricing.main.learning}
Under uniformly demand-responsive pricing rules, almost surely each firm's posterior beliefs converge weakly. For each firm $j$, either the limit beliefs are concentrated on its true type $(\pi_j,Q_j)$, or its beliefs stop changing after finitely many observations, with $a_j(H)\leq0$ thereafter. Thus each firm is learned or lost even at price zero.
\end{proposition}

To compare these limits across markets, we couple histories using the same per-firm feedback streams, as in the proof of Theorem~\ref{thm:cheaper.search.general}. Each market reveals a stream at its own inspection times; with transcripts, streams are coupled screen by screen. The markets therefore have identical posterior beliefs after observing the same feedback from a firm, even when inspection times differ. Under this coupling, the learned set contains that of an otherwise identical fixed-price market (Proposition~\ref{prop:market.lost.pricing}) and weakly expands when search is cheaper, or more informative with transcripts (Proposition~\ref{prop:market.lost.search.pricing}). %
These are comparative statics on the set of learned firms. In the remainder of this section, we discuss comparative statics on welfare, revenue, and consumer surplus, which depend not only on the learned set but also on the prices chosen by the firms.

\xhdr{Steady-state pricing and symmetric markets.}
Given limit beliefs, a steady-state equilibrium is a profile of price distributions, each maximizing the firm's expected per-consumer revenue against the others assuming optimal consumer search. Lost firms have non-positive index even {at price zero} and receive no demand.

In general, steady-state equilibrium is not uniquely determined by the limit beliefs, as there can be a multiplicity of equilibrium prices.
This motivates us to focus, for the remainder of this section, on 
a special case of markets in which all learned firms are homogeneous, i.e., have the same type 
$(\pi, Q)$. %
We show that such markets do admit a unique symmetric pricing equilibrium given limit beliefs.

Formally, we call a market \emph{learned-symmetric} if, for any choice of uniformly demand-responsive pricing rules and almost surely over the resulting history, all firms in the learned set have the same type. For example, consider a market where the common prior over firm types contains only a single type $(\pi^*, Q^*)$ with non-negative index $\sigma_j = \overline{V} Q^* - c/\pi^*$ at price zero; for all other firm types this quantity is strictly negative. In such a market, any firm with a type other than $(\pi^*, Q^*)$ will necessarily be lost, and hence all learned firms in any steady-state {equilibrium} must all share the common type $(\pi^*, Q^*)$.

Under any limiting beliefs where all learned firms are homogeneous (which occurs almost surely for learned-symmetric markets), there is a unique symmetric steady-state equilibrium.

\begin{theorem}[Symmetric equilibrium]\label{thm:pricing.main.symmetric}
Fix a market and a set $S$ of $n \geq 2$ homogeneous firms. Under any steady-state beliefs with learned set $S$, the learned firms have a unique symmetric equilibrium price distribution.
\end{theorem}

Our proof of Theorem~\ref{thm:pricing.main.symmetric} is constructive, yielding a closed-form CDF for the distribution over prices chosen by each firm (see Theorem~\ref{thm:symmetric.equil}). 
Write $\conW(\mM),\Rev(\mM),\totW(\mM)$ for the expected per-consumer surplus, total firm revenue, and total welfare at the equilibrium described in Theorem~\ref{thm:pricing.main.symmetric}, with $\totW=\conW+\Rev$.
Note that these are random quantities for a given market, due to the randomness of the learned set.

\xhdr{Comparative statics.}
For learned-symmetric markets, cheaper search improves consumer surplus and total welfare. More informative search can instead weaken competition.

\begin{theorem}[Cheaper search]\label{thm:pricing.main.cheaper}
Suppose markets $\mM,\tmM$ are learned-symmetric and identical except that $\tilde c<c$. Under uniformly demand-responsive pricing rules in both markets, there is a coupling of histories such that, almost surely,
\[
\conW(\tmM)\geq\conW(\mM),\qquad
\totW(\tmM)\geq\totW(\mM).
\]
\end{theorem}

Total seller revenue can instead fall in expectation over histories: fewer firms may be lost, increasing competition. Conditional on the same number of learned firms, cheaper search weakly increases total seller revenue.

\begin{proposition}[More informative search]\label{prop:pricing.main.informative}
There exist learned-symmetric %
markets with transcripts, identical except for increased search informativeness in $\tmM$, such that no firms are ever lost in either market and $\conW(\tmM)<\conW(\mM)$.
\end{proposition}

This comparison isolates a competition effect rather than a learning failure. More selective screening leaves each fitting firm facing fewer competing fits, allowing higher prices. Thus observing screen-level feedback does not by itself ensure that consumers benefit from more informative search when firms adjust prices.
\end{pricingrevision}

\begin{pricingrevision}
\xhdr{Scope and Limitations.} 
The learned-symmetric assumption in Theorem~\ref{thm:pricing.main.cheaper} is a substantial restriction. 
The cheaper-search equilibrium comparison restricts attention to realizations in which all firms in the learned set share the same type. 
This is a special-case comparison under the maintained independent priors, not a model in which consumers know which firms {have the same type}. Extending this equilibrium comparison beyond learned-symmetric markets remains open. In contrast, the negative informativeness result (Proposition~\ref{prop:pricing.main.informative}) immediately extends to general markets with transcripts, since they contain learned-symmetric markets as a special case.
\end{pricingrevision}

\section{\pricingnew{Discussion}}
\label{sec:discussion}

\pricingnew{We introduced a model of consumer search with market-level learning. At fixed prices, cheaper search improves long-run learning and consumer surplus, while the effects of more informative search depend on the feedback available to later consumers. With endogenous prices, more selective screening can also weaken competition, reducing consumer surplus even when no firms are lost. These results motivate the following further directions.}

\xhdr{\pricingnew{Consumer heterogeneity and adoption.}}
{For finitely many consumer types observed by the platform, one can extend our results by assuming independent priors over type-specific screen probabilities and maintaining separate posterior beliefs.
If each type arrives infinitely often and its conditional model satisfies the baseline assumptions, its arrival subsequence can be analyzed as a separate market.}
\pricingnew{We also assume homogeneous search costs across consumers. Adoption of agentic technology is likely to be uneven, however, with different populations experiencing different effects. This raises distributional questions about who benefits from cheaper search. If early adopters differ systematically from later adopters, the information they generate may be more or less relevant for other consumer types. This type-by-type analysis does not address unobserved types or learning across populations with different preferences.}

\xhdr{\pricingnew{Richer preferences and search technologies.}}
\pricingnew{The screening model can be extended, e.g., to accommodate fractional outcomes or more complex aggregation of these outcomes (e.g., reflecting multi-dimensional preferences).}
\pricingnew{Our model permits interleaving inspections across firms and returning to partially inspected firms, but fixes the screen order within each firm. Allowing consumers to reorder these screens can change which screens are reached before failure and hence what is learned, even when expected inspection cost falls. Such reordering is not generally equivalent to reducing costs: any cost-reduction interpretation must preserve screen probabilities and feedback and satisfy the coordinate-wise comparison of Theorem~\ref{thm:cheaper.search.general}.}

\begin{pricingrevision}
\xhdr{Firm investment and competition.}
Our model holds firm quality fixed, but firms may invest to improve it in response to market incentives. Better search may increase returns to quality investment by helping consumers identify high-quality firms, but may also intensify price competition and shrink margins. Investment on a slower timescale than pricing could change long-run market structure and make previously learned quality estimates stale. Endogenous pricing in the general multi-cost screening model also remains open.

\xhdr{Feedback reliability and strategic behavior.}
We also assume consumers provide truthful feedback and that the specified inspection and transaction signals are observed. In practice, feedback may be strategic, noisy, or absent. The model abstracts from manipulation of signals and from misalignment between agents and the consumers or firms they serve. Consumers, firms, and platforms may also have incentives to select or misreport what reaches later consumers. Understanding how robust our results are to these departures is an important direction for future work.

\xhdr{Platform design and privacy.}
Our analysis highlights the importance of which inspection signals a platform observes and aggregates, not just how much a consumer learns before purchase. Transcripts here identify screen outcomes, not unrestricted causal explanations; our informativeness comparison moves specified screens into pre-transaction inspection rather than comparing arbitrary information improvements. \textcolor{black}{Firms (or the platform) could also subsidize search costs (esp. for AI agents) to improve learning and/or attract customers.} Finally, platforms optimizing for objectives other than consumer welfare (such as revenue or engagement) may make different choices about search and feedback. We do not analyze their design incentives.

Observing screen-level evidence need not require releasing raw conversations. Under the baseline's independent screen priors and truthful feedback, per-firm, per-screen counts of observed successes and failures suffice for posterior updating, so the platform can share these statistics rather than individual transcripts. Such aggregates can still disclose sensitive information: sufficiency for learning is not a formal privacy guarantee. Privacy-preserving observation and aggregation remain open design questions.
\end{pricingrevision}

\section*{Acknowledgements and Disclosures}
\pricingnew{We thank Salvador Candelas, Patrick Jordan, and Matthew Vogel for helpful discussions. We especially thank Salvador Candelas for noting connections with~\citet{choi2018consumer}.}

{The authors are employees of Microsoft. This work received no
third-party funding or support. Beyond their employment at Microsoft,
the authors have no competing financial interests to disclose.}

\pricingnew{This research does not involve crowdsourcing or research with human subjects. LLMs were used for editing and formatting purposes, and as a programming aid during the preparation of the numerical simulation code and result plots. LLMs were also used to review and revise mathematical arguments, proofs, and examples.}

\newpage
\addtocontents{toc}{\protect\setcounter{tocdepth}{2}}
\renewcommand{\contentsname}{APPENDICES}
\tableofcontents
\newpage

\appendix

\section{Technical Connection to Pandora's Box Problem and Multi-Armed Bandits}
\label{app:connection}

Let us spell out the technical connection between our model and these well-studied problems.

The optimal-search problem faced by a single consumer is related to Weitzman's \emph{Pandora's box problem} ~\citep{weitzman1978optimal}. Specifically, a special case of our model with a single inspectable screen (i.e., $L=1$) is also a special case of this problem. Here, firms $j$ correspond to ``boxes". The boxes are inspected via the same process as described above, then one of the inspected boxes can be ``opened" and its reward is realized as an independent random draw from some known distribution. For us, the reward distribution for box $j$ is given by the current Bayesian posterior on ``quality", i.e., the product of uninspectable screen probabilities. Within the scope of Pandora's box problem, we focus on the special case when the inspection outcome is binary. Moreover, we posit \emph{obligatory inspection} (a common assumption that a box must be inspected before it is opened; {see~\citet{beyhaghi2024recent} for a survey of problem variants with and without this assumption}) and an \emph{outside option} (not opening any box is also a possibility). An optimal solution, for general inspection outcomes, is given by the celebrated Weitzman's algorithm \citep{weitzman1978optimal}. 

Collectively, the consumers face a sequential learning problem with a notion of ``bandit feedback": only the outcomes of consumer's actions (i.e., inspections and transaction) are revealed to future consumers, whereas the counterfactuals are not. Hence, the consumers collectively face a trade-off between exploration and exploitation, of which multi-armed bandits is a fundamental model. 

Simple (and, in some sense, degenerate) cases of our model correspond to bandit problems studied in prior work. In a bandit problem, one ``pulls" some arm in each round, observing a random reward drawn from some fixed but unknown distribution specific to this arm \citep{bandits-ucb1}. Now, firms correspond to ``arms".  Consider a  variant of our model when all inspectable screens are known  (so there's no point in inspecting any firm) and some screens are uninspectable. Moreover, assume all prices are $P_j\equiv 0$ and all values are $V_i\equiv 1$. If all fits are inspectable screens are $f_{ijk}=1$, this model variant corresponds to the bandit problem described above, where reward of each arm $j$ is the product of its uninspectable screens. 
\footnote{Notably, Bernoulli rewards is a paradigmatic special case in bandits (along with Gaussian rewards).}
For general screens $f_{ijk}$, we obtain the \emph{sleeping bandits} problem \citep{sleeping-colt08}, where  the set of available arms changes from one round to another.
For general values $V_i$ (and going back to $f_{ijk}\equiv 1$) we have a special case of \emph{contextual bandits}: a bandit problem in which a payoff-relevant \emph{context} is revealed before each round.%
\footnote{Contextual bandit problems come in many variants \citep[e.g., see books][]{slivkins-MABbook,LS19bandit-book}. We have a rather simple special case, where the context is $V_i$ and mean rewards are given by %
the transaction utility.
}

Of course, the crucial difference of our model compared to these bandit problems is that each consumer can ``inspect" arms before ``pulling" them (let alone inspect them multiple times).

While in bandit problems one typically designs an algorithm (which trades off exploration and exploitation in the long run), consumer behavior in our model is "greedy", optimizing for a particular round. Projecting this back to bandits, it corresponds to the ``greedy" bandit algorithm: one that chooses the posterior-best arm in each round.

\newpage

\section{Proof of Theorem~\ref{thm:optimal.search}: Optimal Consumer Search}
\label{app:optimal.search}

We will actually consider a slightly generalized class of problem instances and prove optimality of a correspondingly generalized form of Algorithm~\ref{alg:index}.  In our generalized problem, each firm $j$ could have a firm-specific number of screens $r_j$, the consumer can choose a different inspection depth $L_j$ for each firm, inspection costs are firm-specific and denoted $c_{jk}$, and the outside option value of not transacting with any firm is $Z \geq 0$. The original problem is a special case where the values $r_j$ are identical across firms, as are the values $L_j$ and the costs $c_{jk}$ for each $k$, and the outside option value is $Z = 0$.
We claim that the optimal policy for this more general class of problems is Algorithm~\ref{alg:index.generalized}, which is simply Algorithm~\ref{alg:index} with the appropriate firm-specific quantities $r_j$, $L_j$, and $c_{jk}$ used in place of $r$, $L$, and $c_k$ respectively, and value $Z$ used in place of $0$ as the index threshold at which to stop inspecting firms.

\begin{algorithm}[t]
	\SetAlgoNoLine\SetAlgoNoEnd
        \medskip
    \For{each firm $j$}
        {
        \[ \sigma_{ij} := V_i \cdot \prod_{L_j < \ell \leq r_j}\fitprobest_{j\ell} - P_j - \frac{\sum_{\ell \leq L_j}c_{j\ell} \prod_{\ell' < \ell}\fitprobest_{j\ell'}}{\prod_{\ell \leq L_j}\fitprobest_{j\ell}}.\]
        }
        \medskip
	\For{each firm $j$ with $\sigma_{ij}>Z$, in order of decreasing index $\sigma_{ij}$,}
        {
        \For{each $\ell = 1,2 \LDOTS L_j$}
        {
        Pay $c_{j\ell}$ to inspect $f_{ij\ell}$\;
        If $f_{ij\ell}=0$, break and continue to the next firm.
        }
        Transact with firm $j$ and leave the market \quad \emph{// We must have $f_{ij\ell}=1$ for all $\ell\in [L_j]$}
		}
    \medskip
    Leave the market without transacting \quad \emph{// No firms with index exceeding $Z$ were a fit}
	\caption{Optimal search policy at firm-specific inspection depths $L_j$ and outside option $Z$}
	\label{alg:index.generalized}
\end{algorithm}

We prove optimality of Algorithm~\ref{alg:index.generalized} in two steps.  First, we restrict attention to policies where, if the consumer inspects one firm $j$ and then later inspects another firm $j' \neq j$, the consumer cannot return to firm $j$ to inspect further
or transact.  
We call such policies \emph{persistent}. 

\begin{definition}
    We say that a search policy is \emph{persistent} if, for every firm $j$ and in every history, the set of inspection actions for firm $j$ occur on contiguous steps, and any transaction with firm $j$ occurs immediately following an inspection of firm $j$.
\end{definition}

Note that Algorithm~\ref{alg:index.generalized} is a persistent policy.  We first show that Algorithm~\ref{alg:index.generalized} is optimal over the class of persistent policies.

\begin{lemma}
\label{lem:persistent}
    Given any market $\mM$, history $\hist_{i-1}$, and value $V_i$ for consumer $i$, Algorithm~\ref{alg:index.generalized} generates at least as much expected utility as any persistent policy.
\end{lemma}

\begin{proof}
First observe that, without loss of generality, we can limit attention to persistent policies that stop inspecting a firm $j$ after observing that some $f_{ijk}$ is equal to $0$, since any further inspection of firm $j$ is suboptimal. We can also limit attention to policies that, after inspecting $f_{ij1}$, continue to inspect firm $j$ until observing a $0$ or until screen $f_{ijL_j}$ has been inspected. This is because it is suboptimal to inspect a firm $j$ unless there is at least one possible realization of outcomes in which firm $j$ is transacted with, and this can only occur if $L_j$ screens of firm $j$ are inspected.

Under these assumptions, a persistent policy is determined entirely by the (possibly adaptive) order in which firms are inspected.
Given this observation, there is a natural reduction from the problem of finding an optimal persistent policy to the Weitzman Pandora's Box problem.  Indeed, given the choice to begin inspection of firm $j$, the process of inspection will either fail (a $0$ is observed on some screen) or succeed (all $L_j$ inspected screens are equal to $1$).  Write $\tilde{\fitprob}_j$ for the probability of passing all inspections for firm $j$ from $1$ to $L_j$: $\tilde{\fitprob}_j = \prod_{1 \leq \ell \leq L_j}\fitprobest_{j\ell}$.  If they pass and the transaction occurs, the expected value from the transaction is $\tilde{V}_i := V_i \cdot \prod_{L_j < \ell \leq r_j}\fitprobest_{j\ell}$. And the expected inspection cost is the total sum of costs from each step $\ell$ between $1$ and $L_j$, which is $c_{j\ell}$ times the probability that inspection step $\ell$ is reached.  Call this total inspection cost $\tilde{c}_j$.  Then the process of inspecting firm $j$ is equivalent in expectation to a single search operation of cost $\tilde{c}_j$ that reveals a transaction utility of $\tilde{V}_i - P_j$ with probability $\tilde{\fitprob}_j$, otherwise $-P_j$. 

A persistent policy can perform this inspection operation on firms, in any (possibly adaptive) order over firms. The optimal such policy
can therefore be described via Weitzman's algorithm \citep{weitzman1978optimal}, a known optimal solution to the Pandora's Box problem.  Weitzman's algorithm inspects firms according to an index $\sigma_{ij}$, in the same sense as in Algorithm~\ref{alg:index.generalized}, but for a more generic notion of index. Namely {(fixing the history $\hist_{i-1}$),} the index $\sigma_{ij}$ is the solution to 
    \[ \tilde{c}_j = \E\sbr{\rbr{U_{ij}-\sigma_{ij}}^+ {\mid \hist_{i-1}}}, \]
where $U_{ij}$ is the transaction utility and the expectation is over the random outcome of inspection. 
    Since $U_{ij}$ is $\tilde{V}_i - P_j$ with probability $\tilde{\fitprob}_j$ and otherwise $0$,\footnote{If firm $j$ is not a fit then the transaction utility, if a transaction with firm $j$ were forced, is $-P_j$. But since we assume tie-breaking in favor of not transacting, the consumer would never choose to make such a transaction. So we can take the box reward to be $0$ in this case. That said, all subsequent analysis would still hold without change if we defined $U_{ij}=-P_j$ in the case of a non-fit.} this expression simplifies to
    \[ \tilde{c}_j = \tilde{\fitprob}_{j}(\; \tilde{V}_i - P_j - \sigma_{ij}\;). \]
    Rearranging yields $\sigma_{ij} = \tilde{V}_i - P_j - \tilde{c}_j/\tilde{\fitprob}_j$. Substituting our expressions for $\tilde{V}_i$, $\tilde{c}_j$, and $\tilde{\fitprob}_j$ yields precisely the index rule from Algorithm~\ref{alg:index.generalized}.

    We are not done yet, since our definition of persistent policies requires also that the policy transact with firm $j$ immediately after observing the first $L_j$ screens to be $1$.
    In general the Weitzman search algorithm has no such constraint: it continues to search until the maximum value observed (or the outside option value of $Z$) exceeds the index of all remaining firms, then takes the best value revealed so far (or the outside option). We note, however, that in our setting a successful search of firm $j$ (i.e., where the first $L_j$ screens are all observed to be $1$) always reveals a transaction value weakly higher than the index $\sigma_{ij}$ for firm $j$, and hence higher than the index of all remaining firms. Thus, upon a successful inspection of firm $j$ up to screen $L_j$, the Weitzman optimal search policy will always immediately terminate and transact with firm $j$, as required of a persistent policy.
    
    We conclude that Algorithm~\ref{alg:index.generalized} is exactly the Weitzman optimal search policy for the corresponding Pandora's Box problem instance described above, %
    and is therefore optimal over persistent policies.
\end{proof}

We are now ready to prove the optimality of Algorithm~\ref{alg:index.generalized} over all (possibly non-persistent) policies.  We will proceed by induction on $x := \sum_j L_j$, the total number of inspectable screens over all firms.

For the base case, if $x = 1$ then there is only a single firm $j$. This means that any policy must be persistent, and hence by Lemma~\ref{lem:persistent} Algorithm~\ref{alg:index.generalized} is an optimal policy.

Likewise, if $x > 1$ and $n=1$, then there is again only one available firm $j$ and every policy must be persistent by definition. Thus, once again, Algorithm~\ref{alg:index.generalized} must be optimal by Lemma~\ref{lem:persistent}.

Now suppose $n \geq 2$ and hence $x \geq 2$. To prove the existence of an optimal policy that is persistent, we will show that no violation of the persistence conditions occurs along different potential execution outcomes of the policy. To this end, consider an optimal policy and suppose it begins by inspecting screen $f_{ij1}$ of firm $j$. 
If that inspection reveals that $f_{ij1} = 0$, then (without loss of generality) firm $j$ will not be explored further nor transacted with. We are therefore left with a reduced problem instance with one fewer firms and hence a reduced $x$. By induction, Algorithm~\ref{alg:index.generalized} is an optimal policy for this downstream subproblem. Since Algorithm~\ref{alg:index.generalized} is persistent for the subproblem, and the single inspection of firm $j$ trivially occurred contiguously, we conclude that the overall policy is also persistent on this path of its execution.

If, alternatively, the inspection of firm $j$ reveals $f_{ij1} = 1$, we are left with one of two possible reduced problem instances depending on $L_j$. If $L_j = 1$, meaning that we have now performed all necessary inspections to transact with firm $j$, this reduced subproblem corresponds to removing firm $j$ and increasing the outside option $Z$ to the expected value of transacting with firm $j$. Taking the outside option in this subproblem corresponds to transacting with firm $j$. Alternatively, if $L_j > 1$, then firm $j$ is still present in the reduced problem instance but has one fewer screens; that is, we lower $L_j$ and $r_j$ by one and perform a change of variables $c_{jk} \leftarrow c_{j(k+1)}$ and $\fitprob_{jk} \leftarrow \fitprob_{j(k+1)}$ so that the remaining screens are indexed from $1$ to $r_j - 1$. Induction then implies that, in both cases, Algorithm~\ref{alg:index.generalized} is optimal for this downstream subproblem. 

Since Algorithm~\ref{alg:index.generalized} is persistent, the only potential violation of persistence over all steps (including the initial inspection of firm $j$) is if firm $j$ is neither immediately transacted with (if $L_j = 1$) nor immediately inspected again (if $L_j > 1$). Suppose this violation of persistence occurs, as otherwise we are done.

Note that, from the definition of Algorithm~\ref{alg:index.generalized}, the order in which other firms (besides $j$) are searched, after the first inspection of $j$, follows an index rule and is therefore independent of whether the first inspection of firm $j$ revealed $f_{ij\ell}=0$ or $f_{ij\ell}=1$.
Given this, we claim that our supposedly optimal policy can be weakly improved in a manner that recovers persistence.  The improvement is to delay the first inspection of $j$ until just prior to the step in which the second inspection of $j$ occurs (if $L_j > 1$) or the step in which the outside option is taken (if $L_j = 1$) in the resulting inspection order conditional on $f_{ij1}=1$.  This modified policy is well-defined, since the subsequent ordering (up to the next inspection of firm $j$, if any) follows an index rule and therefore is not determined by the result of the inspection of firm $j$.

The reason this modified policy is an improvement is that the outcome of the initial inspection of $j$ does not influence which firms are inspected, and in what order, up until the next step in which firm $j$ would be inspected or transacted with.
Indeed, if the initial inspection of $j$ reveals $f_{ij1}=0$, then firm $j$ is removed from consideration and the remaining firms are simply considered in index order.  If instead the initial inspection of $j$ reveals $f_{ij1}=1$, then the only change is that firm $j$ remains in the consideration set and would be inserted into the inspection order according to its index (or as an outside option). This is where we use the inductive hypothesis that our policy uses an index rule: since the indices of other firms are not affected by this change, all other firms are searched in the same order up to the point at which firm $j$ is inserted.  In either case, following the initial inspection of firm $j$, the subsequent policy makes the same inspections, in the same order, up until the next time that firm $j$ is inspected or transacted with. We emphasize that an observation that $f_{ij1}=0$ does not terminate the search procedure, it only removes further inspection of firm $j$ from consideration.

As a result, the only payoff impact of delaying the inspection of firm $j$ is that, in the event that the policy does not reach the step in which firm $j$ is to be inspected further or transacted with, the consumer saves a payment of $c_{j1}$ from inspecting $j$ the first time, which is a weak improvement.

We conclude that, in all cases, the optimal policy performs no better than a persistent policy.  Thus, by Lemma~\ref{lem:persistent}, our Algorithm~\ref{alg:index.generalized} is optimal for all problem instances.  This completes the proof of Theorem~\ref{thm:optimal.search}.

\section{Proofs from Section~\ref{sec:convergence}: Feedback and Learning}

\subsection{Projection of a History onto a firm}

The first step in proving Proposition~\ref{prop:market.learning} is Lemma~\ref{lem:indep.signals} below. It states that the posterior beliefs $(\priorf_{ijk})$
depend only on the screen realizations from firm $j$ in $\hist_{i-1}$, and not any other information in the history.

\begin{definition}
    For a history realization $H$, the \emph{projection of $H$ onto firm $j$}, written $H|_j$, contains all the screen realizations for firm $j$ present in $H$. The projection $H|_j$ preserves the order in which these data appear in $H$, but not the specific rounds from $H$ in which they were observed.
\end{definition}

We can think of $H|_j$ as $r$ sequences of binary variables: one for each level of screen.  Element $t$ of each sequence corresponds to the $t$th time that screen was revealed in the history, either through inspection or post-transaction feedback.

\begin{lemma}
\label{lem:indep.signals}
For any finite history realization $\hist_{i-1} = H$ of market $\mM$ and any firm $j$, 
posterior beliefs $\priorf_{ijk}(H)$ are determined by $H|_j$.
\end{lemma}
\begin{proof}
    Fix any partial history realization $\hist_{i-1} = H$ and suppose consumer $i$ with value $V_i$ arrives in round $i$. The consumer's choice of which firms to inspect, and in which order, is fully determined by $c$, $V_i$, and by $(\priorf_{ijk}(H_{i-1}))_{jk}$, so in particular does not reveal any further information about $\fitprob_{jk}$ beyond what is in $H_{i-1}$.  The updated posterior beliefs $\fitprobest_{jk}(H_i)$ are therefore determined entirely by the previous beliefs plus the inspection outcomes and post-transaction feedback.  As only the inspection outcomes and post-transaction feedback for firm $j$ are correlated with $\fitprob_{jk}$, the updated posteriors are determined entirely by information contained in $H|_j$.
\end{proof}

\subsection{Completing the Proof of Proposition~\ref{prop:market.learning}}
\begin{proof}
Consider the following procedure for generating an infinite history.  We first construct, for each firm $j$ and screen $k \in [r]$ in advance, an infinite sequence $\Omega_{jk}$ of independent Bernoulli signals of mean $\fitprob_{jk}$.  Then, whenever a screen realization is needed during the course of infinite history $H$, it is pulled from the corresponding pre-determined sequence $\Omega_{jk}$.  Note then that, for each $j$, the projected history $H|_j$ consists of prefixes of the streams $\Omega_{j1}, \dotsc, \Omega_{jr}$.

With probability $1$, the empirical average of each (infinite) stream $\Omega_{jk}$ converges to its true mean $\fitprob_{jk}$ in the limit over its sequence, so in particular a posterior estimate of the mean will converge to $\fitprob_{jk}$ over infinitely many observations.  We will condition on this event, which is determined entirely by the sequences $\Omega_{jk}$ and occurs with probability $1$.

Since all screen probabilities are bounded away from $0$, any firm that is inspected infinitely often will be transacted with infinitely often and will therefore generate, as part of the history, a signal from each of the $r$ screens infinitely often. Thus, by Lemma~\ref{lem:indep.signals}, for any firm $j$ that is inspected infinitely often, with probability $1$ the posterior beliefs converge to the means from which these signals are drawn.  That is, $\priorf_{ijk}$ converges to $\fitprob_{jk}$ in the limit as $i \to \infty$ for each $k$. This implies that any firm that is inspected infinitely often is learned.

It remains to consider firms that are inspected only finitely often; say $A$ is the set of such firms.  If $j \in A$ then there is a time after which $j$ is no longer selected in the optimal consumer search process. Let $T$ be the maximum such time over all $j \in A$. We claim that $\priorf_{ijk}(H) = \priorf_{ijk}(H_T)$ for all $k \in [r]$, all $j \in A$ and all $i \geq T$.  This follows because each such $j$ is not searched at any time $t > T$, so Lemma~\ref{lem:indep.signals} implies that $\priorf_{ijk}(H_{i-1}) = \priorf_{ijk}(H_T)$ for all $i > T$, and hence the limits exist and equal the claimed quantities.

Finally, since $\fitprob_{jk} < 1$ for all $j$ and $k$, any firm $j$ for which $\sigma_j((\fitprobest_{ijk}), \overline{V}) > 0$ would be inspected by consumer $i$ with positive probability bounded away from $0$.  As $(\fitprobest_{ijk})$ does not change until firm $j$ is inspected, this implies that firm $j$ will be inspected in some round at or after round $i$ with probability $1$.
Thus, if $j \in A$, then there must exist some $i$ for which $\sigma_j((\fitprobest_{ijk}), \overline{V}) \leq 0$, as otherwise $j$ is inspected infinitely often.  This implies that every firm $j \in A$ is lost, as required.
\end{proof}

\subsection{Proof of Proposition~\ref{prop:lost}}
\begin{proof}
For any market $\mM$ and firm $j$, let $H|_j^{\mM}$ denote the projection of $\hist$ onto firm $j$.  Then the marginal distribution over $H|_j^{\mM}$, the projection of history $\hist$ onto firm $j$, has an especially simple form.
Each inspection/feedback outcome in $H|_j^{\mM}$ is an independent draw from a Bernoulli distribution of mean $\pi_{jk}$, for some $k$. In particular, $H|_j^{\mM}$ and $H|_{j'}^{\mM}$ are independent for any two firms $j$ and $j'$.

Let $(\fitprobest_{jtk})$ denote the posterior mean estimates held by a consumer given the first $t$ entries of $H|_j^{\mM}$. Note that this is not the same as the estimates held by consumer $t$.
Let $\omega_j$ be the probability that there exists any $t$ such that $\sigma_{ij}((\fitprobest_{jtk}), \overline{V})\leq 0$. Then, by Lemma~\ref{lem:indep.signals}, $\omega_j$ is precisely the probability that $j$ is lost in history $\hist$. And since $\omega_j$ depends only on the marginal distribution over $H|_j^{\mM}$, which are independent across firms, we conclude that each firm is lost independently with  probability $\omega_j$.

Finally, recall that $\sigma_{ij}((\fitprobest_{jtk}), \overline{V})$ is weakly increasing in $\overline{V}$ and each $\fitprobest_{jtk}$, and weakly decreasing in $c$.  Since firm  $j$ is lost when this quantity becomes negative, we therefore have that the loss probability $\omega_j$ is weakly \emph{decreasing} in $\overline{V}$ and each $\fitprobest_{jtk}$, and weakly \emph{increasing} in $c$.  Since random variables $\fitprobest_{jtk}$ are positively correlated with $\fitprob_{jk}$, respectively, we have that $\omega_j$ is weakly decreasing in $\fitprob_{jk}$ as well. %
\end{proof}

\subsection{Proof of Lemma~\ref{lem:util.converge}}
\begin{proof}
We will condition on the event of Proposition~\ref{prop:market.learning} that every firm that is inspected infinitely often is learned and all others are lost, which occurs with probability $1$.  For each lost firm $j \not\in \learnedS$, there is some finite time at which the firm is lost.  As there are only finitely many firms, there is a fixed time $t_0$ after which no further firms are lost.  We will restrict attention to consumers who arrive to the market following round $t_0$, so that only the firms in $\learnedS$ are in the consideration set of each consumer.\footnote{Note that fixing the value of $t_0$ and the set $\learnedS$ imposes conditions on the progression of the market dynamics.  However, the only property of the stochastic market dynamics that we will use in the remainder of the proof is the assumed convergence of posterior estimates to true means, which we have already conditioned upon.}

Let $U(V; \mM, H)$ denote the optimal utility obtained by a consumer with value $V$ in market $\mM$, following history $H$.
Define
\[ \tilde{U}(V, H) := \E[ U(V; \mM, H)\ |\ H ]\]
to be the expected utility obtained by a consumer with value $V$ given prior history $H$, evaluated with respect to the consumer's posterior beliefs over $\mM$ given $H$.  In other words, $\tilde{U}(V,H)$ is the consumer's expectation about their utility given their belief about $\mM$.

Recall that Algorithm~\ref{alg:index} optimizes expected utility over all search orders (i.e., permutations) of the firms. Moreover, for any given search order, the expected utility obtained by consumer $i$ given a past history $H$ is continuous in $\fitprobest_{jk}(H)$ for all $j$ and $k$.  Since $\tilde{U}(V,H)$ is a maximum over the utilities obtained from a finite set of search orders, we conclude that $\tilde{U}(V,H)$ is continuous in $(\fitprobest_{jk}(H))_j$ as well.
Thus, since $\lim_{i\to \infty}\fitprobest_{ijk} = \fitprob_{jk}$ for all $j \in \learnedS$ by Proposition~\ref{prop:market.learning}, and since each consumer that arrives after round $t_0$ only inspects firms in $\learnedS(\hist)$, we conclude that
\[ \lim_{i \to \infty}\tilde{U}(V, \hist_{i-1}) = U^*(\learnedS). \]
In other words, the consumer beliefs about their anticipated utility converge over time to the static comparator $U^*(\learnedS)$.

It remains to relate these consumer beliefs to the true expected utility obtained by the consumer.
Convergence of posterior mean estimates implies that $\lim_{i\to \infty}|\fitprobest_{ijk} - \fitprob_{jk}| = 0$.  Continuity of $\tilde{U}(V, H)$ in $\fitprobest_{jk}(H)$ then implies
\[ \lim_{i \to \infty}| \tilde{U}(V, \hist_{i-1}) - U(V, \mM, \hist_{i-1})| = 0. \]
A standard sandwich argument therefore implies
\[ \lim_{i \to \infty} U(V, \mM, \hist_{i-1}) = \lim_{i \to \infty} \tilde{U}(V, \hist_{i-1}) = U^*(\learnedS) \]
as required.
\end{proof}

\section{Proofs from Section~\ref{sec:compare}: Improved Consumer Search}

\subsection{Proof of Theorem~\ref{thm:cheaper.search.general}}
\begin{proof}
    We will couple the histories of $\mM$ and $\tmM$ as follows. First, we draw the infinite sequence of consumer values $V_1, V_2, \dotsc$, which will be used for both histories.  Second, for each firm $j$ and screen $k \in [r]$, we independently draw an infinite sequence of screen realizations, each an independent Bernoulli random variable with mean $\pi_{jk}$.  Write $\Omega_{jk}$ for each such sequence, and $\Omega_j$ for the collection of $r$ value sequences for firm $j$.

    To generate history $\hist$, each consumer $i$ will have value $V_i$ and will engage in their consumer optimal search process.  Whenever screen $k$ of a firm $j$ is inspected, say for the $t$'th time, the outcome is taken to be the $t$'th element of the sequence of screen realizations from $\Omega_{jk}$. This fully determines the progression of history $\hist$.

    Write $H|^{\mM}_j$ for the projection of history $\hist$ onto firm $j$. By construction, $H|^{\mM}_j$ consists of prefixes of the sequences in $\Omega_j$. By Proposition~\ref{prop:market.learning}, almost surely each firm is either learned or lost, and the projected history $H|^{\mM}_j$ is finite (and a strict prefix of $\Omega_j$) if firm $j$ is lost, but infinite (and equal to $\Omega_j$) if firm $j$ is learned.

    We generate history $H^{\tmM}$ in the same way, noting that the consumer optimal search process may resolve differently due to the modified search costs. Nevertheless, it is still the case that $H|^{\tmM}_j$ is a prefix of $\Omega_j$, and equal to $\Omega_j$ if firm $j$ is learned.

    It will now be convenient to talk about posterior estimates for a firm $j$ evaluated after a certain number of searches of that firm have occurred.  For a given $z \geq 1$, consider the $z$th round in which firm $j$ is inspected at least once.  A feature of Algorithm~\ref{alg:index} is that the set of signals generated for firm $j$ is fully determined by $\Omega_j$: a realization for each screen will be revealed, in sequence, until the first $0$ is encountered (if any).  In particular, the set of signals revealed is independent of the round index $i$, the consumer value $V_i$, the depth of inspection $L$, %
    and whether or not a transaction occurs with firm $j$.  In particular, this means that the projected histories $H|^{\mM}_j$ and $H|^{\tmM}_j$ will agree on the information revealed about firm $j$ during the first $z$ rounds in which firm $j$ was inspected.
    
    We now write $\fitprobest_{jk}(H, z)$ to denote $\fitprobest_{jk}(H_t)$ where $t$ is the round in which firm $j$ is searched for the $z$th time in history $H$.  Then,
    since $H|^{\mM}_j$ and $H|^{\tmM}_j$ agree up to the minimum of their lengths, by Lemma~\ref{lem:indep.signals} we have that $\fitprobest_{jk}(\hist, z) = \fitprobest_{jk}(H^{\tmM}, z)$ for all $z$ such that firm $j$ is searched at least $z$ times in each history.

    By Proposition~\ref{prop:market.learning}, firm $j$ is lost in market $\tmM$ if and only if there exists $z$ such that
    \begin{align}\label{eq:lost.1}
        \tilde{\sigma}_{ij}((\fitprobest_{jk}(H^{\tmM},z)), \overline{V}) \leq 0.
    \end{align}
    But now, considering whether $j$ is lost in history $\hist$ of market $\mM$, note that if firm $j$ is searched at least $z$ times in $\hist$ then \eqref{eq:lost.1} would imply
    \begin{align*}
\sigma_{ij}((\fitprobest_{jk}(H^{\mM},z)), \overline{V})
& = \sigma_{ij}((\fitprobest_{jk}(H^{\tmM},z)), \overline{V})\\
& \leq \tilde{\sigma}_{ij}((\fitprobest_{jk}(H^{\tmM},z)), \overline{V})\\
& \leq 0
    \end{align*}
    where the first equality follows because
    $\fitprobest_{jk}(H^{\tmM},z) = \fitprobest_{jk}(H^{\mM},z)$ for each $j$ and $k$ from our coupling, the second inequality follows because costs are pointwise lower in $\tmM$ and this can only increase the index (recalling that the depth $L$ of search is fixed and identical in $\mM$ and $\tmM$), and the third inequality is \eqref{eq:lost.1}.  Thus, if firm $j$ is lost in history $H^{\tmM}$, then firm $j$ is lost in history $\hist$ as well.  We conclude that $\learnedS(\tilde{M}) \supseteq \learnedS(\mM)$ under our coupling.

    Finally, since $\learnedS(\tmM) \supseteq \learnedS(\mM)$ and $\tilde{c}_k \leq c_k$ for all $k$, each consumer performing optimal search in the steady-state of $\tmM$ has only a larger set of learned firms to inspect, and a reduced cost of each such search.  Thus, any chosen search order for learned firms in $\mM$ could have been selected in $\tmM$ as well at a weakly improved expected utility.  Optimizing over all potential search orders, we conclude that the long-run expected consumer utility in $\tmM$ under history $H^{\tmM}$ is therefore only greater than that of $\mM$ under $H^{\mM}$, yielding $U(\tmM) \geq U(\mM)$ over our coupling.
\end{proof}

\subsection{Proof of Theorem~\ref{thm:cheaper.search.inform}}
\begin{proof}
{Suppose $\mM$ has inspection depth $k$ and $\tmM$ has depth $\tilde{k}>k$.
We will use the screen-stream coupling from the proof of Theorem~\ref{thm:cheaper.search.general}, with the same sequence of consumer values.
Whenever a firm is inspected, inspection and transaction feedback together reveal its screens up to the first failure (or all $r$ passes), regardless of inspection depth.
Thus, after the same number of visits to a firm, both markets have observed the same screen outcomes.}

{Let $p_{j\ell}(z)$ be the posterior mean for screen $\ell$ after the first $z\geq 0$ such visits, including the prior at $z=0$, and put
\[
a_j(z):=\prod_{\ell\leq k}p_{j\ell}(z),
\qquad
b_j(z):=\prod_{k<\ell\leq\tilde{k}}p_{j\ell}(z).
\]
These posteriors are defined from the streams even if that prefix is never reached in one market.
The indices at value $\oV$ and the two inspection depths satisfy
\begin{align*}
\sigma_j^{(k)}(z)&=\oV\prod_{\ell>k}p_{j\ell}(z)-P_j-\frac{c}{a_j(z)},\\
\sigma_j^{(\tilde{k})}(z)&=\oV\prod_{\ell>\tilde{k}}p_{j\ell}(z)-P_j-\frac{c}{a_j(z)b_j(z)}.
\end{align*}
Consequently,
\begin{equation}\label{eq:lost.2}
\sigma_j^{(k)}(z)
=b_j(z)\sigma_j^{(\tilde{k})}(z)-(1-b_j(z))P_j.
\end{equation}
Since $P_j\geq 0$ and $0<b_j(z)\leq 1$, a non-positive index at depth $\tilde{k}$ implies one at depth $k$.
Firms whose parameters are all known a priori are learned in both markets.
Otherwise, on the probability-one event of Proposition~\ref{prop:market.learning}, if $j$ is learned in $\mM$, then every finite prefix is reached and $\sigma_j^{(k)}(z)>0$ at every prefix.
Equation~\eqref{eq:lost.2} therefore forces $\sigma_j^{(\tilde{k})}(z)>0$ at every prefix, so $j$ cannot be lost in $\tmM$ and must be learned there.
Thus $\learnedS(\tmM)\supseteq\learnedS(\mM)$ under this coupling.

For the utility comparison, fix a consumer value $V$ and an optimal policy on the learned firms of $\mM$.
We consider the following (possibly suboptimal) policy for $\tmM$. In $\tmM$, ignore any additional learned firms and emulate the policy from $\mM$ until it would transact with a firm $j$.
First reveal the additional screens $k+1,\ldots,\tilde{k}$, at zero additional cost.
If they all pass, make the same transaction; otherwise leave without transacting.
This is a feasible policy at depth $\tilde{k}$.
Coupling the consumer's screen outcomes, the two policies incur the same search costs.
If an additional screen fails, the original transaction would not have met the consumer's need, so avoiding its price $P_j\geq 0$ weakly improves realized utility; otherwise utility is unchanged.
Optimizing in $\tmM$ and averaging over $V$ therefore gives $U(\tmM)\geq U(\mM)$ on the same history coupling.}
\end{proof}

\subsection{Proof of Proposition~\ref{prop:inform.bad}}
\begin{proof}
{Take a single firm with price $P=0$, consumers with value $V=1$, and search cost $c=1/5$.
In $\mM$, quality is known to be $Q=1/2$, and the fit probability $\pi$ is either $1/100$ or $99/100$ with equal prior probability.
Fix the market realization $\pi=99/100$.
In $\tmM$, rescale fit by $2/3$ and quality by $3/2$: quality is known to be $\tilde{Q}=3/4$, the fit prior has equally likely values $1/150$ and $33/50$, and the realized fit probability is $33/50$.
All screen probabilities are strictly between $0$ and $1$, and $\tmM$ is more consumer-informative with $\delta=3/2$.
One can view this as moving a known $2/3$-pass screen into inspection, leaving a $3/4$-pass quality screen.

In $\mM$, after $s$ fit successes and $f$ fit failures, the likelihood ratio of high to low fit is $99^{s-f}$, so the posterior mean is
\[
\hat{\pi}(d)=\frac{1}{100}+\frac{98}{100}\frac{99^d}{1+99^d},
\qquad d=s-f.
\]
The index $1/2-c/\hat{\pi}(d)$ is positive exactly when $d\geq 0$.
Quality feedback does not affect beliefs because $Q$ is known.
Thus loss occurs when the success-minus-failure walk first hits $-1$.
In the realized market, this walk steps up with probability $99/100$ and down with probability $1/100$, so its probability of ever hitting $-1$ is $1/99$.\footnote{{For $\rho=1/99$, stopping the martingale $\rho^d$ when the walk hits $-1$ or $N$ gives the lower-boundary hitting probability $(1-\rho^N)/(\rho^{-1}-\rho^N)$; letting $N\to\infty$ gives $\rho$.}}
The probability of learning the firm is therefore $p=98/99$.

In $\tmM$, the initial fit mean is $1/3$, giving positive index $3/4-c/(1/3)=3/20$.
If the first inspection fails, Bayes' rule gives
\begin{align*}
\Pr(\tilde{\pi}=33/50\mid\text{failure})&=51/200,\\
\E[\tilde{\pi}\mid\text{failure}]&=2599/15000<4/15.
\end{align*}
The posterior mean is below $c/\tilde{Q}=4/15$, so the index becomes negative and the firm is lost.
In the realized market, the first inspection fails with probability $17/50$.
Hence the probability $\tilde{p}$ of learning the firm satisfies $\tilde{p}\leq 33/50<98/99=p$.

If the firm is learned, the optimal limiting policy inspects it, and its expected utility in either market is
\[
u_*=\frac{99}{100}\frac{1}{2}-\frac{1}{5}
=\frac{33}{50}\frac{3}{4}-\frac{1}{5}
=\frac{59}{200}>0.
\]
If the firm is lost, limit utility is zero.
Thus $U(\mM)=u_*\mathbf{1}\{1\in\learnedS(\mM)\}$, and analogously for $\tmM$.

Finally, draw $Z\sim\mathrm{Uniform}[0,1]$, prescribing learning in $\mM$ when $Z\leq p$ and in $\tmM$ when $Z\leq\tilde{p}$.
Conditional on these indicators, sample a history from each market's correct conditional distribution.
This gives a coupling of histories with the correct marginals.
Since $\tilde{p}<p$, it satisfies $\learnedS(\tmM)\subseteq\learnedS(\mM)$ and $U(\tmM)\leq U(\mM)$ almost surely.
Both comparisons are strict on $\{\tilde{p}<Z\leq p\}$, which has positive probability $p-\tilde{p}$.}
\end{proof}

\section{Numerical Simulations}
\label{sec:simulations}

\pricingnew{This appendix provides the full specification and additional diagnostics for the simulations in Section~\ref{sec:expts-body}.}\footnote{All simulations were run on a laptop with an Intel Core i7-11370H processor (4 cores, 3.3 GHz) and 32 GB RAM.}

We simulate a market with $m = 15$ firms and $r = 5$ screens per firm. 
Probabilities $\fitprob_{jk}$ are drawn independently from $\text{Beta}(4, 2)$, clipped to $[0.25, 0.95]$, yielding a skewed distribution with mean approximately $0.67$. 
Platform posterior distributions are initialized with independent (unclipped) $\text{Beta}$ priors.\footnote{This maintains tractability of Bayesian posterior estimates. This is consistent with our theoretical model, recalling that our results hold pointwise over market realizations (see Remark~\ref{rem:pointwise}).}
Prices $P_j$ are drawn from $\mathcal{N}(0.08, 0.04^2)$ clipped to $[0.01, 0.20]$. Consumer values $V_i$ are drawn i.i.d.\ from $\text{Uniform}(0.5, 2.0)$, giving $\bar{V} = 2.0$. These parameters were chosen to generate sufficient variability in the resulting market; we test some alternative choices at the end of this section.

Each simulation run has $T = 5000$ consumers unless otherwise specified. We report averages over $R = 50$ independent runs (with firm realizations and other sources of randomness redrawn in each run). All shaded regions and error bars in plots represent 95\% confidence intervals on the means.

{We refer to firms not lost by $T$ as \emph{surviving}; the plot label ``learned'' denotes this finite-horizon count, not verified convergence.
Utility trajectories show windowed realized averages, whereas panels labeled ``steady-state utility'' report terminal welfare estimates, not demonstrated limit outcomes.}

\paragraph{Impact of Search Cost.} We %
{study a fit-quality model with transcripts as in Definition~\ref{defn:fit-quality-transcripts},} fix the search informativeness to $k = 3$ and vary the search cost $c \in \{0.005, 0.01, 0.02, 0.035, 0.05\}$.  See Figure~\ref{fig:cost_dynamics} for the resulting dynamics and Figure~\ref{fig:cost_outcomes} for final outcomes at the end of the simulation. {Lower costs reduce belief error and firm losses and raise consumer utility over the simulated horizon.
At low costs ($c \leq 0.01$), nearly all firms survive to $T=5{,}000$, and belief error is small.} As costs increase beyond $c \approx 0.02$, incorrect firm losses rise sharply. At $c = 0.05$, most firms are lost.

We distinguish between \emph{correct} loss (where the index evaluated at the true screen probabilities is truly negative) versus \emph{incorrect} loss (where the true index is positive and hence the loss is due to inaccurate posterior estimates).
In all settings where losses occur, the majority of the observed losses are incorrect, driven by unlucky streaks of screen failures in initial inspections. As shown in Figure~\ref{fig:cost_outcomes}, higher costs lead to more correct losses (orange) but substantially more incorrect losses (red). %

\begin{figure}[h]
\centering
\begin{subfigure}[b]{0.45\textwidth}
\includegraphics[width=\textwidth]{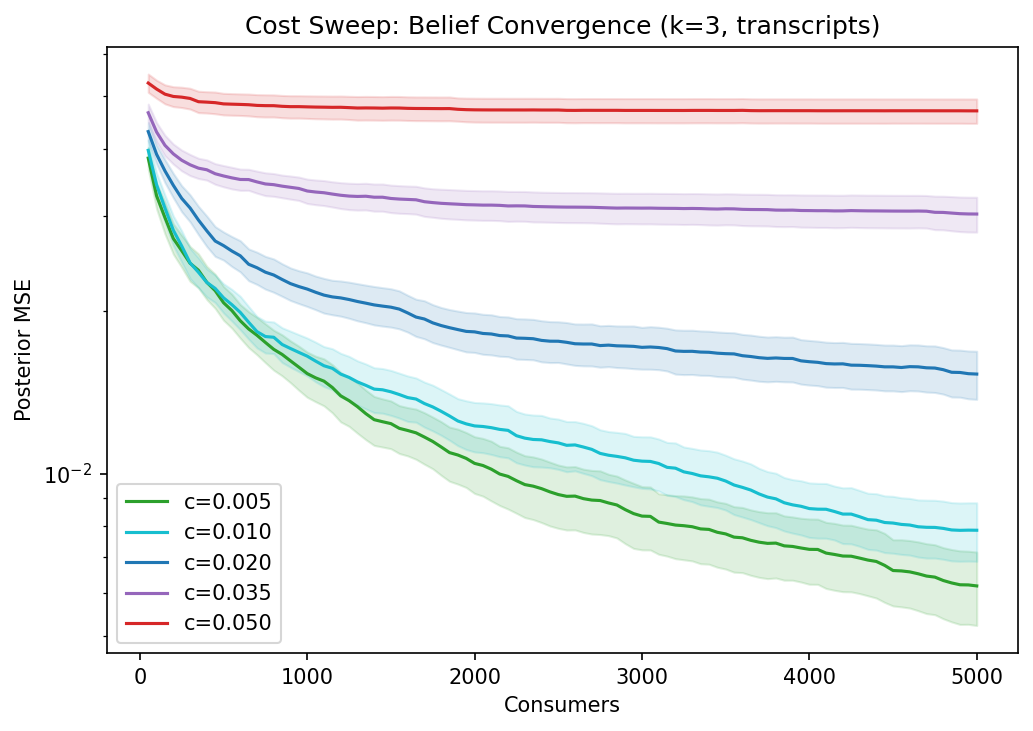}
\caption{Mean squared error of belief estimates.}
\end{subfigure}
\hfill
\begin{subfigure}[b]{0.45\textwidth}
\includegraphics[width=\textwidth]{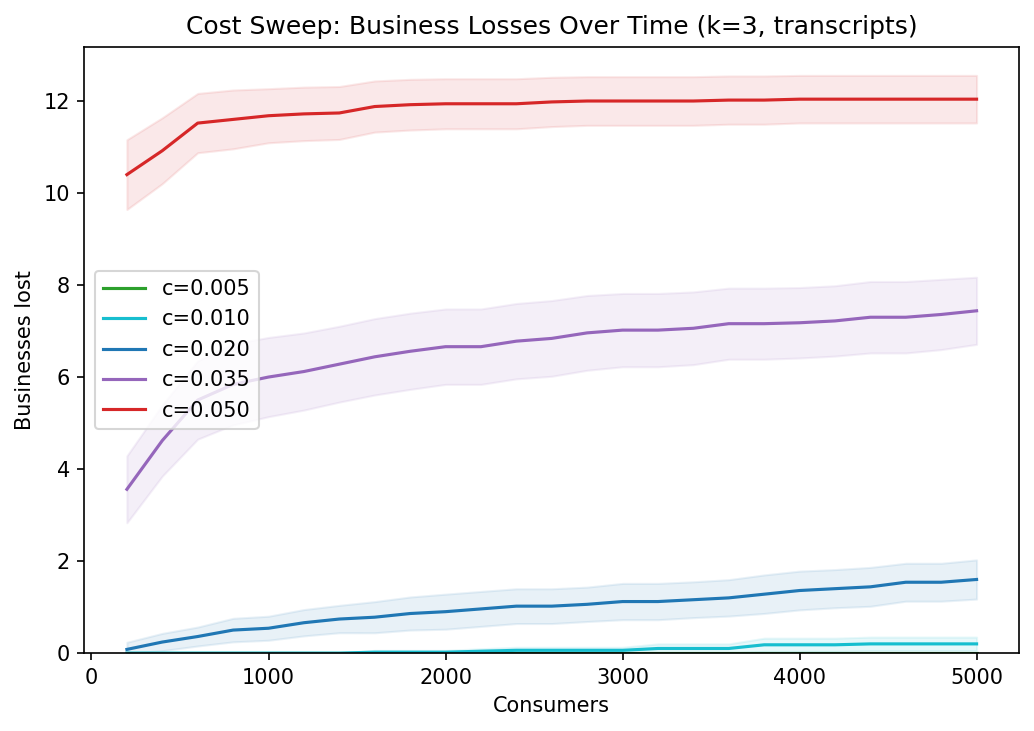}
\caption{Cumulative firms lost over time.}
\end{subfigure}\\
\vspace{1em}
\begin{subfigure}[b]{0.45\textwidth}
\includegraphics[width=\textwidth]{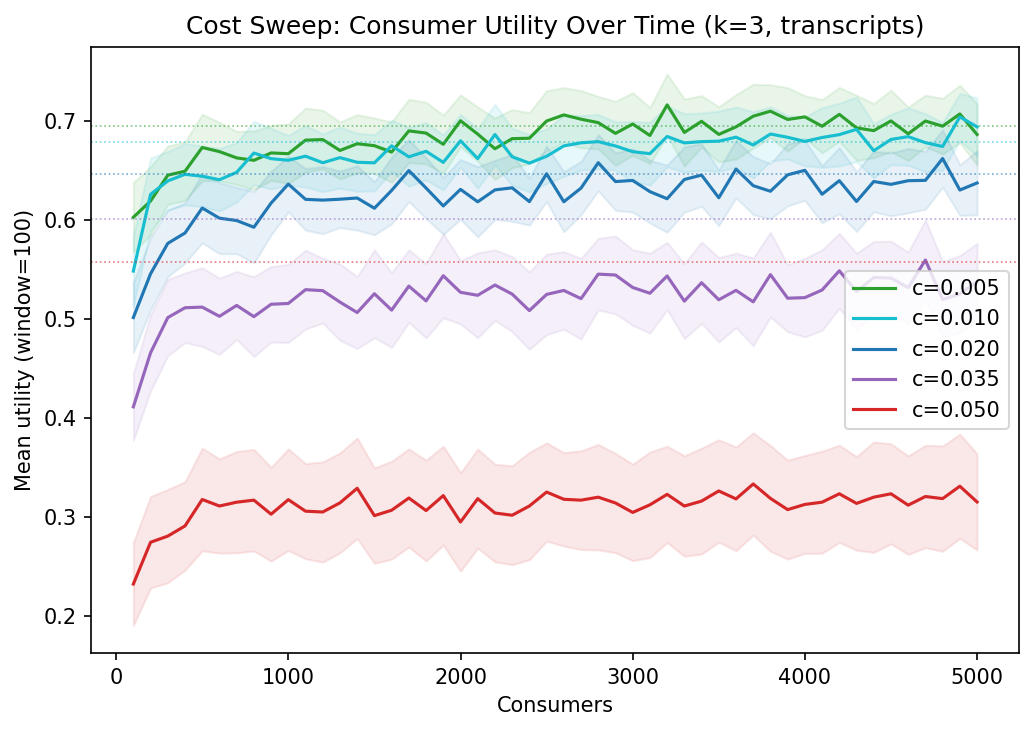}
\caption{Consumer utility trajectory (window = 100).}
\end{subfigure}
\caption{Cost sweep dynamics. (a) MSE declines rapidly for low costs but stalls when firms are lost. (b) firm losses accumulate early and plateau; higher costs yield more losses. {(c) Windowed realized utility rises early and then fluctuates around a plateau. Horizontal lines show expected optimal consumer utility with known screen probabilities and no firms lost, at each cost level.}}
\label{fig:cost_dynamics}
\end{figure}

An illustrating example for incorrect loss, showing the evolution of the index $\sigma_{j}$ for the maximum consumer type $\overline{V}$ (i.e., $\sigma_{j}\rbr{\pihatvec_{ij}, \oV}$) for a single firm $j$ over a sequence of 50 runs, is given in Figure~\ref{fig:drift_trace}. The firm is lost for those traces where this index falls below zero, indicated in red. {Surviving traces approach the true index over the plotted horizon.} This plot uses the same realization of fit probabilities for all runs, so that the true index remains constant.

\begin{figure}[h]
\centering
\begin{subfigure}[b]{0.85\textwidth}
\includegraphics[width=\textwidth]
{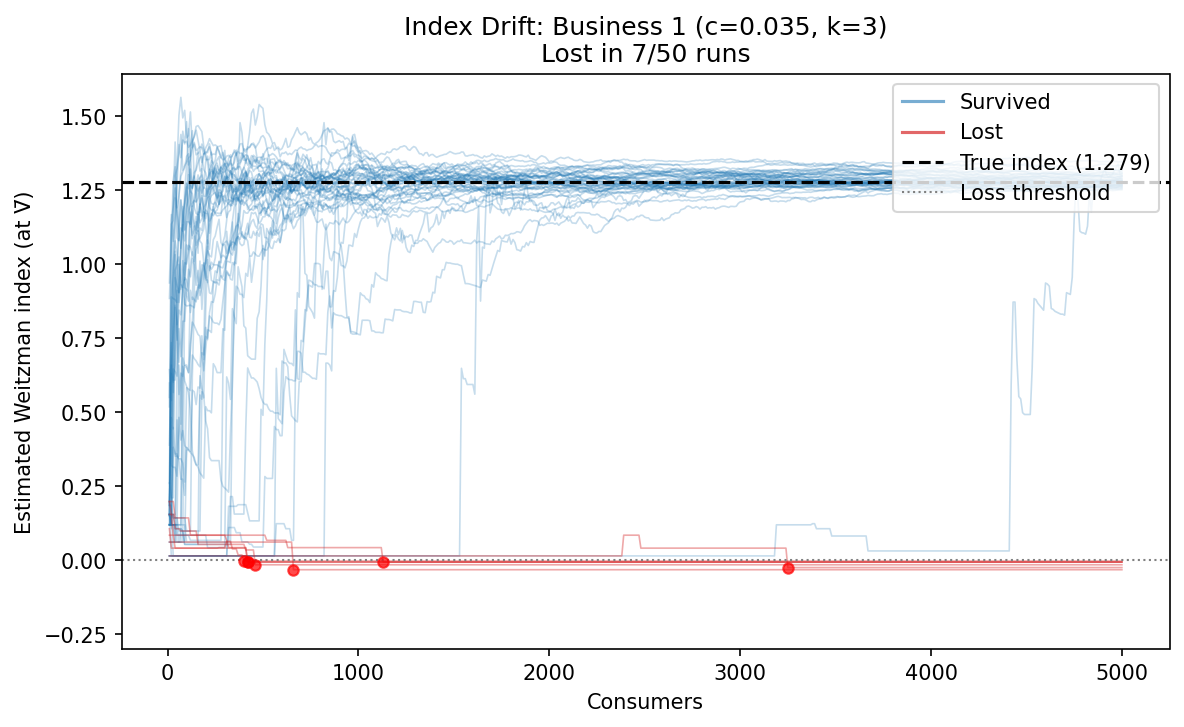}
\end{subfigure}
\caption{{Surviving and lost trajectories for an example firm.} Each trace plots a trajectory of $\sigma_{j}\rbr{\pihatvec_{ij}, \oV}$ (see Proposition~\ref{prop:market.learning}). All runs use the same realizations of $\pi_{jk}$ for all $j$ and $k$. Red traces indicate runs where the business has been lost by round 5000.}
\label{fig:drift_trace}
\end{figure}

\paragraph{Impact of Search Informativeness.}

We next fix cost $c = 0.02$ and vary informativeness $k \in \{1, 2, 3, 4, 5\}$.  See Figure~\ref{fig:info_dynamics} for the resulting dynamics and Figure~\ref{fig:info_outcomes} for the final outcomes at the end of the simulation. {Terminal welfare estimates rise from $0.31$ at $k=1$ to $0.92$ at $k=5$.}
{In contrast, the number of surviving firms after $5{,}000$ rounds is non-monotone, peaking at $k=4$ and declining slightly at $k=5$.} At $T=2000$ this non-monotonicity is more pronounced, with the average number of lost firms minimized at $k=3$ (0.90 lost) but higher for $k=4$ and $k=5$ (1.18 and 1.58 lost, respectively). 

\begin{figure}[h]
\centering
\begin{subfigure}[b]{0.45\textwidth}
\includegraphics[width=\textwidth]{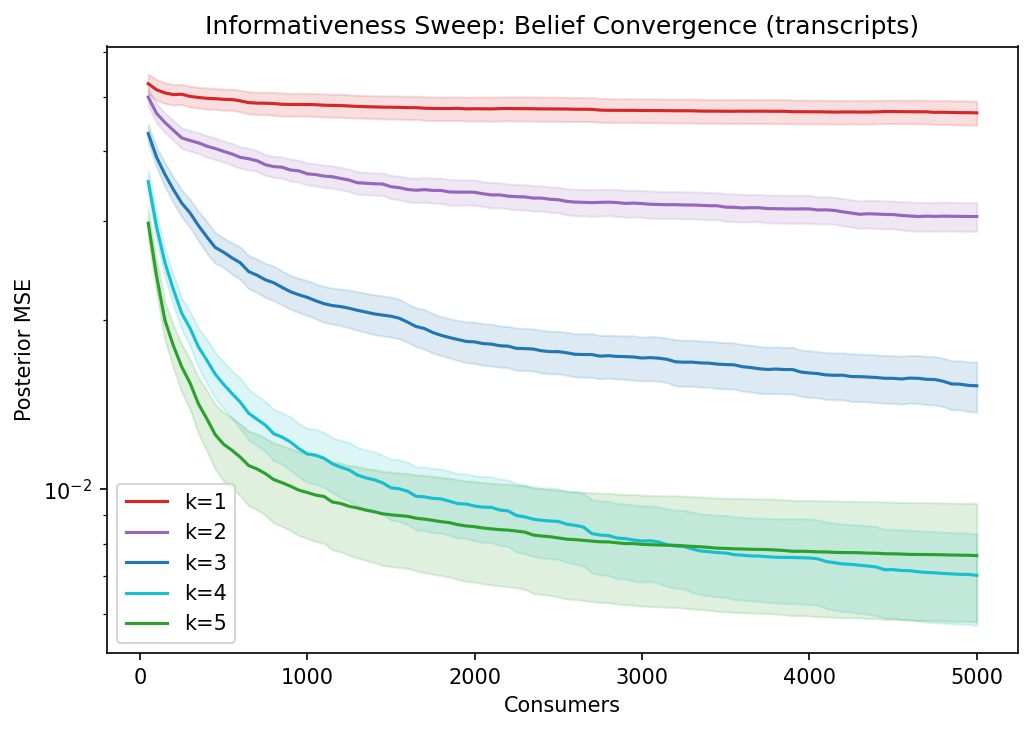}
\caption{Mean squared error of belief estimates.}
\end{subfigure}
\hfill
\begin{subfigure}[b]{0.45\textwidth}
\includegraphics[width=\textwidth]{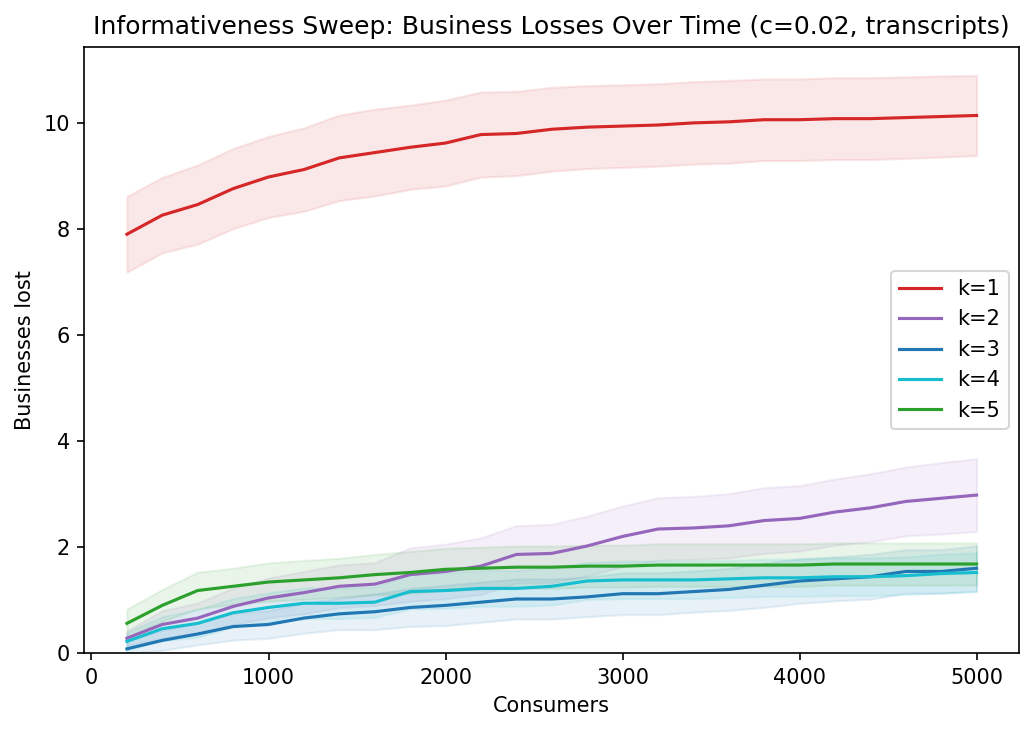}
\caption{Cumulative firms lost over time.}
\end{subfigure}\\
\begin{subfigure}[b]{0.45\textwidth}
\includegraphics[width=\textwidth]{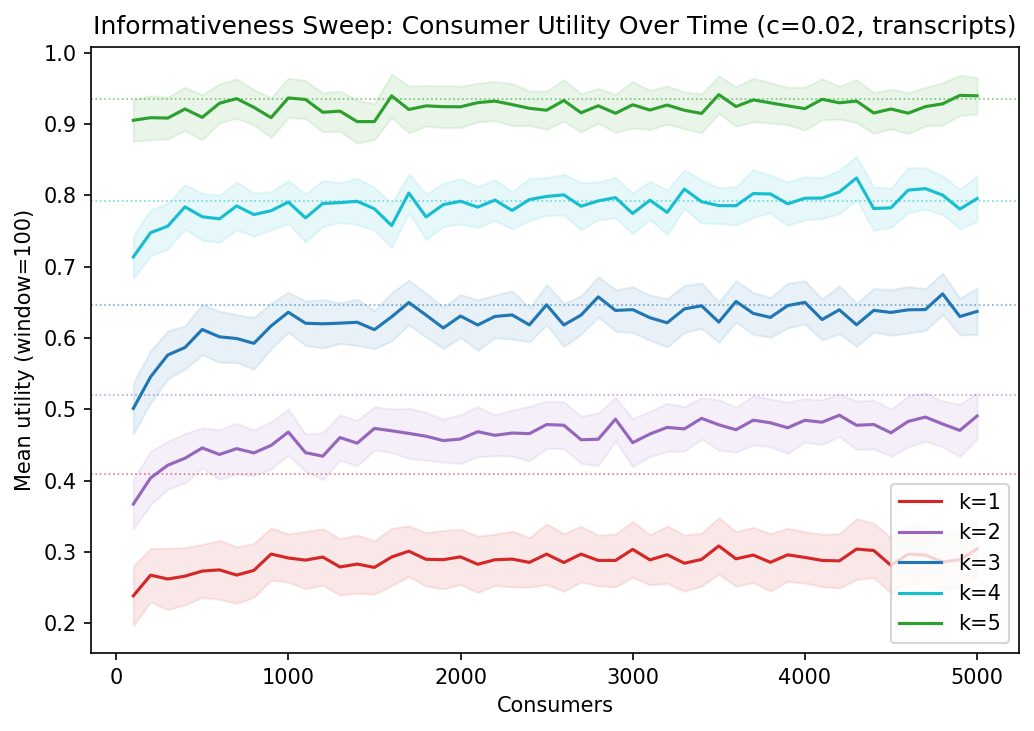}
\caption{Consumer utility trajectory (window = 100).}
\end{subfigure}
\caption{Informativeness sweep dynamics. {(a) Belief error generally falls faster at higher $k$.} (b) Loss trajectories show qualitatively different temporal patterns across $k$ values: high-$k$ losses occur early and plateau, while low-$k$ losses accumulate gradually throughout the horizon. {(c) Higher $k$ yields higher windowed realized utility over the plotted horizon. Horizontal lines show expected optimal consumer utility with known screen probabilities and no firms lost, at each informativeness level.}}
\label{fig:info_dynamics}
\end{figure}

\begin{figure}[h]
\centering
\begin{subfigure}[b]{0.85\textwidth}
\includegraphics[width=\textwidth]{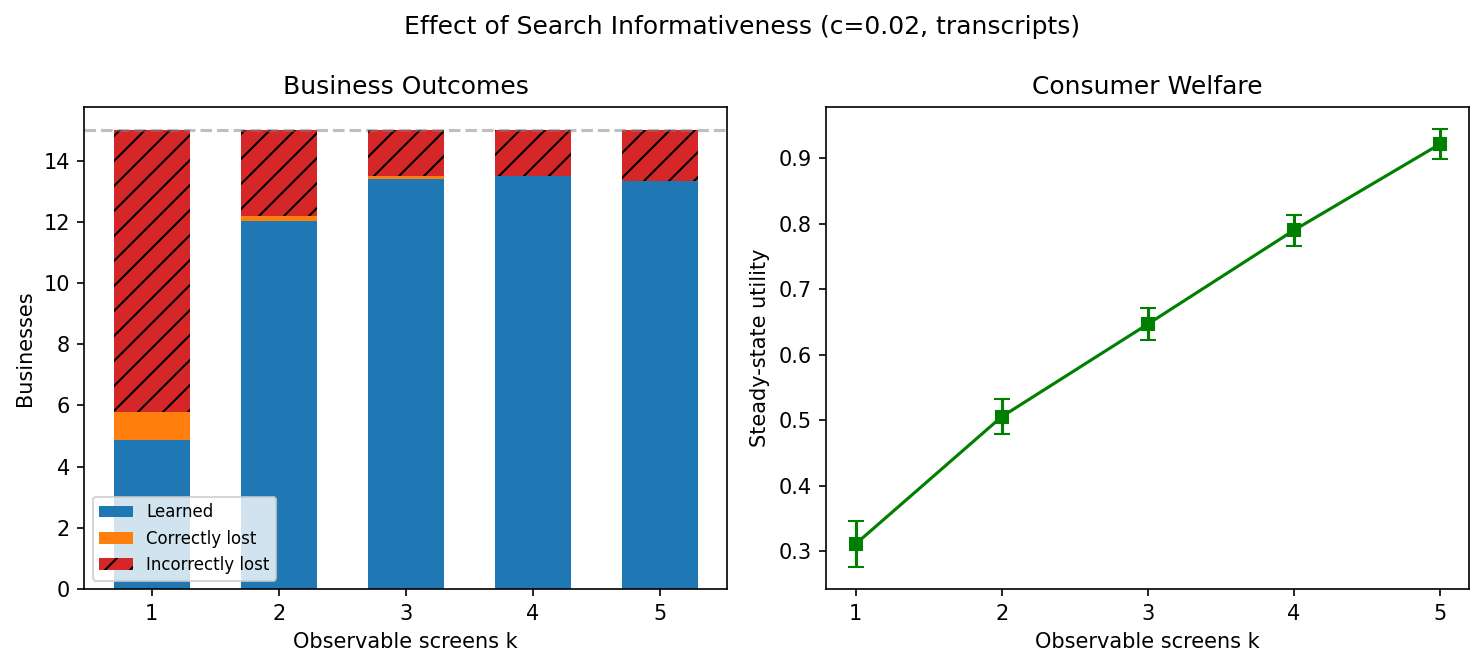}
\end{subfigure}
\caption{{Informativeness sweep outcomes at $T=5{,}000$: surviving firms and losses (left), and terminal welfare estimates (right).}}
\label{fig:info_outcomes}
\end{figure}

{This finite-horizon non-monotonicity does not contradict Theorem~\ref{thm:cheaper.search.inform}, which compares limit outcomes.
To investigate it, we fixed a market realization and examined an individual firm lost less frequently at $k=3$ than at $k=5$ (Figure~\ref{fig:drift_traces}).
The plot covers $T=2{,}000$ rounds: most surviving traces at $k=5$ are close to the true index, whereas several at $k=3$ remain near the loss threshold.
This suggests a transient effect: losses accumulate more slowly under less informative search.
Belief error generally falls faster at higher $k$ (Figure~\ref{fig:info_dynamics}(a)); we conjecture that the slower accumulation of losses reflects reduced sensitivity of the index to posterior changes when $k$ is low.}

\begin{figure}[h]
\centering
\begin{subfigure}[b]{\textwidth}
\includegraphics[width=\textwidth]{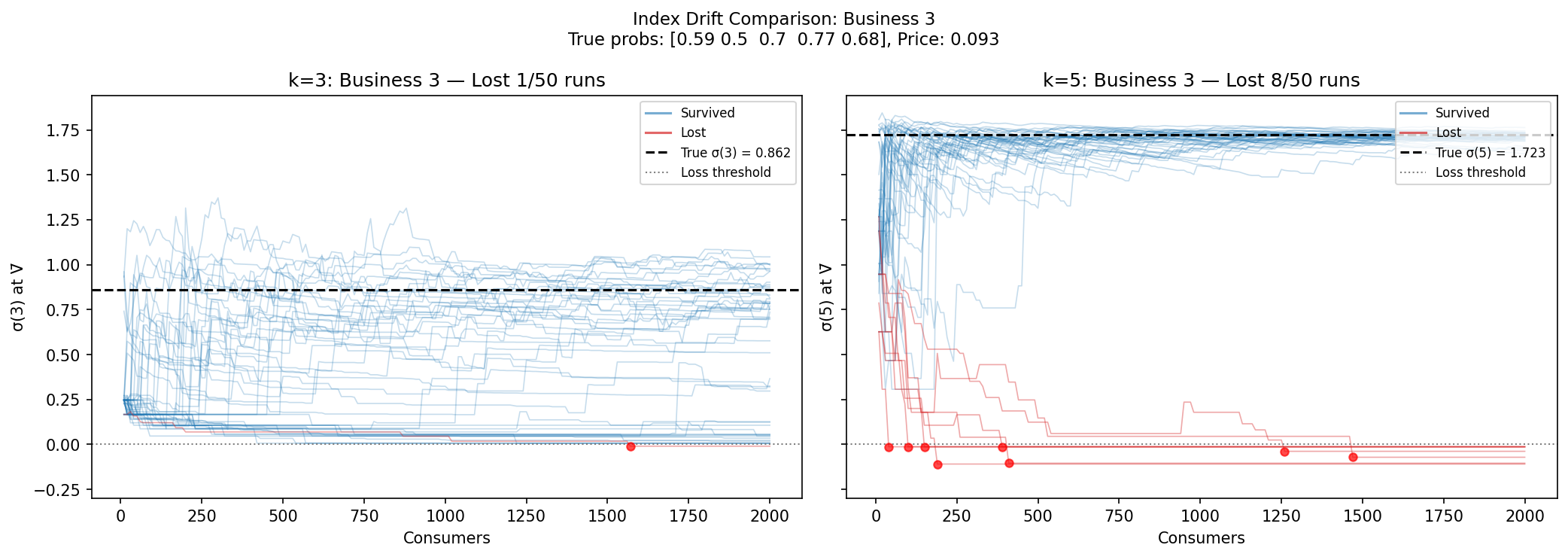}
\end{subfigure}
\caption{Weitzman index drift traces for a sample firm at $k = 3$ (left panel) and $k = 5$ (right panel). {Blue traces have not been lost by $T=2{,}000$; red traces are lost, with red circles marking the crossing point. At $k=5$, losses occur early and most surviving traces approach the true index rapidly. At $k=3$, several remain near the loss threshold, suggesting the possibility of further losses beyond the plotted horizon.}
}
\label{fig:drift_traces}
\end{figure}

{To investigate this finite-horizon explanation, we compare $k=3$ and $k=5$ over $T=20K$ rounds (Figure~\ref{fig:info_dynamics_extended}).
The average loss count for $k=3$ overtakes that for $k=5$, reversing the earlier ordering.}

\begin{figure}[h]
\centering
\begin{subfigure}[b]{0.95\textwidth}
\includegraphics[width=\textwidth]
{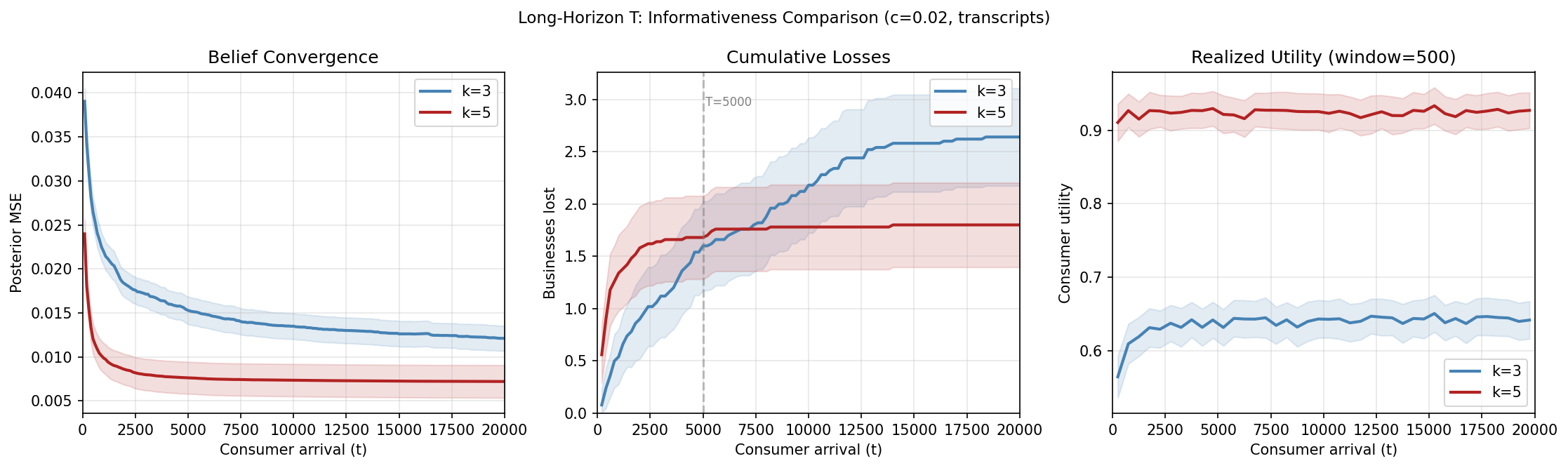}
\end{subfigure}
\caption{{Extended transcript comparison ($T=20{,}000$). Posterior MSE is shown on the left. The average loss count for $k=3$ overtakes that for $k=5$ (center); realized utility (right) uses a 500-consumer window.}}
\label{fig:info_dynamics_extended}
\end{figure}

\paragraph{Impact of Transcripts.}
Finally, we compare the transcript (T) and no-transcript (NT) information structures across levels of search informativeness. Recall that the NT scenario corresponds to a fit-quality market in which the platform only observes screen pass/fail outcomes collectively, rather than for each screen separately, maintaining an aggregate posterior over the success probability for each firm. For each fixed $k$, this aggregate posterior is maintained via Bayesian updating on the Bernoulli signals received from the collective pass/fail outcomes.\footnote{To align with our analytical comparisons (Theorem~\ref{thm:cheaper.search.inform} and Proposition~\ref{prop:inform.bad}), we keep the same initialization of the platform posterior estimates as we vary $k$, so that only consumer informativeness changes across treatments.}
See Figure~\ref{fig:t_vs_nt_dynamics}. {At $k=5$ and $T=5{,}000$, transcripts leave 13.3 surviving firms on average versus 6.7 under NT, with terminal welfare estimates of 0.92 versus 0.76.}

\begin{figure}[h]
\centering
\begin{subfigure}[b]{0.48\textwidth}
\includegraphics[width=\textwidth]{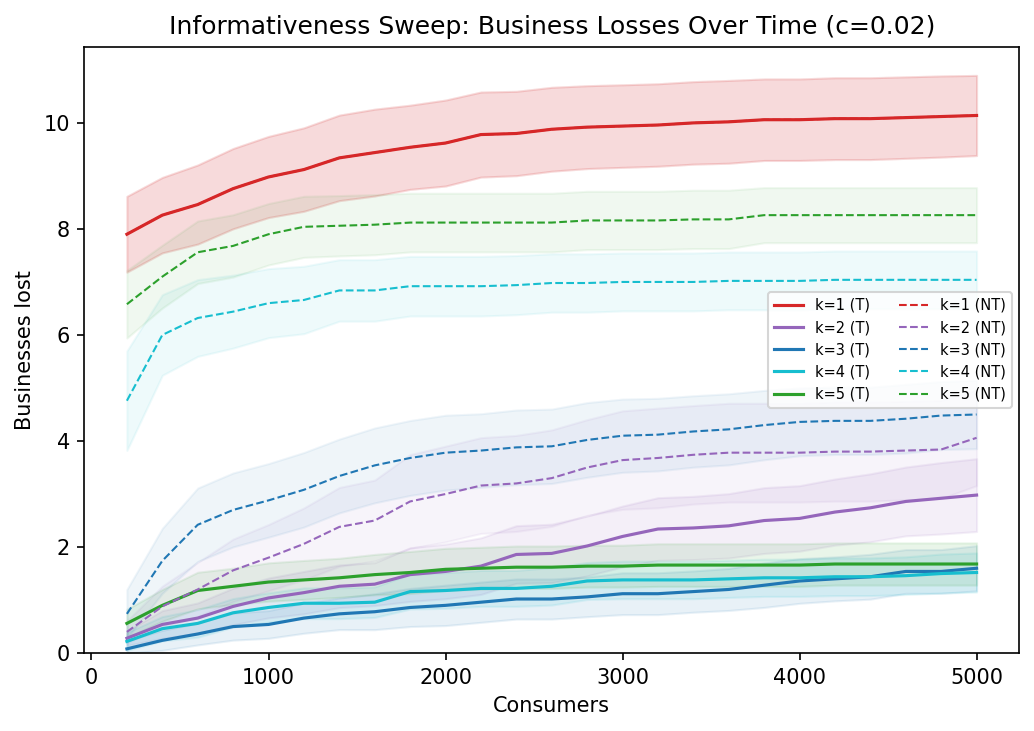}
\caption{firms lost over time (T and NT).}
\end{subfigure}
\hfill
\begin{subfigure}[b]{0.48\textwidth}
\includegraphics[width=\textwidth]{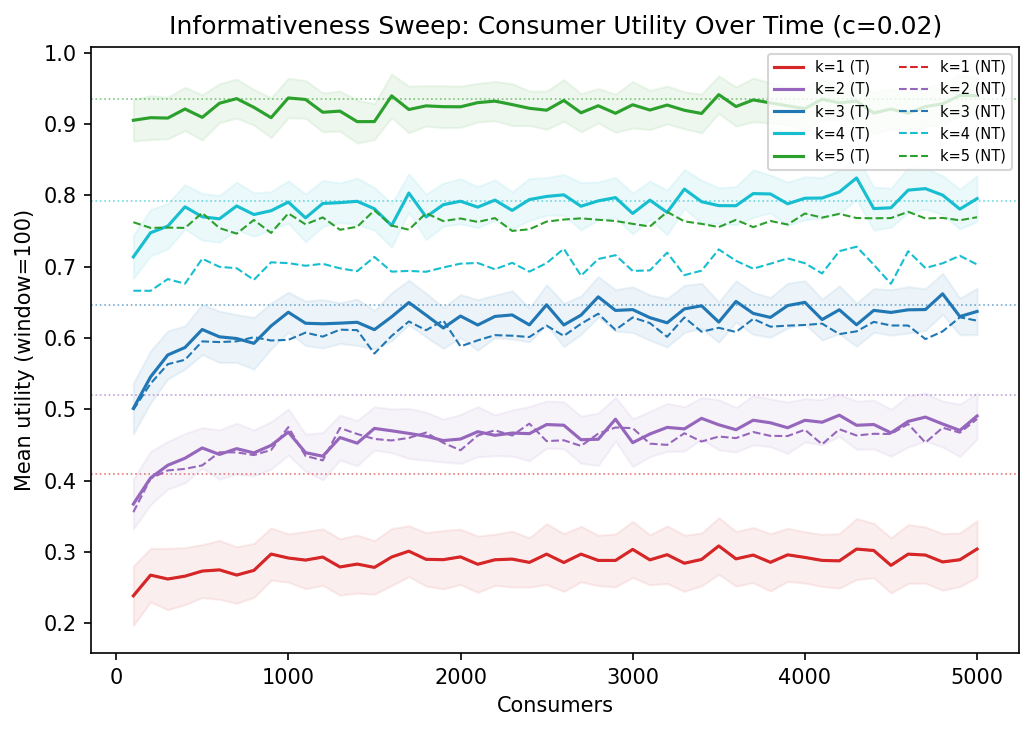}
\caption{Consumer utility trajectory (T and NT).}
\end{subfigure}\\
\vspace{1em}
\begin{subfigure}[b]{0.95\textwidth}
\includegraphics[width=\textwidth]{info_sweep_outcomes.png}
\caption{{Terminal survival, loss, and welfare estimates ($T=5{,}000$).}}
\end{subfigure}
\caption{Informativeness sweep dynamics comparing transcript (solid) and no-transcript (dashed) modes. (a) NT produces substantially more firm losses, particularly at higher $k$. {(b) Both modes show windowed realized utility increasing in $k$, with an advantage for T at larger $k$.}}
\label{fig:t_vs_nt_dynamics}
\end{figure}

Why do transcripts help? Under NT, the aggregate signal ``all $k$ screens passed'' has probability $\prod_{\ell=1}^k \pi_{j\ell}$, which for typical firms is quite low (approximately $0.67^k$). At $k = 5$, this is roughly $0.13$: the platform observes a success only 13\% of the time, receiving a failure signal in the remaining 87\% of visits that is uninformative about \emph{which} screen failed. The aggregate posterior therefore converges slowly and with high variance.
{As aggregate successes become rarer at higher $k$, the gap between feedback regimes grows in this sweep.}
{Surviving-firm counts fall from 10.5 at $k=3$ to 6.7 at $k=5$ under NT, while T remains roughly constant (13.4 to 13.3).}

{We also extend the NT simulation to $T=20{,}000$ (Figure~\ref{fig:info_dynamics_extended_nt}).
Unlike the transcript comparison, the ordering of losses across $k=3,4,5$ persists through this extended horizon.}

\begin{figure}[h]
\centering
\begin{subfigure}[b]{0.95\textwidth}
\includegraphics[width=\textwidth]
{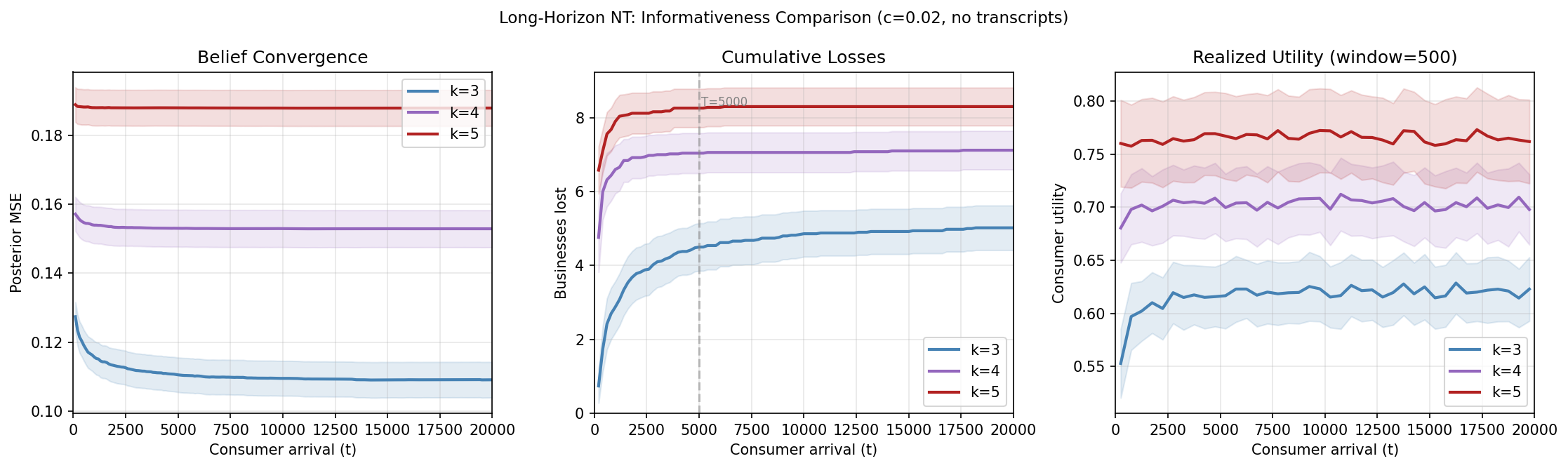}
\end{subfigure}
\caption{{Extended no-transcript comparison ($T=20{,}000$, $k=3,4,5$). Posterior MSE (left) and cumulative losses (center) remain higher at higher $k$ through this horizon; realized utility (right) uses a 500-consumer window.}}
\label{fig:info_dynamics_extended_nt}
\end{figure}

\paragraph{Robustness Checks}
{We test the cost-sweep gains, utility gains from informativeness with transcripts, and non-monotone NT losses under five market variations.} Each variant modifies one aspect of the baseline parameterization while holding all others fixed.

\begin{enumerate}
\item \textbf{More businesses ($m = 30$):} Doubles the number of firms, reducing per-business observation frequency.
\item \textbf{Fewer businesses ($m = 5$):} Reduces competition, concentrating observations.
\item \textbf{Symmetric screen probs (Beta(2,2)):} Replaces the skewed Beta(4,2) with a symmetric distribution (mean 0.5).
\item \textbf{Homogeneous values ($V = 1.25$):} {Eliminates value heterogeneity} by fixing $V$ at the population mean.
\item \textbf{Higher prices ($P_j \sim \mathcal{N}(0.15, 0.04^2)$):} Nearly doubles average prices, squeezing transaction margins.
\end{enumerate}

{Terminal welfare estimates and surviving-firm counts are non-increasing in cost in all five variants (Figure~\ref{fig:robust_cost}).
In the informativeness sweep (Figure~\ref{fig:robust_info}), welfare estimates increase in all five variants, while survivor counts are non-monotone in two ($m=30$ and $m=5$).
These finite-horizon inversions are consistent with the transient pattern in the baseline experiment.
The larger reversal at $m=30$ may reflect fewer observations per firm; the $m=5$ difference is small relative to the plotted uncertainty.}

\begin{figure}[h]
\centering
\includegraphics[width=\textwidth]{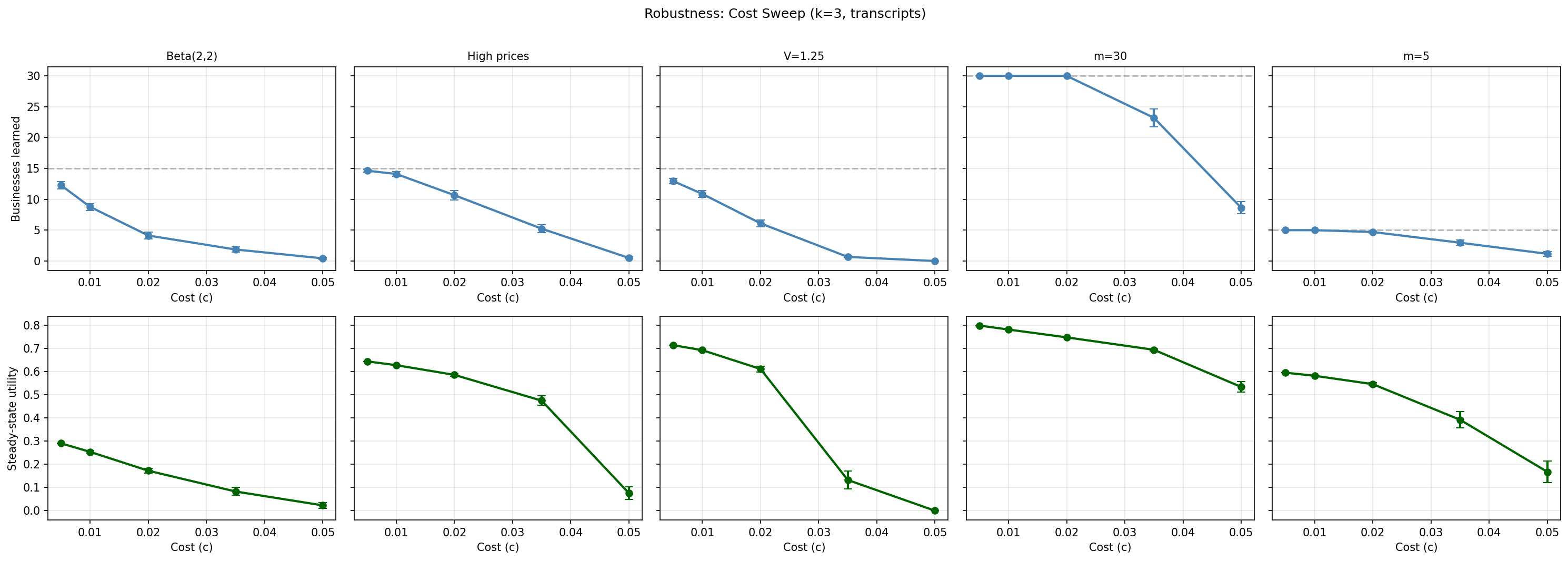}
\caption{{Robustness of cost sweep results across five market variants ($k=3$, transcripts, $T=5{,}000$). Top: surviving firms; bottom: terminal welfare estimates. Both are non-increasing in cost in all variants.}}
\label{fig:robust_cost}
\end{figure}

\begin{figure}[h]
\centering
\includegraphics[width=\textwidth]{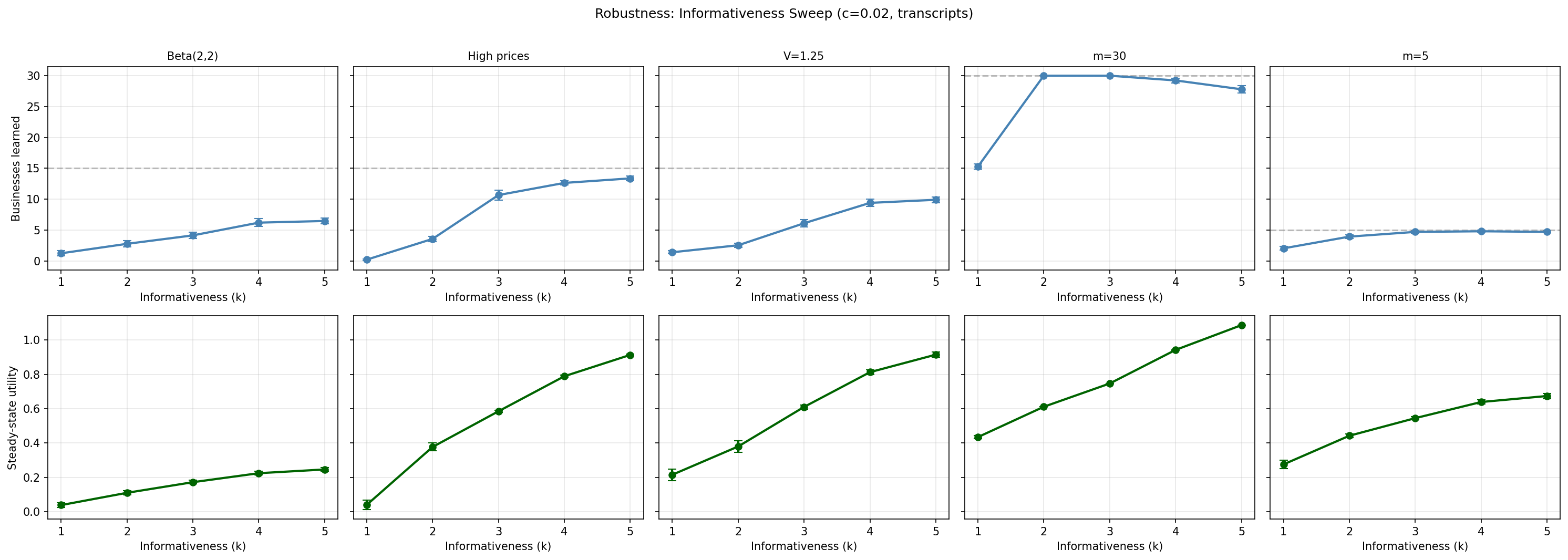}
\caption{{Robustness of informativeness sweep results across five market variants ($c=0.02$, transcripts, $T=5{,}000$). Terminal welfare estimates (bottom) increase in all variants. Surviving-firm counts (top) are monotone in three variants but decline at high $k$ for $m=30$ and $m=5$.}}
\label{fig:robust_info}
\end{figure}

We separately verify the robustness of the finding that businesses lost is non-monotone in consumer informativeness in the absence of transcripts. We run the informativeness sweep ($c = 0.02$, $k \in \{1, \ldots, 5\}$) in NT mode for each of the five variants above. See Figure~\ref{fig:nt_robust}. We find that all five variants exhibit a U-shaped relationship between consumer informativeness and business losses in the NT case.  In four of the five scenarios this non-monotonicity is statistically significant at $p < 0.05$, but it falls within sampling noise for the variant with higher prices.
{Welfare estimates are also non-monotone in some NT variants (bottom row), unlike the baseline sweep.}

\begin{figure}[h]
\centering
\includegraphics[width=\textwidth]{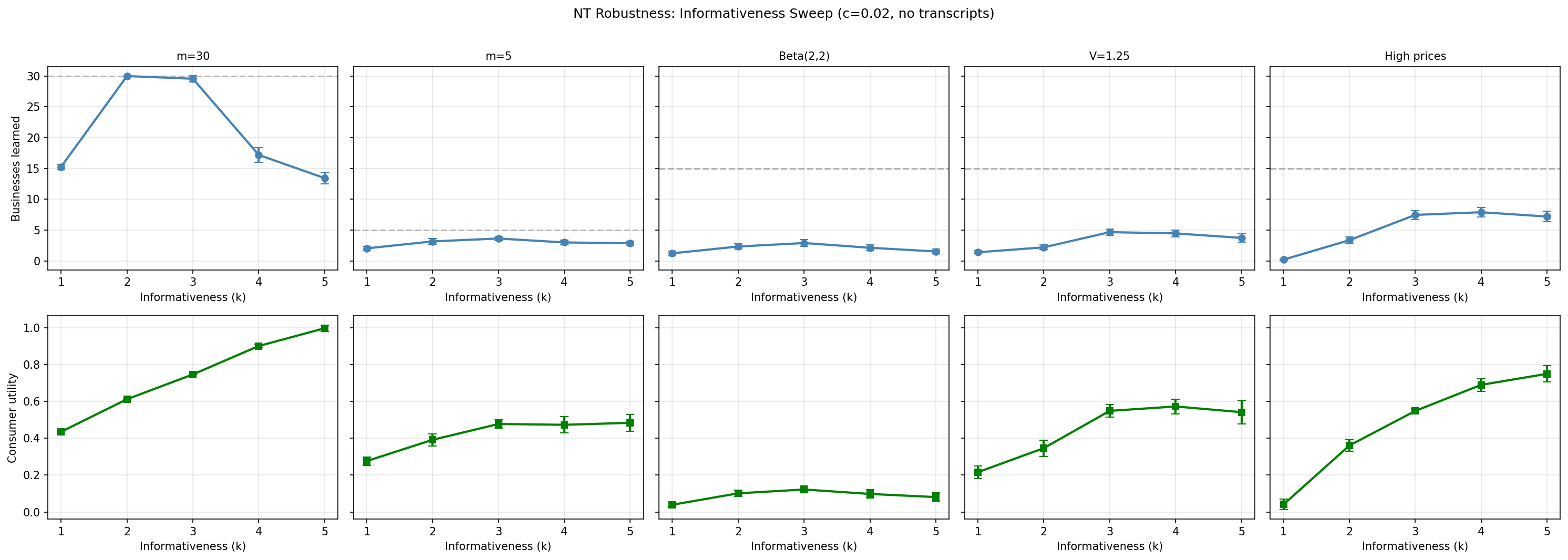}
\caption{{Robustness of informativeness sweep results without transcripts ($c=0.02$, $T=5{,}000$). Surviving-firm counts (top) first increase and then decrease in all five variants, most sharply for $m=30$. Terminal welfare estimates (bottom) need not be monotone in $k$.}}
\label{fig:nt_robust}
\end{figure}

\clearpage
\section{\pricingnew{Proofs for Endogenous Inspection Depth}}
\label{sec:depth.proofs}
\begin{pricingrevision}
Throughout this appendix, prices are fixed, feasible depths are $\mathcal L=\{\Lmin,\ldots,L\}$, and feedback includes the screen-level inspection and transaction signals of the baseline model. All probabilities of screen success are strictly between zero and one. We use the same posterior-consistency argument as in the proof of Proposition~\ref{prop:market.learning}. Expectations over subsequent outcomes are evaluated at the fixed realized market, unless explicitly described as subjective expectations under current posterior beliefs.

\subsection{Proof of Theorem~\ref{thm:depth.optimal}}
\begin{proof}
Conditional on a committed depth $\ell$, this is exactly the problem of Theorem~\ref{thm:optimal.search} with $L=\ell$. Since depth is chosen before observing inspection outcomes, maximizing its optimal conditional payoff over the finite set $\mathcal L$ is optimal.

For completeness, write
\[
b_{j\ell}=\prod_{k\leq\ell}p_{jk},\qquad
q_{j\ell}=\prod_{k>\ell}p_{jk},\qquad
C_{j\ell}=\sum_{k\leq\ell}c_k\prod_{h<k}p_{jh}.
\]
These are the subjective probability of completing inspection, the conditional probability of subsequent satisfaction, and the expected inspection cost. For a fixed value $V$ and depth $\ell$, order positive-index firms as $j_1,\ldots,j_n$ in decreasing index order. Set $W_{n+1}=0$ and evaluate backwards:
\[
W_t=-C_{j_t\ell}+b_{j_t\ell}(Vq_{j_t\ell}-P_{j_t})
+(1-b_{j_t\ell})W_{t+1},\qquad t=n,\ldots,1.
\]
Then $W_1$ is the subjective expected utility of Algorithm~\ref{alg:index}. Products, costs, and indices can be computed in $O(mr)$ operations for each depth; sorting costs $O(m\log m)$ and this recurrence costs $O(m)$. Enumerating feasible depths gives the stated bound. An empty positive-index list has payoff zero.
\end{proof}

\subsection{Proof of Proposition~\ref{prop:depth.learning}}
\begin{proof}
\emph{Posterior convergence.} Generate independent feedback streams for every firm and screen as in the proof of Proposition~\ref{prop:market.learning}. Whenever a firm is inspected at least once, inspection and transaction feedback together reveal its screens through the first failure, or all screens if they pass. This holds at every feasible depth: passing all pre-transaction screens leads to a transaction, and subsequent feedback identifies the first remaining failure. Thus choosing a depth affects when a firm's streams are sampled, but not the feedback generated by its next inspection.

Let $S$ be the set of firms inspected infinitely often and $A=[m]\setminus S$. Every screen of a firm in $S$ is observed infinitely often, since its preceding screens all pass with a fixed strictly positive probability. The feedback-stream consistency argument therefore gives weak convergence of each screen posterior to its true parameter. For firms in $A$, observations are finite and beliefs eventually freeze. Write $\boldsymbol p_i$ for the full array of posterior means before round $i$, and $\bar{\boldsymbol p}$ for its limit. Its coordinates equal true parameters on $S$ and frozen posterior means on $A$; all are interior.

\emph{Vanishing inspection probabilities.} For each deterministic subset $B\subseteq[m]$, let $E_i(B)$ be the event that consumer $i$ inspects some firm in $B$, and put
\[
q_i(B)=\Pr(E_i(B)\mid\hist_{i-1}).
\]
The conditional Borel--Cantelli lemma gives $\sum_i q_i(B)<\infty$ almost surely on the event that only finitely many $E_i(B)$ occur. Apply this simultaneously to all finitely many subsets $B$. On a realized path with $A=[m]\setminus S$, it follows that
\begin{equation}\label{eq:depth.vanishing}
q_i(A)\longrightarrow0.
\end{equation}
Here $q_i(A)$ means evaluation at the realized subset of the functions $q_i(B)$ defined for deterministic $B$; we have not conditioned on the future event determining $S$.

Let $\widehat q_i(B)$ be the subjective probability of $E_i(B)$ under current posterior means, using the same value distribution and the same consumer policy. A round reveals only finitely many binary outcomes. Since the realized parameters are interior and $\boldsymbol p_i$ converges to an interior array, the likelihood ratio of any possible within-round path under subjective beliefs to that under the realized market is bounded above by a finite path-dependent constant $K$, uniformly in the consumer value and policy for all sufficiently late rounds. Independence of the screen parameters under the posterior makes these path probabilities the products of the relevant posterior means and their complements. Consequently,
\begin{equation}\label{eq:depth.subjective.vanishing}
\widehat q_i(A)\leq Kq_i(A)\longrightarrow0.
\end{equation}

\emph{Equality of limiting optimal values.} Let $J(\boldsymbol p,V)$ be the optimal subjective payoff using all firms, with depth optimized, and let $J_S(\boldsymbol p,V)$ be its counterpart restricted to $S$. For each depth, Theorem~\ref{thm:optimal.search} permits restriction to persistent policies specified by an ordered subset of firms. There are finitely many such policies over all depths. Each fixed policy's expected payoff is a polynomial in the screen probabilities and affine in $V$. Hence both value functions are continuous, and convergence in $\boldsymbol p$ is uniform over $V\in[0,\oV]$.

Every realized payoff is bounded in absolute value by, for example,
\[
M=1+\oV+\max_jP_j+m\sum_{k\leq L}c_k.
\]
Modify the consumer's optimal policy by stopping without transacting just before its first inspection of a firm in $A$. The modified policy uses only $S$, and differs from the original only on $E_i(A)$. Therefore
\begin{equation}\label{eq:depth.value.gap}
0\leq\E_V[J(\boldsymbol p_i,V)-J_S(\boldsymbol p_i,V)]
\leq2M\widehat q_i(A)\longrightarrow0.
\end{equation}
Taking limits using uniform continuity gives a continuous nonnegative function
$J(\bar{\boldsymbol p},V)-J_S(\bar{\boldsymbol p},V)$ with zero integral under $\distV$. It must vanish at every point of the support of $\distV$: otherwise continuity gives a neighborhood with strictly positive gap and positive probability. Thus
\begin{equation}\label{eq:depth.limit.value}
J(\bar{\boldsymbol p},V)
=J_S(\bar{\boldsymbol p},V)
=\max_{\ell\in\mathcal L}U_\ell^S(V)
\quad\text{for every }V\in\operatorname{supp}(\distV).
\end{equation}
This equality compares all depths and all firms, not just the indices at depths optimal on $S$.

\emph{Exclusion at optimal depths.} Fix any $V\in\operatorname{supp}(\distV)$, depth $\ell\in\mathcal L^*(V;S)$, and firm $j\in A$. Suppose $\sigma_{j,\ell}(V;\bar{\boldsymbol p}_j)>0$. Follow an optimal persistent policy on $S$ at depth $\ell$ and, if it makes no transaction, inspect $j$ at that depth. The probability $a$ of reaching this additional search is strictly positive, because each searched firm has a strictly positive probability of failing inspection; if the policy searches no firms then $a=1$. The additional subjective expected payoff is
\[
a\left(\prod_{k\leq\ell}\bar p_{jk}\right)
\sigma_{j,\ell}(V;\bar{\boldsymbol p}_j)>0.
\]
This feasible policy strictly improves on the right-hand side of~\eqref{eq:depth.limit.value}, contradicting that equality. The nonpositive-index conclusion follows simultaneously for all supported values and all maximizing depths.
\end{proof}

\subsection{Proof of Lemma~\ref{lem:depth.utility}}
\begin{proof}
Work on the probability-one event established in the preceding proof, using its notation. Let $\alpha_i$ be the consumer's chosen policy, including its value-dependent choice of depth, and let $\alpha_i^S$ be the modification that stops just before the first inspection of a firm in $A$. For a policy $\alpha$, let $R_i(\alpha)$ be its true expected payoff, integrated over the next consumer's value and outcomes, and $\widehat R_i(\alpha)$ its subjective expected payoff at the same history. The payoff bound and~\eqref{eq:depth.vanishing}--\eqref{eq:depth.subjective.vanishing} imply
\[
|R_i(\alpha_i)-R_i(\alpha_i^S)|\leq2Mq_i(A),\qquad
|\widehat R_i(\alpha_i)-\widehat R_i(\alpha_i^S)|
\leq2M\widehat q_i(A).
\]
All firms used by $\alpha_i^S$ lie in $S$, where posterior means converge to the truth. There are finitely many persistent policies, including policies truncated before a specified firm, and their payoffs converge uniformly over values. Hence, even though $\alpha_i^S$ may vary with the value and history,
\[
\sup_{\alpha\text{ using only }S}
|R_i(\alpha)-\widehat R_i(\alpha)|\longrightarrow0,
\]
where it suffices to take the supremum over value-dependent selections of this finite policy collection. Combining these estimates yields
\[
R_i(\alpha_i)-\widehat R_i(\alpha_i)\longrightarrow0.
\]
Optimality gives $\widehat R_i(\alpha_i)=\E_V[J(\boldsymbol p_i,V)]$. By~\eqref{eq:depth.value.gap} and continuity of the learned-only value,
\[
\widehat R_i(\alpha_i)
\longrightarrow\E_V[J_S(\bar{\boldsymbol p},V)]
=\E_V\!\left[\max_{\ell\in\mathcal L}U_\ell^S(V)\right]
=U_{\mathrm{endo}}^*(S).
\]
Since $R_i(\alpha_i)=\E[U_i\mid\hist_{i-1}]$, this proves the claim without requiring convergence of the depth or the search order.
\end{proof}
\end{pricingrevision}

\section{\pricingnew{Endogenous Pricing and Competition}}
\label{sec:strategic.firms}

\begin{pricingrevision}
In this appendix we develop the endogenous-pricing extension of Section~\ref{sec:extension.pricing}, restricting attention to the single-cost inspection setting of Section~\ref{sec:search.informativeness}. We present this appendix as a self-contained and expanded version of the (more informal) material in Section~\ref{sec:extension.pricing}. 

We use the terms business and firm interchangeably. In a fit-quality market (Definition~\ref{defn:fit-quality}), a consumer pays $c$ to observe fit $F_{ij}$ and observes quality realization $q_{ij}$ only upon a transaction. The corresponding parameters are $\pi_j$ and $Q_j$. In a market with transcripts (Definition~\ref{defn:fit-quality-transcripts}), these quantities aggregate the screens before and after the inspection depth $k$:
\[
F_{ij}=\prod_{\ell\leq k}f_{ij\ell},\qquad
\pi_j=\prod_{\ell\leq k}\pi_{j\ell},\qquad
q_{ij}=\prod_{\ell>k}f_{ij\ell},\qquad
Q_j=\prod_{\ell>k}\pi_{j\ell}.
\]
Inspection costs $c$ once, irrespective of how many of the first $k$ screens are revealed. Feedback is aggregate in the fit-quality market and screen-level in the transcript variant, as in the main text. This appendix does not establish pricing results for general vectors of positive screen costs or endogenous inspection depth.

We retain the utility and timing conventions of Section~\ref{sec:model}: the price is paid whenever a transaction occurs, including an unsuccessful one; prices are observed before search; and consumers prefer not to inspect or transact when indifferent to the outside option. Fix a common rule for breaking ties between positive search indices. The notation $\priorF_j(H),\priorQ_j(H)$ denotes the posterior distributions of $\pi_j,Q_j$, with means $\fitprobest_j(H),\hat Q_j(H)$. With transcripts these are induced by the screen-level posteriors.

As in Section~\ref{sec:extension.pricing}, assume that $\distV$ has a continuous, strictly positive density on an interval $[\underline V,\oV]$, where $0\leq\underline V<\oV<\infty$.
\end{pricingrevision}
\pricingnew{We also impose regularity for the equilibrium comparisons below.}
A distribution with CDF $F$ and density function $f$ is regular if its virtual value function $\phi(v)=v-\tfrac{1-F(v)}{f(v)}$ is non-decreasing.
\pricingnew{These additional distributional assumptions are specific to the pricing analysis, not the baseline model. All screen probabilities remain strictly between zero and one.}

\subsection{\pricingnew{Pricing Model and Learning Dynamics}}
\label{sec:pricing.learning}

Throughout this section, each business is now a strategic agent that can adaptively adjust their price over time, {its utility defined as revenue.} A history $\hist$ now includes the prices $P_{ij}$ chosen by each business in each round $i$. A {(possibly randomized)} pricing rule $P_j$ for business $j$ takes as input a history $H$ and returns a price $P_j(H)$. {Thus, each consumer $i$ is offered a publicly-observable price $P_{ij}:=P_j(\hist_{i-1})$.}
{The market platform is a passive participant that collects history and disseminates it to the consumers.}
In particular, the platform is non-strategic and does not attempt to make inferences from the pricing decisions of businesses.

\pricingnew{Pricing rules use only the public history and randomization independent of private values, screen outcomes, and firm parameters conditional on that history. Firms have no additional private information. Thus prices themselves convey no further information about types, and Lemma~\ref{lem:indep.signals} continues to apply to the projection excluding prices. Algorithm~\ref{alg:index} reduces to ordering businesses by the positive indices}
\pricingnew{\[
\sigma_j(P,\fitprobest_j,\hat Q_j,V)=V\hat Q_j-P-c/\fitprobest_j.
\]}

{We argue about convergence without imposing} an equilibrium assumption on the businesses' choice of pricing rules. For example, some businesses might myopically maximize expected revenue each round, while others might maximize long-run time-discounted payoffs under some beliefs about other businesses' pricing strategies.

\begin{definition}
Pricing rule $P_j$ for business $j$ is \emph{demand-responsive} if, for each {history $H=\hist_{i-1}$, price $P_{ij}=P_j(H)$} is chosen so that $\sigma_j(P_{ij},\fitprobest_j(H),\hat Q_j(H),\oV)>0$ if any such $P_{ij}\geq0$ exists.
\end{definition}

Demand-responsiveness requires that, if possible, the business sets a price such that \emph{some} consumer would prefer to search the business than take the outside option. Indeed, since $\pi_j<1$ for all businesses $j$, demand-responsiveness implies the business will have non-trivial demand for any choice of prices by the other businesses, since there is a non-zero probability that no other business fits the needs of any given consumer. Importantly, it does not require that $P_j(H)$ be a best-response with respect to any particular objective; only that the price is not so high that it would generate zero demand even if all other businesses left the market.

\begin{pricingrevision}
To ensure continued learning, we require that this demand does not vanish merely because prices approach the highest consumer's willingness to pay while beliefs remain unchanged.
\begin{definition}\label{def:uniform-demand-responsive}
Pricing rules are \emph{uniformly demand-responsive} if there is a constant $\eta\in(0,1)$ such that, at every history $H$, whenever
\[
a_j(H):=\sigma_j(0,\fitprobest_j(H),\hat Q_j(H),\oV)>0,
\]
every price chosen by business $j$ satisfies $0\leq P_j(H)\leq(1-\eta)a_j(H)$.
\end{definition}
This is a sufficient non-starvation condition on the learning dynamics, not a restriction on price deviations in the steady-state equilibrium defined below. With unchanged beliefs and $a_j(H)>0$, it leaves a fixed positive mass of consumer values willing to inspect the business. Demand-responsiveness alone would not suffice: with one business, $c=0$, known $Q=1/2$, uncertain fit, and $V_i$ uniform on $[0,1]$, the prices $P_i=Q(1-2^{-i})$ induce inspection probabilities $2^{-i}$. Only finitely many inspections then occur almost surely, although the zero-price index remains positive.
\end{pricingrevision}

\begin{proposition}\label{prop:market.learning.pricing}
\pricingnew{Fix a market $\mM$ with uniformly demand-responsive pricing rules and business $j$. Almost surely over the history, the posterior distributions $\priorF_{ij},\priorQ_{ij}$ converge weakly. Either the business is learned, with limit beliefs concentrated on $(\pi_j,Q_j)$, or}
\pricingnew{\[
\sigma_j(0,\fitprobest_{ij},\hat Q_{ij},\oV)\leq0
\]}
\pricingnew{for some consumer $i$ and every consumer afterwards, in which case the business is lost even at price zero. The same conclusion holds screen by screen with transcripts.}
\end{proposition}

\pricingnew{Businesses whose parameters are already known may be learned without further observations. An already known business with non-positive zero-price index is payoff-irrelevant and can be omitted from the equilibrium analysis.}

Proposition~\ref{prop:market.learning.pricing} allows us to argue about long-term outcomes \pricingnew{under uniformly demand-responsive pricing rules}.
For now we will only track the learned set $\learnedS(\mM)$.
Later in Section~\ref{sec:pricing.equil} we will impose further assumptions about the way businesses set prices (namely, that they form a market equilibrium given the learned set), which will enable an analysis of long-term consumer utility.

We first show that, all things equal, endogenizing prices weakly improves the learned set. Intuitively, this is because businesses can reduce the probability of becoming lost by lowering prices in response to poor posterior beliefs.

\begin{proposition}\label{prop:market.lost.pricing}
Fix any market $\mM$ with endogenous prices and \pricingnew{uniformly} demand-responsive pricing rules. Let $\tmM$ be a market with (arbitrary) fixed prices, otherwise identical to $\mM$. Then
\pricingnew{\[\learnedS(\mM)\supseteq\learnedS(\tmM),\]}
almost surely over some coupling between the histories of $\mM$ and $\tmM$.
\end{proposition}

As with exogenous prices, we can also compare what happens when consumer search gets cheaper or more informative. We obtain similar conclusions:

\begin{proposition}\label{prop:market.lost.search.pricing}
\pricingnew{Let $\mM,\tmM$ be markets with uniformly demand-responsive pricing rules and otherwise identical primitives, except that $\tmM$ either (a) has lower search cost, or (b) both markets have transcripts as in Definition~\ref{defn:fit-quality-transcripts}, but $\tmM$ has increased search informativeness. Then}
\pricingnew{\[\learnedS(\tmM)\supseteq\learnedS(\mM),\]}
\pricingnew{almost surely over some coupling of histories. The pricing rules may differ between the markets, provided that each profile satisfies Definition~\ref{def:uniform-demand-responsive}.}
\end{proposition}

\pricingnew{The arguments use the same per-business signal streams as the fixed-price comparisons. Uniform demand-responsiveness ensures that a stream stops precisely when the zero-price index becomes non-positive; prices affect the timing, but not this stopping condition. Proofs appear in Section~\ref{sec:pricing.proofs}.}

\subsection{\pricingnew{Steady-State Pricing Equilibrium}}
\label{sec:pricing.equil}

Motivated by the convergence results above, we now define a market equilibrium in which businesses set their prices in best response to each other under limit beliefs $\limP:=(\limP_j:j\in[m])$.

Informally, the market is in equilibrium if (a) the posterior information revealed about all businesses is consistent with a steady-state of the market {as per Proposition~\ref{prop:market.learning.pricing}} (i.e., each business is either learned or lost), and (b) each business's (possibly randomized) choice of price maximizes their profit given the customer search behavior and the choices of all other businesses.

\begin{definition}
A (steady-state) market equilibrium specifies limit beliefs $\limP$ and the tuple of distributions $\Gamma_j$, $j\in[m]$ over prices $P_j$, such that the following holds:
\begin{enumerate}
\item {Under beliefs $\limP$, each business is either learned or lost.}
\item For each business $j$, fixing the {price distributions} of all other businesses besides $j$, each price $P_j$ in the support of $\Gamma_j$ maximizes the expected revenue obtained by business $j$ {from a consumer with beliefs $\limP$ and value drawn from $\mathcal D^V$ who engages in optimal search}.
\end{enumerate}
\end{definition}

\begin{remark}
{For convenience we state the price-optimality condition for a single consumer, but it is equivalent to measure expected per-round revenue from a stream of consumers each with beliefs $\limP$.}
\end{remark}

{While our equilibrium notion formally specifies limit beliefs,} we note that an equilibrium is effectively characterized by the learned set $\learnedS(\mM)$ and the price distributions. Indeed, fixing the learned set essentially determines the associated limit beliefs $\limP$: {if business $j$ is learned then $\limP_j$ is a point mass on $(\pi_j,Q_j)$ and if business $j$ is lost then the precise beliefs $\limP_j$ are not payoff-relevant because lost businesses are never inspected. We will therefore tend to describe an equilibrium as a learned set and choice of price distributions for learned businesses.}

{At equilibrium}, businesses must account for consumer search behavior when setting prices. As it turns out, the pricing game faced by businesses in this market with consumer search is equivalent to an alternative static market where consumers have full information but distorted preferences. The following is a variation of a characterization result of~\citet{choi2018consumer} for equilibrium pricing under consumer search, tailored to our setting.

\begin{definition}\label{def:transformed-market}
Fix any market $\mM$. Define the \emph{$S$-transformed market} as follows, for a given subset $S$ of businesses. In this new market, there's no consumer search (i.e., consumers can transact with businesses, but cannot inspect them). Each consumer $i$ has some value for each product $j$, and knows it.
\pricingnew{Draw one value $V_i\sim\distV$ for consumer $i$, and fit indicators $F_{ij}$ independently across businesses, with means $\pi_j$, independently of $V_i$. The value for product $j$ is zero if $j\notin S$, and otherwise is}
\pricingnew{\[
W_{ij}=\begin{cases}V_iQ_j-c/\pi_j,&F_{ij}=1,\\0,&F_{ij}=0.\end{cases}
\]}
We refer to $W_{ij}$ as the \emph{effective value} of consumer $i$ for business $j$.
\end{definition}

\pricingnew{Conditional on a fit, $W_{ij}$ equals the zero-price search index. Effective values for different businesses share the same $V_i$ and are not generally independent. The consumer purchases a product maximizing strictly positive $W_{ij}-P_j$, using the same tie-breaking rule as in the search market.}

\begin{lemma}\label{lem:equivalent.market}
Fix market $\mM$ and {an arbitrary realization $S$ of the learned set, $\learnedS(\mM)$.} Then price distributions $(\Gamma_j:j\in[m])$ form a market equilibrium {with learned set $S$ (for some choice of limit beliefs)} if and only if $(\Gamma_j:j\in[m])$ form a mixed Nash equilibrium of the pricing game for the {$S$-transformed market.}
\end{lemma}

\pricingnew{The equivalence also preserves expected consumer surplus and total welfare, as shown in Section~\ref{sec:pricing.proofs}. For existence we will use the explicit symmetric construction below; we do not require a separate existence result for arbitrary asymmetric transformed markets.}

\subsection{\pricingnew{Symmetric Equilibrium Characterization}}
\label{sec:strategic.firms.symmetric}

{We will now focus our attention on markets that satisfy a certain symmetry assumption.}

\begin{definition}
A market $\mM$ is \emph{learned-symmetric} if, almost surely over the realization $S$ of learned set $\learnedS(\mM)$, {all businesses in $S$ have the same type:} $Q_j=Q_{j'}$ and $\pi_j=\pi_{j'}$ for all $j,j'\in S$.
\end{definition}

\begin{remark}
\pricingnew{A sufficient condition for learned-symmetry is that the prior assigns probability one to every business having the same type $(\pi,Q)$.} But a market can be learned-symmetric even if $\mP$ is supported on multiple types, as long as all but a single type of business is lost {almost surely}.
\pricingnew{For example, other types can have strictly negative zero-price index at value $\oV$. If such a type were learned, its posterior index would eventually be negative, so it must instead be lost. This allows non-trivial learning with independent priors across businesses.}
\end{remark}

To define this equilibrium, it will be convenient to consider {a monopolist seller facing a consumer with a (random) value $W_{ij}$, as defined in Definition~\ref{def:transformed-market}.} Write $D(p)=\Pr[W_{ij}>p]$ for the probability of sale at a given price $p$. We emphasize that $D(p)\leq\pi$ for all $p\geq0$. Let $R^*=\sup_p pD(p)$ be the optimal revenue obtainable by a monopolist seller.
\pricingnew{When $R^*>0$, let $p^*$ be the smallest price attaining this revenue.}

\begin{theorem}\label{thm:symmetric.equil}
\pricingnew{Fix a learned-symmetric market $\mM$ and a realization $S$ of its learned set. Suppose $n=|S|\geq2$ and $R^*>0$. There is a unique symmetric equilibrium price distribution for the learned businesses.}
In this equilibrium, each business obtains expected revenue $(1-\pi)^{n-1}R^*$ and randomizes prices on range $[\underline p,p^*]$ according to the cumulative distribution function
\[
G(p)=\frac1\pi\left[1-(1-\pi)\left(\frac{p^*D(p^*)}{pD(p)}\right)^{\frac1{n-1}}\right],
\]
where $\underline p$ is the maximum value such that $G(\underline p)=0$.
\pricingnew{Extend $G$ by zero below $\underline p$ and by one above $p^*$. If $n=1$ and $R^*>0$, a monopoly price maximizes $pD(p)$; we select the smallest such price $p^*$ for the canonical equilibrium. If $n=0$, or if $R^*=0$, no trade occurs and all welfare and revenue quantities are zero; choose zero prices for the payoff-irrelevant businesses. Prices of lost businesses are immaterial throughout.}
\end{theorem}

\pricingnew{The proof constructs $G$ by equalizing revenue across its support, then rules out atoms and gaps in any symmetric equilibrium with positive revenue. It also handles monopoly and no-trade outcomes separately.}

\begin{corollary}\label{cor:eqm}
\pricingnew{For each learned-symmetric market $\mM$ and feasible realization $S$ of its learned set, there exists a market equilibrium with learned set $S$.}
\end{corollary}

\pricingnew{Write $\Phi(\mM,S)$ for the symmetric equilibrium given by Theorem~\ref{thm:symmetric.equil}, with the stated selections in the boundary cases.}

\begin{definition}\label{def:canon}
Equilibrium $\CanEqm(\mM):=\Phi(\mM,\learnedS(\mM))$ is called the \emph{canonical equilibrium} for market $\mM$.
\end{definition}

\begin{remark}
While our market dynamics could plausibly converge to $\CanEqm(\mM)$, we emphasize that Proposition~\ref{prop:market.learning.pricing} only establishes convergence of the learned set and not prices. Rather, we posit $\Phi(\mM,S)$ as a long-run solution concept given the learned set $S$ to which the market converges. Note that $\CanEqm(\mM)$ is a random object, with randomness coming from the history (via the learned set).
\end{remark}

\subsection{\pricingnew{Cheaper and More Informative Search}}
\label{sec:pricing.comparisons}

We first prove that reduced search costs improve expected consumer utility at equilibrium, as was the case with fixed prices. Moreover, {we prove that the total welfare} (sum of consumer utilities and business revenues) increases as well. The intuition is that a reduction in search cost increases the effective value of every consumer for every business by an additive amount, which results in increased gains from trade that are shared between the two sides of the market.

Stating this result precisely requires some additional notation. Let $\conW(\mM,\Phi)$ be the expected consumer welfare for market $\mM$ at equilibrium $\Phi$. Next, let $\conW(\mM):=\conW(\mM,\CanEqm(\mM))$ be this quantity evaluated at the canonical equilibrium; call it the \emph{canonical consumer welfare}. Note that this is a random object, randomness coming from the history of $\mM$. This is the quantity that we'll compare across two markets. Similarly, define $\totW(\mM,\Phi)$ for expected \emph{total} welfare at equilibrium $\Phi$, and $\totW(\mM):=\totW(\mM,\CanEqm(\mM))$ for the \emph{canonical total welfare.}
\pricingnew{Here consumer welfare is net of prices and inspection costs. The notation $\conW$ distinguishes equilibrium consumer surplus from the fixed-price limit utility $U(\mM)$ in the main text.}

\begin{theorem}\label{thm:cheaper.search.pricing}
\pricingnew{Suppose markets $\mM,\tmM$ have uniformly demand-responsive pricing rules, all businesses have the same realized type $(\pi,Q)$ in both markets, and all primitives other than the search cost are identical, with $\tilde c<c$. Then}
\[
\conW(\mM)\leq\conW(\tmM),\qquad
\totW(\mM)\leq\totW(\tmM),
\]
almost surely over some coupling between the histories of the two markets.
\end{theorem}

\begin{proof}[Proof Sketch]
\begin{pricingrevision}
Conditional on fit, cheaper search increases the effective value by $\lambda=(c-\tilde c)/\pi$. For a fixed number of learned businesses, the associated quantile revenue curve changes from $R(q)$ to $\tilde R(q)=R(q)+\lambda q$, where $q\in[0,\pi]$ indexes the fit branch. Comparing the equilibrium formulas gives $\tilde G(p+\lambda)\geq G(p)$: equilibrium prices rise by no more than the increase in effective value. Consumer surplus therefore increases. Monopoly revenue also increases, so total welfare increases when the learned set is held fixed. Finally, the learned-set coupling and the symmetric price distribution imply that additional learned businesses can only improve consumer surplus and total welfare. Section~\ref{sec:pricing.proofs} provides the details, including boundary cases.
\end{pricingrevision}
\end{proof}

\begin{figure}[t]
\centering
\begin{tabular}{cc}
\includegraphics[width=0.45\textwidth]{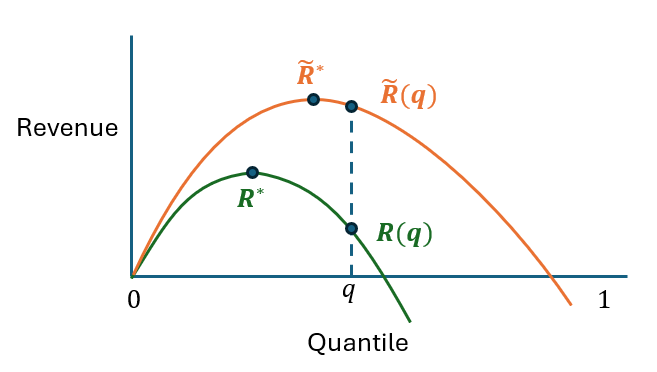}&
\includegraphics[width=0.45\textwidth]{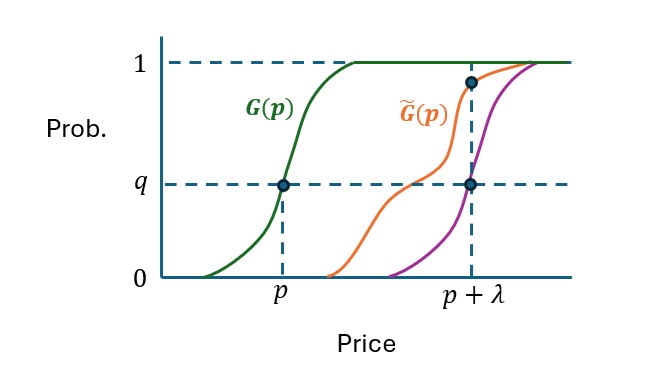}\\
\quad(a)&\qquad(b)
\end{tabular}
\caption{\pricingnew{An illustration of the proof of Theorem~\ref{thm:cheaper.search.pricing}. Left: monopolist quantile revenue curves before (green) and after (orange) the cost reduction. The relevant revenue ratios determine the equilibrium price distribution. Right: equilibrium price CDFs before (green) and after (orange) the reduction. Purple shifts the original CDF to the right by $\lambda=(c-\tilde c)/\pi$. The comparison $\tilde G(p+\lambda)\geq G(p)$ bounds the price increase by the effective-value gain.}}
\label{fig:cheaper.search.pricing}
\end{figure}

Theorem~\ref{thm:cheaper.search.pricing} describes the impact on consumer utility and total welfare, but what is the impact on business revenue? \pricingnew{For a monopolist with $\pi=Q=1/2$ and values uniform on $[1,11/10]$, reducing $c$ from $1/10$ to zero raises the optimal price from $3/10$ to $1/2$ and revenue from $3/20$ to $1/4$.} In general, however, average seller revenue could decrease when search costs are reduced. The reason is that high search costs can cause more businesses to be lost, which reduces competition between sellers and increases the revenue of those that remain. \pricingnew{Holding the number of learned businesses fixed instead gives weakly higher revenue under cheaper search. Section~\ref{sec:informative.search.revenue.impact} states and proves both claims.}

We next examine the impact of making search more informative for consumers. Unlike reducing search costs, and unlike Theorem~\ref{thm:cheaper.search.inform} from the case of exogenously fixed prices, the impact of more informative search on consumer surplus is ambiguous when businesses can price strategically. On one hand, more informative search leads to higher expected consumer value conditional on fit. On the other hand, more selective screening during consumer search means that each business faces less effective competition, which can allow them \pricingnew{to} price more aggressively, extracting more of the consumer's surplus as revenue and distorting the allocation. The following proposition shows that there are instances where the latter effect dominates the former.

\begin{proposition}\label{prop:inform.search.pricing}
\pricingnew{There exist symmetric markets $\mM,\tmM$ with transcripts as in Definition~\ref{defn:fit-quality-transcripts}, identical except that $\tmM$ is more consumer-informative, such that no businesses are lost in either market and}
\[\conW(\tmM)<\conW(\mM).\]
\end{proposition}

\subsection{\pricingnew{Seller Revenue}}
\label{sec:informative.search.revenue.impact}

\pricingnew{Under endogenous pricing, cheaper search can reduce business revenue by changing the number of learned businesses. Holding this number fixed, expected revenue can only increase.}

We need some additional notation. Let $\Rev(\mM,\Phi)$ be the expected total business revenue for market $\mM$ at equilibrium $\Phi$. Focusing on the canonical equilibrium, let
\pricingnew{\[\Rev(\mM):=\Rev(\mM,\CanEqm(\mM))\]}
be the \emph{canonical business revenue}.

\begin{proposition}\label{prop:cheaper.search.pricing.revenue}
\pricingnew{Suppose $\mM,\tmM$ satisfy the assumptions of Theorem~\ref{thm:cheaper.search.pricing}, with $\tilde c<c$.}
\begin{itemize}
\item[(a)] There exist $\mM,\tmM$ such that $\tmM$ achieves lower canonical business revenue, in the sense that
\[\E[\Rev(\tmM)]<\E[\Rev(\mM)].\]
\item[(b)] Let $S,\tilde S$ be any feasible realizations of the learned sets of $\mM,\tmM$, resp., such that $|S|=|\tilde S|$. Let $\Phi=\Phi(\mM,S)$, $\widetilde\Phi=\Phi(\tmM,\tilde S)$ be the corresponding market equilibria. Then
\[\Rev(\tmM,\widetilde\Phi)\geq\Rev(\mM,\Phi).\]
\end{itemize}
\end{proposition}

\begin{proof}
\begin{pricingrevision}
For part (b), write $n=|S|=|\tilde S|$. If $n=0$, both revenues are zero. Otherwise total equilibrium revenue is $n(1-\pi)^{n-1}R^*$, including the monopoly case. Lower search cost increases every effective value conditional on fit, so it weakly increases $R^*$ and therefore total revenue at fixed $n$.
\end{pricingrevision}

We now construct an example where average revenue is lower in $\tmM$ than in $\mM$.
\begin{pricingrevision}
There are $m$ businesses with realized fit probability $\pi=1/2$ and known quality $Q=1/2$. Fit probabilities have independent Beta$(1,1)$ priors. Consumer values are uniform on $[1,11/10]$, so the value distribution satisfies the assumptions of this appendix. Take $c=11/60$ in $\mM$ and $\tilde c=0$ in $\tmM$, with uniformly demand-responsive pricing rules in both markets.

In $\tmM$, every business is learned. The effective value conditional on fit is uniform on $[1/2,11/20]$. Its monopoly price is $1/2$, giving $\tilde R^*=1/4$ and expected total revenue
\[
\E[\Rev(\tmM)]=\frac m4\left(\frac12\right)^{m-1}.
\]
In $\mM$, the zero-price index at value $\oV=11/10$ becomes non-positive precisely when the posterior fit mean is at most $1/3$. After $s$ successes and $f$ failures this mean is $(s+1)/(s+f+2)$, so loss occurs when $f\geq2s+1$. The walk $2s-f$ has independent increments $+2$ and $-1$, each with probability $1/2$. Let $r=(\sqrt5-1)/2$, which satisfies $(r^2+r^{-1})/2=1$. The martingale $r^{2s-f}$, stopped when the walk first hits $-1$ or reaches a high positive level, gives eventual loss probability $r$ on letting that level tend to infinity. Thus the probability of learning is $\rho=1-r\in(0,1/2)$.

The per-business stream characterization in the proof of Proposition~\ref{prop:market.learning.pricing} implies that these learning events are independent, with probability $\rho$ unaffected by $m$. Conditional on learning, the effective value is uniform on $[2/15,11/60]$, with monopoly price $2/15$ and monopoly revenue $R^*=1/15$. Averaging over the learned set therefore gives
\[
\E[\Rev(\mM)]=\frac{m\rho}{15}\left(1-\frac\rho2\right)^{m-1}.
\]
Consequently,
\[
\frac{\E[\Rev(\mM)]}{\E[\Rev(\tmM)]}
=\frac{4\rho}{15}(2-\rho)^{m-1}.
\]
Since $\rho>0$ and $2-\rho>1$, this ratio exceeds one for sufficiently large $m$, proving part (a).
\end{pricingrevision}
\end{proof}

\subsection{\pricingnew{Supporting Proofs}}
\label{sec:pricing.proofs}

\subsubsection{Proof of Proposition~\ref{prop:market.learning.pricing}}
\begin{proof}
The proof closely follows the proof of Proposition~\ref{prop:market.learning}.
\begin{pricingrevision}
Generate the history from independent per-business signal streams, as in Appendix~\ref{app:connection} and the proof of Proposition~\ref{prop:market.learning}. An inspection reveals a fit outcome and, if it passes, generates a transaction and quality feedback. With transcripts, the combined feedback reveals successive screens up to the first failure. The resulting observations depend on the business's own stream, not on the round in which it is inspected. Prices reveal no additional information, so Lemma~\ref{lem:indep.signals} applies. Infinite inspection gives infinitely many observations of every relevant parameter, since all screen probabilities are strictly positive. The corresponding posterior distributions converge weakly to their true parameters.

It remains to consider a business inspected only finitely often. After its last inspection its beliefs are fixed. Suppose its zero-price index at $\oV$ were a fixed number $a>0$. Uniform demand-responsiveness implies that its index at the posted price is at least $\eta a$ at value $\oV$, and at least $\eta a/2$ for a fixed interval of values just below $\oV$. This interval has strictly positive probability under $\distV$. Independently, all other businesses fail to fit the consumer with probability at least $\prod_{h\ne j}(1-\pi_h)>0$. On these events the consumer reaches and inspects business $j$, regardless of the other prices and beliefs. Hence the conditional inspection probability has a positive lower bound while its beliefs remain fixed, contradicting a last inspection. Its zero-price index must instead be non-positive, and non-negative prices then preclude every subsequent inspection. The posterior distributions remain constant, establishing their convergence as well.

This also identifies the learned set directly from the independent streams: a business's observations advance until the first prefix, including the empty prefix, whose zero-price index at $\oV$ is non-positive, or continue forever if no such prefix exists. This is the same stopping rule as in the otherwise identical market with all prices fixed at zero. Parameters known a priori remain known regardless of the stopping rule. In particular, the learned-set distribution is unaffected by the timing of inspections under different uniformly demand-responsive pricing rules.
\end{pricingrevision}
\end{proof}

\subsubsection{Proof of Proposition~\ref{prop:market.lost.pricing}}
\begin{proof}
We will couple histories $H^{\mM}$ and $H^{\tmM}$ in the same manner as in the proof of Theorem~\ref{thm:cheaper.search.general}, by independently drawing the sequence of consumer values and the sequences of fit and quality realizations for each business, $\Omega_j$.
\pricingnew{For transcript markets, use per-screen streams instead. After the same number of inspections of a business, the two markets have observed the same prefix of its feedback stream whenever both prefixes are reached. Include independent pricing randomization when generating the endogenous-price history.}

\pricingnew{If a business is learned in the fixed-price market but not in the endogenous-price market, the preceding proof supplies a finite prefix at which its zero-price index is non-positive. The fixed-price market reaches that prefix because it learns the business. Let $\fitprobest_j,\hat Q_j$ be the common posterior means there, and let $P_j\geq0$ be its fixed price in $\tmM$. Then}
\pricingnew{\[
\sigma_j(P_j,\fitprobest_j,\hat Q_j,\oV)
\leq\sigma_j(0,\fitprobest_j,\hat Q_j,\oV)\leq0.
\]}
\pricingnew{This stops its observations in the fixed-price market as well, a contradiction. Businesses known a priori are learned in both markets. Thus $\learnedS(\mM)\supseteq\learnedS(\tmM)$ under the coupling.}
\end{proof}

\subsubsection{\pricingnew{Proof of Proposition~\ref{prop:market.lost.search.pricing}}}
\begin{proof}
\begin{pricingrevision}
By the stream characterization above, each market has the same learned set as its zero-price counterpart under the natural stream coupling. Apply Theorem~\ref{thm:cheaper.search.general} to these zero-price markets for a cost reduction, or Theorem~\ref{thm:cheaper.search.inform} for an increase in informativeness with transcripts. This gives the claimed inclusion, including when pricing rules differ across the two markets.
\end{pricingrevision}
\end{proof}

\subsubsection{Proof of Lemma~\ref{lem:equivalent.market}}
\begin{proof}
This result follows from analysis of~\cite{choi2018consumer} specialized to our setting. We include a proof here for completeness.

\begin{pricingrevision}
Fix a consumer value $V$ and a price profile. On the learned set, define $\sigma_j=VQ_j-P_j-c/\pi_j$. The optimal search policy inspects businesses in decreasing order of positive $\sigma_j$ and transacts with the first fit. Conditional on the full vector of fit outcomes, this is precisely the product with greatest strictly positive $W_{ij}-P_j$ in the transformed market. The shared tie-breaking rule gives the same allocation at positive ties; neither market trades at the outside-option tie. Lost businesses are never inspected at any non-negative price and have zero effective value, so they cannot profitably deviate.
\end{pricingrevision}

Applying this argument to each pair of businesses separately, we conclude that expected seller payoffs from any profile of prices in the market with consumer search is equal to the expected seller payoffs from the same profile of prices in the equivalent market without search. The set of mixed equilibria for the two markets must therefore coincide.

\begin{pricingrevision}
For consumer surplus, let $\mu_j$ be the probability that the consumer reaches business $j$ in this search order. Its contribution to expected utility is
\[
\mu_j\bigl[\pi_j(VQ_j-P_j)-c\bigr]=\mu_j\pi_j\sigma_j.
\]
Summing over the searched businesses gives $\E_F[\max\{0,\max_{j\in S}(W_{ij}-P_j)\}]$, the transformed-market surplus. Expected revenues agree by the allocation argument, so total welfare agrees as well. Averaging over values and prices proves the asserted equivalence for all three quantities.
\end{pricingrevision}
\end{proof}

\subsubsection{Proof of Theorem~\ref{thm:symmetric.equil}}
\begin{proof}
\begin{pricingrevision}
First suppose $n\geq2$ and $R^*>0$. For a realized consumer value $V$, write $a(V)=VQ-c/\pi$. All fitting businesses then offer the same gross effective value $a(V)$; conditional on a consumer being willing to buy at price $p$, every lower-priced fitting business is preferred. If competitors use an atomless CDF $G$, a business posting $p$ therefore earns
\end{pricingrevision}
\[
(1-\pi G(p))^{n-1}pD(p).
\]
Our equilibrium condition therefore requires that this quantity be equal to the equilibrium revenue of $(1-\pi)^{n-1}p^*D(p^*)$. Solving for $G(p)$ yields
\[
G(p)=\frac1\pi\left[1-(1-\pi)\left(\frac{p^*D(p^*)}{pD(p)}\right)^{\frac1{n-1}}\right].
\]
The floor price $\underline p$ is then defined to be the maximum value such that $G(\underline p)=0$.

Each business is indifferent between prices in $[\underline p,p^*]$, by construction. Regularity of the value distribution further implies that any price lower than $\underline p$ leads to only lower revenue, as the monopolist revenue decreases and the business is guaranteed to have the lowest price at any price below $\underline p$. Furthermore, any price $p>p^*$ leads to a revenue of $(1-\pi)^{n-1}pD(p)$, which is at most $(1-\pi)^{n-1}p^*D(p^*)$ by optimality of $p^*$. We conclude that no seller has an improving deviation, and hence price distribution $G$ is indeed an equilibrium.

\begin{pricingrevision}
For completeness, regularity makes $q\mapsto q[QF^{-1}(1-q/\pi)-c/\pi]$ concave on $[0,\pi]$. Thus $pD(p)$ is strictly increasing below the smallest monopoly price $p^*$ and is non-increasing above it. This establishes that the displayed formula is a continuous CDF on its stated support.

To prove uniqueness, a seller can guarantee $(1-\pi)^{n-1}R^*>0$ by posting $p^*$, regardless of competing prices, whenever all other fits fail. Hence zero cannot be in the support of an equilibrium price distribution. A positive atom would permit a small undercut: there is positive probability of meeting another seller at that atom and of the consumer fitting both, so some seller would obtain a discrete increase in demand at arbitrarily small price cost. Symmetry therefore rules out atoms. Let $H$ be the resulting continuous CDF and $b$ the top of its support. Equilibrium revenue is $(1-\pi)^{n-1}bD(b)$, by taking a limit along support prices approaching $b$. Posting $p^*$ instead gives $(1-\pi H(p^*))^{n-1}R^*$. Since $bD(b)\leq R^*$ and $H(p^*)\leq1$, the best-response condition forces equality in both inequalities. Thus $H(p^*)=1$ and $bD(b)=R^*$. As $p^*$ is the smallest monopoly maximizer, these equalities imply $b=p^*$.

There can be no gap in the support below $p^*$: within a gap the competitive factor is unchanged, while raising the price strictly increases monopoly revenue. At the top of the support $p_2$, equilibrium revenue is $(1-\pi)^{n-1}p_2D(p_2)$. If $p_2<p^*$, raising the price to $p^*$ increases revenue without changing that competitive factor. Thus $p_2=p^*$; indifference then forces the stated formula throughout the support and determines its lower endpoint. This proves uniqueness for the learned businesses. The choice of the smallest monopoly price also covers a flat interval of monopoly maximizers: the CDF reaches one at its lower endpoint and has no mass in the rest of that interval.

If $n=1$, the problem is monopoly pricing, so the stated selection of $p^*$ is optimal, without any claim that a monopoly maximizer must be unique. If $n=0$ or $R^*=0$, no consumer can obtain strictly positive effective surplus at a non-negative price, and no trade occurs. These cases give the remaining equilibrium selections and prove Corollary~\ref{cor:eqm}.
\end{pricingrevision}
\end{proof}

\subsubsection{Proof of Theorem~\ref{thm:cheaper.search.pricing}}
\begin{proof}
\begin{pricingrevision}
We first hold fixed a nonempty learned set of size $n$. Couple consumer values and fits across markets. Conditional on fit, the effective value increases by
\[
\lambda=\frac{c-\tilde c}{\pi}>0.
\]
For $q\in[0,\pi]$, let $p(q)=QF^{-1}(1-q/\pi)-c/\pi$ be the inverse value on the fit branch, and set $R(q)=qp(q)$. The corresponding functions in $\tmM$ satisfy
\[
\tilde p(q)=p(q)+\lambda,\qquad
\tilde R(q)=R(q)+\lambda q.
\]
These are fit-branch quantiles, not quantiles of the atom at zero from a failed fit. Only positive prices enter the equilibrium formula.

If $R^*=0$, the original market has zero surplus and welfare, so both comparisons hold immediately. Otherwise let $q^*$ and $\tilde q^*$ be the largest monopoly-maximizing quantiles. The two optimality inequalities imply $\tilde q^*\geq q^*$, and hence $\tilde p^*\leq p^*+\lambda$. In particular, for $n=1$ the canonical monopoly price rises by at most $\lambda$, which proves the consumer-surplus comparison. Monopoly revenue also rises, giving the welfare comparison.

For $n\geq2$, we claim that $\tilde G(p+\lambda)\geq G(p)$ for every price $p$. Below the old support the claim is immediate, as it is whenever $p+\lambda\geq\tilde p^*$. In the remaining case let $q$ satisfy $p(q)=p$; then $q>\tilde q^*$ and $R(q)>0$. Since $p(q)$ decreases with $q$,
\[
\frac{\tilde R^*}{R^*}
=\frac{R(\tilde q^*)+\lambda\tilde q^*}{R^*}
\leq1+\frac{\lambda}{p(\tilde q^*)}
\leq1+\frac{\lambda}{p(q)}
=\frac{\tilde R(q)}{R(q)}.
\]
Thus $\tilde R^*/\tilde R(q)\leq R^*/R(q)$. Substitution in the equilibrium CDF formula, with truncation to $[0,1]$ below each support, proves the claim. It permits a coupling of independent seller prices with $\tilde P_j\leq P_j+\lambda$. Every fitting product's net effective value is then weakly higher, so transformed-market consumer surplus, and hence actual expected consumer surplus by Lemma~\ref{lem:equivalent.market}, weakly increases. Total seller revenue at fixed $n$ is $n(1-\pi)^{n-1}R^*$ and also increases. Their sum gives the total-welfare comparison at fixed $n$.

It remains to allow the learned set to expand. Proposition~\ref{prop:market.lost.search.pricing} gives $\learnedS(\tmM)\supseteq\learnedS(\mM)$ under a coupling. At fixed cost, the equilibrium price formula is weakly decreasing in the number of learned sellers in the sense of first-order stochastic dominance, including entry from monopoly. Additional learned sellers therefore improve consumer surplus. For total welfare, fix a value $V$ and prices, and let $K$ be the number of sellers with positive index $a(V)-P_j$, where $a(V)=VQ-c/\pi$. Expected total welfare for this consumer is
\[
a(V)\bigl[1-(1-\pi)^K\bigr].
\]
Indeed, each inspection costs $c$ and succeeds with probability $\pi$, so expected total inspection cost is $c/\pi$ times the transaction probability. When $K>0$, $a(V)>0$. Under lower prices and entry, $K$ weakly increases, so this welfare expression does as well. Starting from no learned seller also gives zero welfare. Combining the cost reduction at the old learned-set size with entry at the new cost proves both comparisons on the coupled histories.
\end{pricingrevision}
\end{proof}

\subsubsection{Proof of Proposition~\ref{prop:inform.search.pricing}}
\begin{proof}
We construct an explicit example.
\begin{pricingrevision}
There are two businesses, with three independent screens having known probabilities $(1-\epsilon,1/2,3/4)$ for each business, where $\epsilon>0$ is small. Values are uniform on $[1,11/10]$ and $c=0$. Market $\mM$ has inspection depth one, so $\pi=1-\epsilon$ and $Q=3/8$. Market $\tmM$ has inspection depth two, so $\tilde\pi=(1-\epsilon)/2$ and $\tilde Q=3/4$. Thus $\tmM$ is more consumer-informative, with all screen probabilities strictly interior, and both markets reveal transcripts. All parameters are known, so every business is learned from the outset; uniformly demand-responsive rules can also generate infinite inspection of each business.

For either market the effective value conditional on fit is uniform on $[Q,11Q/10]$. The monopoly price is $Q$ and $R^*=\pi Q$. Equilibrium prices are at most $Q$, so almost every consumer transacts if at least one business fits. With two businesses, total welfare is $Q\E[V](2\pi-\pi^2)$ and total revenue is $2\pi(1-\pi)Q$. Their difference is
\[
\conW=Q\left[\pi^2+\frac1{20}(2\pi-\pi^2)\right].
\]
As $\epsilon\downarrow0$, this expression tends to $63/160$ in $\mM$ and $69/320$ in $\tmM$. The former is strictly larger. The displayed expressions are continuous in $\epsilon$, so the strict comparison holds for all sufficiently small positive $\epsilon$, without deterministic screen outcomes or point-mass consumer values.
\end{pricingrevision}
\end{proof}

\end{document}